\pdfoutput=1 

\documentclass[sigconf, nonacm]{acmart}

\usepackage{xifthen}
\usepackage[nolist,nohyperlinks]{acronym}
\usepackage[ruled,vlined]{algorithm2e}
\usepackage{amsthm}
\usepackage[capitalise,noabbrev]{cleveref}
\usepackage[binary-units=true]{siunitx}

\newtheorem{lemma}{Lemma}

\SetKw{Break}{break}

\newif\ifextended\extendedtrue 
\newcommand{\myAlg}[2][]{
    \ifthenelse{\isempty{#1}}%
    {\begin{figure}}
    {\begin{figure}[#1]}
    \begingroup 
    \csname @twocolumnfalse\endcsname
    \noindent
    \resizebox{\columnwidth}{!}{%
        \begin{minipage}{1.33\columnwidth}
            \begin{algorithm}[H]
                \DontPrintSemicolon
                {#2}
            \end{algorithm}
        \end{minipage}%
    }
    \endgroup
    \end{figure}
}

\newcommand*{\symDefine}[2]{\newcommand{{#1}}{{#2}}}

\symDefine{\symExponentialRate}{a}
\symDefine{\symBase}{b}
\symDefine{\symCountVariate}{C}
\symDefine{\symCountVariateEstimate}{\hat{C}}
\symDefine{\symFourierCoefficient}{c}
\symDefine{\symDiffCountVariate}{D}
\symDefine{\symElement}{d}
\symDefine{\symSetA}{U}
\symDefine{\symSetB}{V}
\symDefine{\symSetC}{W}
\symDefine{\symUpdateCounter}{w}
\symDefine{\symAnyFunc}{f}
\symDefine{\symAnyFuncTwo}{g}
\symDefine{\symDistributionFunc}{F}
\symDefine{\symHashFunc}{h}
\symDefine{\symHarmonic}{H}
\symDefine{\symImaginary}{\mathrm{i}}
\symDefine{\symIndexI}{i}
\symDefine{\symIndexK}{k}
\symDefine{\symIndexJ}{j}
\symDefine{\symJaccard}{J}
\symDefine{\symJaccardEstimate}{\hat{\symJaccard}}
\symDefine{\symRegVal}{k}
\symDefine{\symQuantity}{g}
\symDefine{\symQuantityEstimate}{\hat{\symQuantity}}
\symDefine{\symRegValVariate}{K}
\symDefine{\symRegValVariateLow}{\symRegValVariate_\textnormal{low}}
\symDefine{\symIndexL}{l}
\symDefine{\symFisher}{I}
\symDefine{\symLikelihood}{\mathcal{L}}
\symDefine{\symNumReg}{m}
\symDefine{\symCardinality}{n}
\symDefine{\symInputSize}{n}
\symDefine{\symCardinalityEstimate}{\hat{\symCardinality}}
\symDefine{\symCardinalityCorrectedEstimate}{\symCardinalityEstimate_\textnormal{corr}}
\symDefine{\symBigO}{\mathcal{O}}
\symDefine{\symProbFunc}{p}
\symDefine{\symProbability}{P}
\symDefine{\symMaxRegularValue}{q}
\symDefine{\symRNG}{R}
\symDefine{\symIntPower}{s}
\symDefine{\symHyperMinHashParameter}{r}
\symDefine{\symSetS}{S}
\symDefine{\symStatisticX}{X}
\symDefine{\symPoint}{x}
\symDefine{\symX}{x}
\symDefine{\symY}{y}
\symDefine{\symZ}{z}
\symDefine{\symCardinalityA}{\symCardinality_\symSetA}
\symDefine{\symCardinalityB}{\symCardinality_\symSetB}
\symDefine{\symCardinalityANorm}{u}
\symDefine{\symCardinalityBNorm}{v}
\symDefine{\symCardinalityANormEstimate}{\hat{\symCardinalityANorm}}
\symDefine{\symCardinalityBNormEstimate}{\hat{\symCardinalityBNorm}}
\symDefine{\symCardinalityUnion}{\symCardinality_{\symSetA\cup\symSetB}}
\symDefine{\symCardinalityAEstimate}{\symCardinalityEstimate_\symSetA}
\symDefine{\symCardinalityBEstimate}{\symCardinalityEstimate_\symSetB}
\symDefine{\symCardinalityUnionEstimate}{\symCardinalityEstimate_{\symSetA\cup\symSetB}}

\symDefine{\symSmallProbability}{\varepsilon}
\symDefine{\symGamma}{\gamma}
\symDefine{\symEta}{\eta}
\symDefine{\symPowerSeriesFunc}{\xi}
\symDefine{\symHelperFunc}{\zeta}
\symDefine{\symSmallCorrectionFunc}{\sigma}
\symDefine{\symLargeCorrectionFunc}{\tau}
\symDefine{\symError}{\varepsilon}
\symDefine{\symPoissonRate}{\lambda}
\symDefine{\symDensity}{\rho}
\symDefine{\symBinaryOperation}{\varphi}

\DeclareMathOperator*{\symExponential}{Exp}
\DeclareMathOperator*{\symUniform}{Uniform}
\DeclareMathOperator*{\symPoisson}{Poisson}
\DeclareMathOperator*{\symVariance}{Var}
\DeclareMathOperator*{\symCovariance}{Cov}
\DeclareMathOperator*{\symExpectation}{\mathbb{E}}
\DeclareMathOperator*{\symRMSE}{RMSE}

\begin{acronym}
\acro{HLL}{HyperLogLog}
\acro{GHLL}{generalized HyperLogLog}
\acro{MH}[MinHash]{minwise hashing}
\acro{ML}{maximum likelihood}
\acro{LSH}{locality-sensitive hashing}
\acro{RSD}{relative standard deviation}
\acro{RMSE}{root-mean-square error}
\end{acronym}

\allowdisplaybreaks
\DontPrintSemicolon

\begin{document}
\title{SetSketch: Filling the Gap between MinHash and HyperLogLog} 

\author{Otmar Ertl}
\affiliation{%
  \institution{Dynatrace Research}
  \city{Linz}
  \state{Austria}
  \postcode{4040}
}
\email{otmar.ertl@dynatrace.com}

\begin{abstract}
MinHash and HyperLogLog are sketching algorithms that have become indispensable for set summaries in big data applications. While HyperLogLog allows counting different elements with very little space, MinHash is suitable for the fast comparison of sets as it allows estimating the Jaccard similarity and other joint quantities. This work presents a new data structure called SetSketch that is able to continuously fill the gap between both use cases. Its commutative and idempotent insert operation and its mergeable state make it suitable for distributed environments. Fast, robust, and easy-to-implement estimators for cardinality and joint quantities, as well as the ability to use SetSketch for similarity search, enable versatile applications. The presented joint estimator can also be applied to other data structures such as MinHash, HyperLogLog, or HyperMinHash, where it even performs better than the corresponding state-of-the-art estimators in many cases.
\end{abstract}

\maketitle

\pagestyle{plain}

\section{Introduction}
\label{sec:intro}
Data sketches \cite{Cormode2017} that are able to represent sets of arbitrary size using only a small, fixed amount of memory have become important and widely used tools in big data applications. Although individual elements can no longer be accessed after insertion, they are still able to give approximate answers when querying the cardinality or joint quantities which may include the intersection size, union size, size of set differences, inclusion coefficients, or similarity measures like the Jaccard and the cosine similarity.

Numerous such algorithms with different characteristics have been developed and published over the last two decades. 
They can be classified based on following properties \cite{Pettie2020} and use cases:

\noindent\textbf{Idempotency:} Insertions of values that have already been added should not further change the state of the data structure. Even though this seems to be an obvious requirement, it is not satisfied by many sketches for sets \cite{Chen2011, Mitzenmacher2014, Qi2020, Wang2019, Xiao2020}.

\noindent\textbf{Commutativity:}
The order of insert operations should not be relevant for the final state. If the processing order cannot be guaranteed, commutativity is needed to get reproducible results. Many data structures are not commutative \cite{Helmi2012, Cohen2015, Ting2014}.

\noindent\textbf{Mergeability:}
In large-scale applications with distributed data sources it is essential that the data structure supports the union set operation. It allows combining data sketches resulting from partial data streams to get an overall result.  Ideally, the union operation is idempotent, associative, and commutative to get reproducible results. Some data structures trade mergeability for better space efficiency \cite{Cohen2015, Ting2014, Chen2011}.

\noindent\textbf{Space efficiency:}
A small memory footprint is the key purpose of data sketches. They generally allow to trade space for estimation accuracy, since variance is typically inversely proportional to the sketch size. Better space efficiencies can sometimes also be achieved at the expense of recording speed, e.g. by compression \cite{Durand2004,Scheuermann2007,Lang2017}.

\noindent\textbf{Recording speed:}
Fast update times are crucial for many applications. They vary over many orders of magnitude for different data structures. For example, they may be constant like for the \ac{HLL} sketch \cite{Flajolet2007} or proportional to the sketch size like for \ac{MH} \cite{Broder1997}.

\noindent\textbf{Cardinality estimation:}
An important use case is the estimation of the number of elements in a set. Sketches have different efficiencies in encoding cardinality information. Apart from estimation error, robustness, speed, and simplicity of the estimation algorithm are important for practical use. Many algorithms rely on empirical observations \cite{Heule2013} or need to combine different estimators for different cardinality ranges, as is the case for the original estimators of \ac{HLL} \cite{Flajolet2007} and HyperMinHash \cite{Yu2020}.

\noindent\textbf{Joint estimation:}
Knowing the relationship among sets is important for many applications. Estimating the overlap of both sets in addition to their cardinalities eventually allows the computation of joint quantities such as Jaccard similarity, cosine similarity, inclusion coefficients, intersection size, or difference sizes. Data structures store different amounts of joint information. In that regard, for example, \ac{MH} contains more information than \ac{HLL}. Extracting joint quantities is more complex compared to cardinalities, because it involves the state of two sketches and requires the estimation of three unknowns in the general case, which makes finding an efficient, robust, fast, and easy-to-implement estimation algorithm challenging. If a sketch supports cardinality estimation and is mergeable, the joint estimation problem can be solved using the inclusion-exclusion principle $|\symSetA\cap\symSetB| = |\symSetA| + |\symSetB| - |\symSetA\cup\symSetB|$. However, this naive approach is inefficient as it does not use all information available.

\noindent\textbf{Locality sensitivity:}
A further use case is similarity search. If the same components of two different sketches are equal with a probability that is a monotonic function of some similarity measure, it can be used for \ac{LSH} \cite{Indyk1998, Bawa2005, Lv2007,Zhu2016}, an indexing technique that allows querying similar sets in sublinear time. For example, the components of \ac{MH} have a collision probability that is equal to the Jaccard similarity.

The optimal choice of a sketch for sets depends on the application and is basically a trade-off between estimation accuracy, speed, and memory efficiency. Among the most widely used sketches are \ac{MH} \cite{Broder1997} and \ac{HLL} \cite{Flajolet2007}. Both are idempotent, commutative, and mergeable, which is probably one of the reasons for their popularity. The fast recording speed, simplicity, and memory-efficiency have made \ac{HLL} the state-of-the-art algorithm for cardinality estimation. In contrast, \ac{MH} is significantly slower and is less memory-efficient with respect to cardinality estimation, but is locality-sensitive and allows easy and more accurate estimation of joint quantities.

\subsection{Motivation and Contributions}
The motivation to find a data structure that supports more accurate joint estimation and locality sensitivity like \ac{MH}, while having a better space-efficiency and a faster recording speed that are both closer to those of \ac{HLL}, has led us to a new data structure called SetSketch. It fills the gap between \ac{HLL} and \ac{MH} and can be seen as a generalization of both as it is configurable with a continuous parameter between those two special cases. The possibility to fine-tune between space-efficiency and joint estimation accuracy allows better adaptation to given requirements.

We present fast, robust, and simple methods for cardinality and joint estimation that do not rely on any empirically determined constants and that also give consistent errors over the full cardinality range. The cardinality estimator matches that of both \ac{MH} or \ac{HLL} in the corresponding limit cases. Together with the derived corrections for small and large cardinalities, which do not require any empirical calibration, the estimator can also be applied to \ac{HLL} and \ac{GHLL} sketches over the entire cardinality range. Since its derivation is much simpler than that in the original paper based on complex analysis \cite{Flajolet2007}, this work also contributes to a better understanding of the \ac{HLL} algorithm.

The presented joint estimation method can be also specialized and used for other data structures. In particular, we derived a new closed-form estimator for \ac{MH} that dominates the state-of-the-art estimator based on the number of identical components. Furthermore, our approach is able to improve joint estimation from \ac{HLL}, \ac{GHLL}, and HyperMinHash sketches except for very small sets compared to the corresponding state-of-the-art estimators. Given the popularity of all those sketches, specifically of \ac{MH} and \ac{HLL}, these side results could therefore have a major impact in practice, as more accurate estimates can be made from existing and persisted data structures. We also found that a performance optimization actually developed for SetSketch can be applied to \ac{HLL} to speed up recording of large sets.

The presented methods have been empirically verified by extensive simulations. For the sake of reproducibility the corresponding source code including a reference implementation of SetSketch is available at \url{https://github.com/dynatrace-research/set-sketch-paper}. An extended version of this paper with additional appendices containing mathematical proofs and more experimental results is also available \cite{Ertl2021}.

\subsection{MinHash}
\Ac{MH} \cite{Broder1997} maps a set $\symSetS$ to an $\symNumReg$-dimensional vector $(\symRegValVariate_1, \ldots,\symRegValVariate_\symNumReg)$ using
\begin{equation*}
\symRegValVariate_\symIndexI := \min_{\symElement\in\symSetS} \symHashFunc_\symIndexI(\symElement).
\end{equation*}
Here $\symHashFunc_\symIndexI$ are independent hash functions. The probability, that components $\symRegValVariate_{\symSetA\symIndexI}$ and $\symRegValVariate_{\symSetB\symIndexI}$ of two different \ac{MH} sketches for sets $\symSetA$ and $\symSetB$ match, equals the Jaccard similarity $\symJaccard$
\begin{equation*}
\symProbability(\symRegValVariate_{\symSetA\symIndexI} = \symRegValVariate_{\symSetB\symIndexI}) = {\textstyle\frac{|\symSetA \cap \symSetB|}{|\symSetA \cup \symSetB|}} = \symJaccard.
\end{equation*}
This locality sensitivity makes \ac{MH} very suitable for set comparisons and similarity search, because the Jaccard similarity can be directly and quickly estimated from the fraction of equal components with a \ac{RMSE} of $\sqrt{\symJaccard(1-\symJaccard)/\symNumReg}$. \ac{MH} also allows the estimation of cardinalities \cite{Clifford2012, Cohen2015} and other joint quantities \cite{Dasu2002, Cohen2017}.
 
Meanwhile, many improvements and variants have been published that either improve the recording speed or the memory efficiency. One permutation hashing \cite{Li2012} reduces the costs for adding a new element from $\symBigO(\symNumReg)$ to $\symBigO(1)$. However, there is a high probability of uninitialized components for small sets leading to large estimation errors. This can be remedied by applying a finalization step called densification \cite{Shrivastava2014, Shrivastava2017, Mai2019} which may be expensive for small sets \cite{Ertl2020} and also prevents further aggregations. Alternatively, fast similarity sketching \cite{Dahlgaard2017}, SuperMinHash \cite{Ertl2017b}, or weighted minwise hashing algorithms like BagMinHash \cite{Ertl2018} and ProbMinHash \cite{Ertl2020} specialized to unweighted sets  can be used instead. Compared to one permutation hashing with densification, they are mergeable, allow further element insertions, and even give more accurate Jaccard similarity estimates for small set sizes.

To shrink the memory footprint, b-bit minwise hashing \cite{Li2010} can be used to reduce all \ac{MH} components from typically 32 or 64 bits to only a few bits in a finalization step. Although this loss of information must be compensated by increasing the number of components $\symNumReg$, the memory efficiency can be significantly improved if precise estimates are needed only for high similarities. However, the need of more components increases the computation time and the sketch cannot be further aggregated or merged after finalization.

Besides the original application of finding similar websites \cite{Henzinger2006}, \ac{MH} is nowadays widely used for nearest neighbor search \cite{Indyk1998}, association-rule mining \cite{Cohen2001}, machine learning \cite{Li2011}, metagenomics \cite{Ondov2016, Berlin2015,Marcais2019, Elworth2020}, molecular fingerprinting \cite{Probst2018}, graph embeddings \cite{Beres2019}, or malware detection \cite{Raff2017, Nissim2019}.

\subsection{HyperLogLog}
\label{sec:intro_hyperloglog}
The \acf{HLL} data structure consists of $\symNumReg$ integer-valued registers $\symRegValVariate_1,\symRegValVariate_2,\ldots,\symRegValVariate_\symNumReg$ similar to \ac{MH}. While \ac{MH} typically uses at least 32 bits per component, \ac{HLL} only needs 5-bit or 6-bit registers to count up to billions of distinct elements \cite{Flajolet2007, Heule2013}.
The state for a given set $\symSetS$ is defined by
\begin{equation}
\label{equ:hll_update}
\textstyle\symRegValVariate_\symIndexI := \max_{\symElement\in\symSetS} \lfloor 1 - \log_\symBase \symHashFunc_\symIndexI(\symElement)\rfloor\quad\textnormal{with}\ \symHashFunc_\symIndexI(\symElement) \sim \symUniform(0,1)
\end{equation}
where $\symHashFunc_\symIndexI$ are $\symNumReg$ independent hash functions. $\symRegValVariate_\symIndexI=0$ is used for initialization and the representation of empty sets. The original \ac{HLL} uses the base $\symBase=2$, which makes the logarithm evaluation very cheap. We refer to the \acf{GHLL} when using any $\symBase>1$ \cite{Clifford2012, Pettie2020}. Definition \eqref{equ:hll_update} leads to a recording time of $\symBigO(\symNumReg)$. Therefore, to have a constant time complexity, \ac{HLL} implementations usually use stochastic averaging \cite{Flajolet2007} which distributes the incoming data stream over all $\symNumReg$ registers 
\begin{equation*}
\textstyle\symRegValVariate_\symIndexI := \max_{\symElement\in\symSetS:\symHashFunc_1(\symElement)=\symIndexI} \lfloor 1 - \log_\symBase \symHashFunc_2(\symElement)\rfloor,
\end{equation*}
where $\symHashFunc_1$ and $\symHashFunc_2$ are independent uniform hash functions mapping to $\lbrace 1, 2,\ldots,\symNumReg\rbrace$ and the interval $(0, 1)$, respectively. Instead of using two hash functions, implementations typically calculate a single hash value that is split into two parts. 
The original cardinality estimator \cite{Flajolet2007} was inaccurate for small and large cardinalities. 
Therefore, a series of improvements have been proposed \cite{Heule2013, YunxiangZhao2016, Qin2016}, which finally ended in an estimator that does not require empirical calibration \cite{Ertl2017, Ertl2017a} and that is already successfully used by the Redis in-memory data store.

\ac{HLL} is not very suitable for joint estimation or similarity search. The reason is that there is no simple estimator for the Jaccard similarity and no closed-form expression for the collision probability like for \ac{MH}. The inclusion-exclusion principle gives significantly worse estimates for joint quantities than \ac{MH} using the same memory footprint \cite{Dasgupta2016}. A \ac{ML} based method \cite{Ertl2017, Ertl2017a} for joint estimation was proposed, that performs significantly better, but requires solving a three-dimensional optimization problem, which does not easily translate into a fast and robust algorithm. 
A computationally less expensive method was proposed in \cite{Nazi2018} for estimating inclusion coefficients. However, it relies on precomputed lookup tables and does not properly account for stochastic averaging, which can lead to large errors for small sets.

These joint estimation approaches have proven that \ac{HLL} also encodes some joint information. This might be one reason why \ac{HLL}'s asymptotic space efficiency, when considering cardinality estimation only, is not optimal \cite{Pettie2020,Pettie2020a}. 
Lossless compression \cite{Durand2004, Scheuermann2007, Lang2017} is able to improve the memory efficiency at the expense of recording speed and is therefore rarely used in practice. Approaches using lossy compression \cite{Xiao2020} are not recommended as their insert operations are not idempotent and commutative.

Nevertheless, due to its simplicity, \ac{HLL} has become the state-of-the-art cardinality estimation algorithm, especially if the data is distributed. This is evidenced by the many databases like Redis, Oracle, Snowflake, Microsoft SQL Server, Google BigQuery, Vertica, Elasticsearch, Aerospike, or Amazon Redshift that use \ac{HLL} under the hood to realize approximate distinct count queries. \ac{HLL} is also used for many applications like metagenomics \cite{Baker2019, Marcais2019, Elworth2020}, graph analysis \cite{Boldi2011, Priest2018, Priest2020}, query optimization \cite{Freitag2019}, or fraud detection \cite{Chabchoub2014}.

\subsection{HyperMinHash}
\label{sec:hyperminhash_intro}
HyperMinHash \cite{Yu2020} is the first approach to combine some properties of \ac{HLL} and \ac{MH} in one data structure. It can be seen as a generalization of \ac{HLL} and one permutation hashing. HyperMinHash supports cardinality and joint estimation by using larger registers and hence more space than \ac{HLL}. Unfortunately, the theoretically derived Jaccard similarity estimator is relatively complex and too expensive for practical use. Therefore, an approximation based on empirical observations was proposed.
However, the original estimation approach is not optimal, as our results presented later have shown that even the naive inclusion-exclusion principle works better in some scenarios.

HyperMinHash with parameter $\symHyperMinHashParameter$ is very similar to \ac{GHLL} with base $\symBase=2^{2^{-\symHyperMinHashParameter}}$. The reason is that HyperMinHash approximates the probability distribution of \ac{GHLL} with probabilities that are just powers of $\frac{1}{2}$ as shown in \cref{fig:probability_densities}. This has the advantage that hash values can be mapped to corresponding update values using only a few CPU instructions. Like those of \ac{HLL} and \ac{GHLL}, the register values of HyperMinHash are not locality-sensitive.

\begin{figure}
  \centering
  \includegraphics[width=\linewidth]{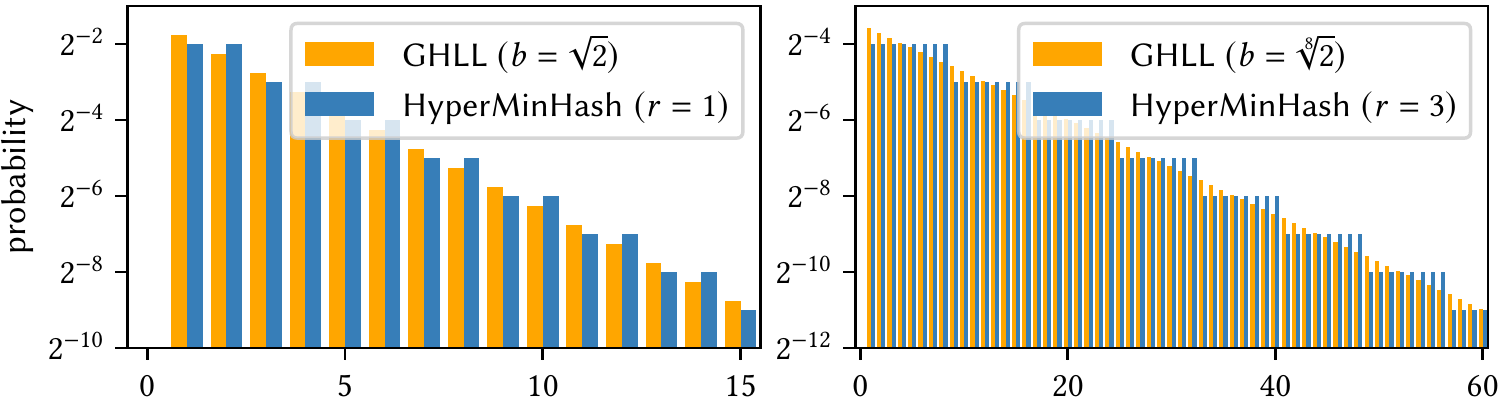}
  \caption{Probabilities of register update values for \acs*{GHLL} and HyperMinHash with equivalent configurations.} \label{fig:probability_densities}
\end{figure}

\subsection{Further Related Work}
A simple, though obviously not very memory-efficient, way to combine the properties of \ac{MH} and \ac{HLL} is to use them both in parallel \cite{Pascoe2013, CastroFernandez2019}. 
Probably the best alternative to \ac{MH} and \ac{HLL} which also works for distributed data and which even supports binary set operations is the ThetaSketch \cite{Dasgupta2016}. The downsides are a significantly worse memory efficiency compared to \ac{HLL} in terms of cardinality estimation and that it is not locality-sensitive like \ac{MH}.

\section{Methodology}
The basics of SetSketch follow directly from the properties introduced in \cref{sec:intro}. Similar to \ac{MH} and \ac{HLL} we consider a mapping from a set $\symSetS$ to $\symNumReg$ integer-valued registers $\symRegValVariate_1,\symRegValVariate_2,\ldots,\symRegValVariate_\symNumReg$. Furthermore, we assume that register values are obtained through some hashing procedure and that they are identically distributed. Since we want to estimate the cardinality from the register values, the distribution of 
$\symRegValVariate_\symIndexI$ should only depend on the cardinality $\symCardinality = |\symSetS|$
\begin{equation*}
\symProbability(\symRegValVariate_\symIndexI\leq \symRegVal) = \symDistributionFunc(\symRegVal;\symCardinality).
\end{equation*}
Mergeability requires that the register value $\symRegValVariate_{\symSetA\cup\symSetB,\symIndexI}$ of the union of two sets $\symSetA$ and $\symSetB$ can be computed directly from the corresponding individual register values $\symRegValVariate_{\symSetA\symIndexI}$ and $\symRegValVariate_{\symSetB\symIndexI}$ using some binary operation $\symRegValVariate_{\symSetA\cup\symSetB,\symIndexI}=\symBinaryOperation(\symRegValVariate_{\symSetA\symIndexI}, \symRegValVariate_{\symSetB\symIndexI})$.
To enable cardinality estimation, register values need to have some monotonic dependence on the cardinality. Without loss of generality, we consider non-decreasing monotonicity. Therefore, and because $|\symSetA\cup\symSetB|\geq|\symSetA|$ and $|\symSetB|\geq |\symSetC| \Rightarrow |\symSetA\cup\symSetB| \geq |\symSetA\cup\symSetC|$, we want $\symBinaryOperation$ to satisfy $\symBinaryOperation(\symX,\symY)\geq \symX$ and $\symY\geq\symZ \Rightarrow \symBinaryOperation(\symX,\symY) \geq \symBinaryOperation(\symX,\symZ)$, respectively. The only binary operation with these properties and that is also idempotent and commutative, is the maximum function \ifextended (see \cref{lem:max_binary_op})\else\cite{Ertl2021}\fi. Therefore, we require
$\symRegValVariate_{\symSetA\cup\symSetB,\symIndexI} = \max(\symRegValVariate_{\symSetA\symIndexI}, \symRegValVariate_{\symSetB\symIndexI})$.

For two disjoint sets $\symSetA$ and $\symSetB$ with cardinalities $|\symSetA|=\symCardinalityA$ and  $|\symSetB|=\symCardinalityB$ we have $|\symSetA\cup\symSetB|=\symCardinalityA+\symCardinalityB$. In this case $\symRegValVariate_{\symSetA\symIndexI}$ and $\symRegValVariate_{\symSetB\symIndexI}$ are independent and therefore 
$\symProbability(\symRegValVariate_{\symSetA\cup\symSetB,\symIndexI}\leq \symRegVal) =\symProbability(\max(\symRegValVariate_{\symSetA\symIndexI}, \symRegValVariate_{\symSetB\symIndexI})\leq \symRegVal) =\symProbability(\symRegValVariate_{\symSetA\symIndexI}\leq \symRegVal)\cdot \symProbability(\symRegValVariate_{\symSetB\symIndexI}\leq \symRegVal)$ which results in a functional equation 
\begin{equation}
\label{equ:dist_join}
\symDistributionFunc(\symRegVal; \symCardinalityA + \symCardinalityB) 
=
\symDistributionFunc(\symRegVal; \symCardinalityA) 
\cdot
\symDistributionFunc(\symRegVal; \symCardinalityB).
\end{equation}
To allow estimation with constant relative error over a wide range of cardinalities, the location of the distribution should increase logarithmically with the cardinality $\symCardinality$, while its shape should remain essentially the same. For a discrete distribution this can be enforced by the functional equation
\begin{equation}
\label{equ:dist_shift}
\symDistributionFunc(\symRegVal;\symCardinality) = \symDistributionFunc(\symRegVal + 1;\symCardinality\symBase).
\end{equation}
Multiplying the cardinality with the constant $\symBase>1$, which corresponds to an increment on the logarithmic scale, shifts the distribution to the right by 1.

The system of functional equations composed of \eqref{equ:dist_join} and \eqref{equ:dist_shift} has the solution
\begin{equation}
\label{equ:set_sketch_distribution}
\symProbability(\symRegValVariate_\symIndexI \leq \symRegVal) =  \symDistributionFunc(\symRegVal;\symCardinality) = e^{-\symCardinality\symExponentialRate \symBase^{-\symRegVal}}
\end{equation}
with some constant $\symExponentialRate>0$ \ifextended(see \cref{lem:func_equation})\else\cite{Ertl2021}\fi. For a set with a single element, in particular, this is
\begin{equation}
\label{equ:reg_val_distribution}
\symProbability(\symRegValVariate_\symIndexI \leq \symRegVal\mid \symCardinality = 1) =  \symDistributionFunc(\symRegVal;1) = e^{-\symExponentialRate\symBase^{-\symRegVal}}.
\end{equation}
Therefore, $\symRegValVariate_\symIndexI$ needs to be distributed as
$\symRegValVariate_\symIndexI
\sim
\lfloor 1 - \log_\symBase\symStatisticX \rfloor$ where $\symStatisticX\sim\symExponential(\symExponentialRate)$ is exponentially distributed with rate $\symExponentialRate$.
This directly leads to the definition of our new SetSketch data structure, which sets the state for a given set $\symSetS$ as
\begin{equation}
\label{equ:setsketch_update}
\textstyle\symRegValVariate_\symIndexI := \max_{\symElement\in\symSetS} \lfloor 1 - \log_\symBase \symHashFunc_\symIndexI(\symElement)\rfloor\quad\textnormal{with}\ \symHashFunc_\symIndexI(\symElement) \sim \symExponential(\symExponentialRate),
\end{equation}
where $\symHashFunc_\symIndexI(\symElement)$ are hash functions with exponentially distributed output. This definition is very similar to that of \ac{HLL} without stochastic averaging \eqref{equ:hll_update} where the hash values are distributed uniformly instead of exponentially.

\subsection{Ordered Register Value Updates}

\myAlg{
\caption{SetSketch}
\label{alg:set_sketch}
\KwData{$\symSetS$}
\KwResult{$\symRegValVariate_1,\symRegValVariate_2,\ldots,\symRegValVariate_\symNumReg$}
$(\symRegValVariate_1,\symRegValVariate_2,\ldots,\symRegValVariate_\symNumReg)\gets(0,0,\ldots,0)$\;
$\symRegValVariateLow \gets 0$\;
$\symUpdateCounter \gets 0$\;
\ForAll{$\symElement\in\symSetS$}{
  initialize pseudorandom number generator with seed $\symElement$\;
  \For{$\symIndexJ\gets 1$ \KwTo $\symNumReg$}{
    $\symPoint_\symIndexJ \gets$ generate $\symIndexJ$-th smallest out of $\symNumReg$ exponentially distributed random values with rate $\symExponentialRate$ using \eqref{equ:uncorrelated} for SetSketch1 or \eqref{equ:correlated} for SetSketch2\;
    \lIf{$\symPoint_\symIndexJ > \symBase^{-\symRegValVariateLow}$}{\Break}
    $\symRegVal\gets \max(0, \min(\symMaxRegularValue + 1, \lfloor 1 - \log_\symBase \symPoint_\symIndexJ\rfloor))$\;
    \lIf{$\symRegVal \leq \symRegValVariateLow$}{\Break}
    $\symIndexI \gets$ sample from $\{1,2,\ldots,\symNumReg\}$ without replacement\;
    \If{$\symRegVal > \symRegValVariate_\symIndexI$}{
      $\symRegValVariate_\symIndexI\gets \symRegVal$\;
      $\symUpdateCounter\gets\symUpdateCounter+1$\;
      \If{$\symUpdateCounter \geq \symNumReg$}{
        $\symRegValVariateLow \gets \min(\symRegValVariate_1,\symRegValVariate_2,\ldots,\symRegValVariate_\symNumReg)$\;
        $\symUpdateCounter \gets 0$\;
      }
    }
  }
}
}

Updating SetSketch registers based on \eqref{equ:setsketch_update} is not very efficient, because processing a single element requires $\symNumReg$ hash function evaluations and $\symNumReg$ is typically in the hundreds or thousands. Therefore, \cref{alg:set_sketch} uses an alternative method based on ideas from our previous work \cite{Ertl2017b, Ertl2018, Ertl2020} to reduce the average time complexity for adding an element from $\symBigO(\symNumReg)$ to $\symBigO(1)$ for sets significantly larger than $\symNumReg$. 

Since register values are increasing with cardinality, only the smallest hash values $\symHashFunc_\symIndexI(\symElement)$ will be relevant for the final state according to \eqref{equ:setsketch_update}.
In particular, if $\symRegValVariateLow\leq\min(\symRegValVariate_1,\ldots,\symRegValVariate_\symNumReg)$ denotes some lower bound of the current state, only hash values $\symHashFunc_\symIndexI(\symElement)\leq \symBase^{-\symRegValVariateLow}$ will be able to alter any register.
As more elements are added to the SetSketch, the register values increase and so does their common minimum, allowing $\symRegValVariateLow$ to be set to higher and higher values. For large sets, $\symBase^{-\symRegValVariateLow}$ will eventually become so small that almost all hash values of further elements will be greater. Therefore, computing the hash values in ascending order would allow to stop the processing of an element after a hash value is greater than $\symBase^{-\symRegValVariateLow}$, and thus to achieve an asymptotic time complexity of $\symBigO(1)$.

For that, we consider all hash values $\symHashFunc_\symIndexI(\symElement)$ as random values generated by a pseudorandom number generator that was seeded with element $\symElement$. This perspective allows us to use any other random process that assigns exponentially distributed random values to each $\symHashFunc_\symIndexI(\symElement)$. 
In particular, we can use an appropriate ascending random sequence $0<\symPoint_1<\symPoint_2<\ldots<\symPoint_\symNumReg$ whose values are randomly shuffled and assigned to $\symHashFunc_1(\symElement),\symHashFunc_2(\symElement),\ldots, \symHashFunc_\symNumReg(\symElement)$. Shuffling corresponds to sampling without replacement and can be efficiently realized as described in \cite{Ertl2018,Ertl2020} based on the Fisher-Yates algorithm \cite{Fisher1938}.
The random values $\symPoint_\symIndexJ$ must be chosen such that $\symHashFunc_\symIndexI(\symElement)$ are eventually exponentially distributed as required by \eqref{equ:setsketch_update}. 

One possibility to achieve that is to use exponentially distributed spacings \cite{Devroye1986}
\begin{equation}
\label{equ:uncorrelated}
{\textstyle
\symPoint_\symIndexJ \sim \symPoint_{\symIndexJ-1} + \frac{1}{\symNumReg + 1 - \symIndexJ}\symExponential(\symExponentialRate)
\quad
\textnormal{with}
\
\symPoint_0 = 0}.
\end{equation}
In this way the final hash values $\symHashFunc_\symIndexI(\symElement)$ will be statistically independent and the state of SetSketch will look like as if it was generated using $\symNumReg$ independent hash functions.

Alternatively, the domain of the exponential distribution with rate $\symExponentialRate$ can be divided into $\symNumReg$ intervals $[\symGamma_{\symIndexJ-1},\symGamma_{\symIndexJ})$ with $\symGamma_0=0$ and $\symGamma_\symNumReg=\infty$. Setting $\symGamma_\symIndexJ:=\frac{1}{\symExponentialRate}\log(1 + \symIndexJ/(\symNumReg-\symIndexJ))$ ensures that 
an exponentially distributed random value $\symStatisticX\sim\symExponential(\symExponentialRate)$ has equal probability to fall into any of these $\symNumReg$ intervals $\symProbability(\symStatisticX\in[\symGamma_{\symIndexJ-1},\symGamma_{\symIndexJ})) = \frac{1}{\symNumReg}$ \ifextended(see \cref{lem:setsketch2})\else\cite{Ertl2021}\fi.
Hence, if the points are sampled in ascending order according to 
\begin{equation}
\label{equ:correlated}
\symPoint_\symIndexJ \sim
\symExponential(\symExponentialRate; \symGamma_{\symIndexJ-1},\symGamma_{\symIndexJ})
,
\end{equation}
where $\symExponential(\symExponentialRate; \symGamma_{\symIndexJ-1},\symGamma_{\symIndexJ})$ denotes the truncated exponential distribution with rate $\symExponentialRate$ and domain $[\symGamma_{\symIndexJ-1},\symGamma_{\symIndexJ})$, the shuffled assignment will lead to hash values $\symHashFunc_\symIndexI(\symElement)$ that are exponentially distributed with rate $\symExponentialRate$. However, in contrast to the first approach, the hash values $\symHashFunc_\symIndexI(\symElement)$ will be statistically dependent, because there is always exactly one point sampled from each interval. This correlation will be less significant for large sets $\symCardinality \gg \symNumReg$, where it is unlikely that one element is responsible for the values of different registers at the same time.

Dependent on which method is used in conjunction with \eqref{equ:setsketch_update} to calculate the register values, we refer to SetSketch1 for the uncorrelated approach using exponential spacings and to SetSketch2 for the correlated approach using sampling from disjoint intervals.

\subsection{Lower Bound Tracking}
\label{sec:lower_bound_tracking}

Different strategies can be used to keep track of a lower bound $\symRegValVariateLow$ as needed for an asymptotic constant-time insert operation. Since maintenance of the minimum of all register values with a small worst case complexity would require additional space by using either a histogram \cite{Ertl2017} or binary trees \cite{Ertl2018, Ertl2020}, we decided to simply update the lower bound regularly by scanning the whole register array as suggested in \cite{Reviriego2020} which takes $\symBigO(\symNumReg)$ time. However, since this approach is inefficient for small bases $\symBase$, \cref{alg:set_sketch} counts the number of register modifications $\symUpdateCounter$ instead of the number of register values greater than the minimum. After every $\symNumReg$ register updates when $\symUpdateCounter\geq \symNumReg$, $\symRegValVariateLow$ is updated and set to the current minimum of all register values, and $\symUpdateCounter$ is reset. By definition, this contributes at most amortized constant time to every register increment and hence does not change the expected time complexity.

\subsection{Parameter Configuration}
\label{sec:parameter_config}
SetSketch has 4 parameters, $\symNumReg$, $\symBase$, $\symExponentialRate$, and $\symMaxRegularValue$, that need to be set appropriately. If we have $\symNumReg$ registers, which are able to represent all values from $\lbrace 0,1,2,\ldots, \symMaxRegularValue,\symMaxRegularValue+1\rbrace$ with some nonnegative integer $\symMaxRegularValue$, the memory footprint will be $\symNumReg\lceil\log_2(\symMaxRegularValue+2)\rceil$ bits without special encoding. 
$\symNumReg$ determines the accuracy of cardinality estimates, because, as shown later, the \ac{RMSE} is in the range $[1/\sqrt{\symNumReg}, 1.04/\sqrt{\symNumReg}]$ for $\symBase\leq 2$. While the choice of the base $\symBase$ has only little influence on cardinality estimation, joint estimation and locality sensitivity are significantly improved as $\symBase\rightarrow 1$. 
The cardinality range, for which accurate estimation is possible, is controlled by the parameters $\symExponentialRate$ and $\symMaxRegularValue$.
Since the support of distribution \eqref{equ:set_sketch_distribution} is unbounded, $\symExponentialRate$ and $\symMaxRegularValue$ need to be chosen such that register values smaller than 0 or greater than $\symMaxRegularValue+1$ are very unlikely and can be ignored for the expected cardinality range. The lower bound of this range is typically 1 as we want SetSketches to be able to represent any set with at least one element.

If we choose $\symExponentialRate \geq \log(\symNumReg/\symSmallProbability)/\symBase$ for some $\symSmallProbability\ll 1$, it is guaranteed that negative register values occur only with a maximum probability of $\symSmallProbability$ \ifextended(see \cref{lem:singleton})\else\cite{Ertl2021}\fi. Therefore, the error introduced by using zero as initial value as done in \cref{alg:set_sketch} is negligible for small enough values of $\symSmallProbability$. In practice, setting $\symExponentialRate=20$ is a good choice in most cases. Even in the extreme case with $\symBase\rightarrow 1$ and $\symNumReg=2^{20}$, the probability is still less than \SI{0.22}{\percent} to encounter at least one negative value.
Similarly, setting $\symMaxRegularValue \geq \lfloor \log_\symBase \frac{\symNumReg\symCardinality_\textnormal{max}\symExponentialRate}{\symSmallProbability}\rfloor$ guarantees that register values greater than $\symMaxRegularValue+1$ occur only with a maximum probability of $\symSmallProbability$ for all cardinalities up to $\symCardinality_\textnormal{max}$ \ifextended(see \cref{lem:max})\else\cite{Ertl2021}\fi.
A small $\symBase$ implies a larger value of $\symMaxRegularValue$ and therefore also a larger memory footprint, if the cardinality range is fixed.

To give a concrete example, consider a SetSketch with parameters $\symNumReg=4096$, $\symBase=1.001$, $\symExponentialRate = 20$, and $\symMaxRegularValue = 2^{16}-2 = 65534$. The last parameter ensures that two bytes are sufficient to represent a single register and the whole data structure takes  $\SI{8}{\kilo\byte}$. 
The probability that there is at least one register with negative value is \num{8.28e-6} for a set with just a single element. Furthermore, the probability that any register value is greater than $\symMaxRegularValue+1$ is \num{2.93e-6} for $\symCardinality=\num{e18}$. Therefore, a SetSketch using this configuration is suitable to represent any set with up to \num{e18} distinct elements. The expected error of cardinality estimates is approximately $1/\sqrt{\symNumReg}\approx\SI{1.56}{\percent}$.

\section{Estimation}
In this section we present methods for cardinality and joint estimation. We will assume that register values are statistically independent which is only true for SetSketch1 and only a good approximation for SetSketch2 in the case of large sets with $\symCardinality\gg\symNumReg$. However, our experiments presented later have shown that the estimators derived under this assumption also perform well for SetSketch2 over the entire cardinality range. The correlation even has a positive effect and reduces the estimation errors for small sets.

We introduce some special functions and corresponding approximations that are used for the derivation of the estimators. The functions $\symPowerSeriesFunc_\symBase^1(\symX)$ and $\symPowerSeriesFunc_\symBase^2(\symX)$ defined by
\begin{equation}
\label{equ:approx_xi}
  \symPowerSeriesFunc_\symBase^\symIntPower(\symX):= 
{\textstyle
\frac{\log \symBase}{\Gamma(\symIntPower)} \sum_{\symRegVal = -\infty}^\infty
  \symBase^{\symIntPower (\symX-\symRegVal)}
  e^{-\symBase^{\symX-\symRegVal}}}
\approx 1,
\end{equation}
where $\Gamma$ denotes the gamma function, are both very close to 1 for $\symBase\leq 2$ \ifextended(see \cref{lem:xi1_approx} and \cref{lem:xi2_approx})\else\cite{Ertl2021}\fi. For $\symBase=2$ we have $\max_\symX \mathopen|\symPowerSeriesFunc_\symBase^1(\symX) - 1\mathclose|\leq \num{e-5}$ and $\max_\symX \mathopen|\symPowerSeriesFunc_\symBase^2(\symX) - 1\mathclose|\leq \num{e-4}$, which shows that these approximations introduce only a small error. Smaller values of $\symBase$ lead to even smaller errors as both functions converge quickly towards 1 as $\symBase\rightarrow 1$.
Another approximation that we will use is
\begin{equation}
\label{equ:approx_zeta}
\symHelperFunc_\symBase(\symX_1,\symX_2):={\textstyle\sum_{\symRegVal=-\infty}^\infty e^{-\symBase^{\symX_1-\symRegVal}} - e^{-\symBase^{\symX_2-\symRegVal}}}
\approx
\symX_2-\symX_1
.
\end{equation}
The relative error $|\frac{\symHelperFunc_\symBase(\symX_1,\symX_2) - (\symX_2 - \symX_1)}{\symX_2 - \symX_1}|$ is smaller than \num{e-5} for $\symBase=2$ and approaches zero quickly as $\symBase\rightarrow 1$ \ifextended(see \cref{lem:zeta_approx})\else\cite{Ertl2021}\fi.

\subsection{Cardinality Estimation}
\label{sec:cardinality_estimation}
Instead of using the straightforward way of estimating the cardinality using the \ac{ML} method based on \eqref{equ:set_sketch_distribution}, we derive a closed-form estimator, based on our previous work \cite{Ertl2017, Ertl2017a}, that is simpler to implement, cheaper to evaluate, and gives almost identical results.
The expectation of $(\symExponentialRate \symBase^{-\symRegValVariate_\symIndexI})^\symIntPower$ with $\symIntPower\in\lbrace 1,2 \rbrace$ and $\symRegValVariate_\symIndexI$ distributed according to \eqref{equ:set_sketch_distribution} is given by
\begin{align}
\label{equ:moment}
\symExpectation((\symExponentialRate \symBase^{-\symRegValVariate_\symIndexI})^\symIntPower) &=
  {\textstyle\sum_{\symRegVal = -\infty}^\infty
(\symExponentialRate\symBase^{-\symRegVal})^\symIntPower
  (e^{-\symCardinality\symExponentialRate \symBase^{-\symRegVal}}-e^{-\symCardinality \symExponentialRate \symBase^{-\symRegVal+1}})}\nonumber\\
  &={\textstyle(1-\symBase^{-\symIntPower})\sum_{\symRegVal = -\infty}^\infty
  (\symExponentialRate\symBase^{-\symRegVal})^\symIntPower
  e^{-\symCardinality\symExponentialRate \symBase^{-\symRegVal}}}
  \nonumber\\
  &={\textstyle\frac{1-\symBase^{-\symIntPower}}{\symCardinality^\symIntPower}\sum_{\symRegVal = -\infty}^\infty
  \symBase^{{\symIntPower\left(\log_\symBase(\symCardinality\symExponentialRate)-\symRegVal\right) } }
  e^{-\symBase^{\left(\log_\symBase(\symCardinality\symExponentialRate)-\symRegVal\right)}}}
  \nonumber\\
  &=
  {\textstyle\frac{(1-\symBase^{-\symIntPower})\,\Gamma(\symIntPower)}{\symCardinality^\symIntPower \log \symBase}
  \symPowerSeriesFunc_\symBase^\symIntPower(\log_\symBase(\symCardinality\symExponentialRate))
\approx
  \frac{(1-\symBase^{-\symIntPower})\,\Gamma(\symIntPower)}{\symCardinality^\symIntPower \log \symBase}}.
\end{align}
Here the asymptotic identity \eqref{equ:approx_xi} was used as approximation in the final step. 
We now consider the statistic $\symStatisticX_\symNumReg = \frac{\log \symBase}{1-1/\symBase}\frac{1}{\symNumReg}\sum_{\symIndexI=1}^\symNumReg \symExponentialRate \symBase^{-\symRegValVariate_\symIndexI}$. Using \eqref{equ:moment} for the cases $\symIntPower=1$ and $\symIntPower=2$ we get for its expectation and its variance \ifextended(see \cref{lem:card_variance})\else\cite{Ertl2021}\fi
\begin{equation*}
{\textstyle
\symExpectation(\symStatisticX_\symNumReg) \approx \frac{1}{\symCardinality}
\quad
\textnormal{and}
\quad
\symVariance(\symStatisticX_\symNumReg) \approx \frac{1}{\symNumReg\symCardinality^2}\left(\frac{\symBase+1}{\symBase-1}\log(\symBase) - 1\right)}.
\end{equation*}
The delta method \cite{Casella2002} gives for $\symNumReg\rightarrow\infty$ 
\begin{equation*}
{\textstyle
\symExpectation(\symStatisticX_\symNumReg^{-1}) \approx \symCardinality
\quad
\textnormal{and}
\quad
\symVariance(\symStatisticX_\symNumReg^{-1}) \approx \frac{\symCardinality^2}{\symNumReg}\left(\frac{\symBase+1}{\symBase-1}\log(\symBase) - 1\right)}
\end{equation*}
which suggests to use $\symStatisticX_\symNumReg^{-1}$ as estimator for the cardinality $\symCardinality$
\begin{equation}
\label{equ:raw_cardinality_estimator}
\textstyle\symCardinalityEstimate
=
\frac{\symNumReg(1-1/\symBase)}{\symExponentialRate\log(\symBase)\sum_{\symIndexI=1}^\symNumReg  \symBase^{-\symRegValVariate_\symIndexI}}.
\end{equation}
This estimator can be quickly evaluated, if the powers of $\symBase$ are precalculated and stored in a lookup table for all possible exponents which are known to be from $\lbrace 0,1,\ldots,\symMaxRegularValue+1\rbrace$. 

The corresponding \ac{RSD} $\sqrt{\symVariance(\symCardinalityEstimate)}/\symCardinality$ is given by $\sqrt{\frac{1}{\symNumReg}(\frac{\symBase+1}{\symBase-1}\log(\symBase) - 1)}$. It is minimized for $\symBase\rightarrow 1$, where it equals $1/\sqrt{\symNumReg}$, and increases slowly with $\symBase$. For $\symBase=2$ the expected error is still as low as $1.04/\sqrt{m}$. However, as already discussed, smaller values of $\symBase$ require larger values of $\symMaxRegularValue$ and hence more space. Optimal memory efficiency, measured as the product of the variance and the memory footprint, is typically obtained for values of $\symBase$ ranging from 2 to 4. Theoretically, if register values are compressed, a better memory efficiency can be achieved for $\symBase\rightarrow\infty$ \cite{Pettie2020}.

Values significantly greater than $2$ invalidate the approximation \eqref{equ:approx_xi} which was used for the derivation of estimator \eqref{equ:raw_cardinality_estimator}. To reduce the estimation error in this regime, random offsets, which correspond to register-specific values of parameter $\symExponentialRate$, would be necessary \cite{Pettie2020, Lukasiewicz2020}. 
This work, however, focuses on $\symBase\leq 2$. Although small values of $\symBase$ are inefficient for cardinality estimation, they contain, as we will show, more information about how sets relate to each other, which ultimately justifies slightly larger memory footprints.

\subsection{Joint Estimation}
\label{sec:joint_estimation}
The relationship of two sets $\symSetA$ and $\symSetB$ can be characterized by three quantities. Without loss of generality we choose 
the cardinalities $\symCardinalityA = |\symSetA|$, $\symCardinalityB = |\symSetB|$, and the Jaccard similarity $\symJaccard = \frac{|\symSetA\cap\symSetB|}{|\symSetA\cup\symSetB|}$ with the natural constraint $\symJaccard\in[0,\min({\textstyle\frac{\symCardinalityA}{\symCardinalityB},\frac{\symCardinalityB}{\symCardinalityA}})]$ for parameterization.
Other quantities such as
\begin{align*}
&|\symSetA\cup\symSetB| = {\textstyle\frac{\symCardinalityA+\symCardinalityB}{1+\symJaccard}}, &(\textnormal{union size})\\
&|\symSetA\cap\symSetB| = {\textstyle\frac{(\symCardinalityA+\symCardinalityB)\symJaccard}{1+\symJaccard}}, &(\textnormal{intersection size})\\
&|\symSetA\setminus\symSetB| = {\textstyle\frac{\symCardinalityA - \symCardinalityB\symJaccard}{1+\symJaccard}},\ |\symSetB\setminus\symSetA| = {\textstyle\frac{\symCardinalityB - \symCardinalityA\symJaccard}{1+\symJaccard}},  &(\textnormal{difference sizes})\\
&\textstyle{\frac{|\symSetA\cap\symSetB|}{\sqrt{|\symSetA||\symSetB|}}} = \textstyle{\frac{(\symCardinalityA+\symCardinalityB)\symJaccard}{\sqrt{\symCardinalityA\symCardinalityB}(1+\symJaccard)}}, &(\textnormal{cosine similarity})\\
&\textstyle{\frac{|\symSetA\cap\symSetB|}{|\symSetA|}} = {\textstyle\frac{(\symCardinalityA+\symCardinalityB)\symJaccard}{\symCardinalityA(1+\symJaccard)}},\ \textstyle{\frac{|\symSetA\cap\symSetB|}{|\symSetB|}} = {\textstyle\frac{(\symCardinalityA+\symCardinalityB)\symJaccard}{\symCardinalityB(1+\symJaccard)}}, &(\textnormal{inclusion coefficients})\\
\end{align*}
can be expressed in terms of only these three variables. If we have two SetSketches representing two different sets, we would like to find estimates for $\symCardinalityA$, $\symCardinalityB$, and $\symJaccard$, which would allow us to estimate any other quantity of interest.

While we can use the cardinality estimator derived in the previous section to get estimates for $\symCardinalityA$ and $\symCardinalityB$, respectively, the estimation of $\symJaccard$ is more challenging. 
One possibility is to leverage the mergeability of data sketches and use the inclusion-exclusion principle. Assume $\symCardinalityAEstimate$, $\symCardinalityBEstimate$, and $\symCardinalityUnionEstimate$ are cardinality estimates of $|\symSetA|$, $|\symSetB|$, and $|\symSetA\cup\symSetB|$, respectively. Then the inclusion-exclusion principle allows estimating the intersection size as $\symCardinalityAEstimate+\symCardinalityBEstimate-\symCardinalityUnionEstimate$. This can be used to estimate $\symJaccard$ as
\begin{equation}
\label{equ:joint_incl_excl}
{\textstyle
\symJaccardEstimate_\textnormal{in-ex} = 
\frac{\symCardinalityAEstimate + \symCardinalityBEstimate-\symCardinalityUnionEstimate}{\symCardinalityUnionEstimate}}.
\end{equation}
To ensure the natural constraint $\symJaccard\in[0,\min({\textstyle\frac{\symCardinalityA}{\symCardinalityB},\frac{\symCardinalityB}{\symCardinalityA}})]$ and the non-negativity of estimated intersection and difference sizes, it might be necessary to trim $\symJaccardEstimate_\textnormal{in-ex}$ to the range $[0, \min(\symCardinalityAEstimate/\symCardinalityBEstimate,\symCardinalityBEstimate/\symCardinalityAEstimate)]$.

In the following a new joint estimation approach is derived that dominates the inclusion-exclusion principle as our experiments have shown. Using \eqref{equ:set_sketch_distribution}, the joint cumulative distribution function of registers $\symRegValVariate_{\symSetA\symIndexI}$ and $\symRegValVariate_{\symSetB\symIndexI}$ of two SetSketches representing sets $\symSetA$ and $\symSetB$, respectively, is given by
\begin{align*}
&\symProbability(\symRegValVariate_{\symSetA\symIndexI} \leq \symRegVal_\symSetA\wedge\symRegValVariate_{\symSetB\symIndexI} \leq \symRegVal_\symSetB)
\\
&=
\symProbability(\symRegValVariate_{\symSetA\setminus\symSetB,\symIndexI} \leq \symRegVal_\symSetA)
\,
\symProbability(\symRegValVariate_{\symSetB\setminus\symSetA,\symIndexI} \leq \symRegVal_\symSetB)
\,
\symProbability(\symRegValVariate_{\symSetA\cap\symSetB,\symIndexI} \leq \min(\symRegVal_\symSetA, \symRegVal_\symSetB))
\\
&=
e^{
-
\frac{\symCardinalityA - \symCardinalityB\symJaccard}{1+\symJaccard} \symExponentialRate \symBase^{-\symRegVal_\symSetA}
}
e^{
-
\frac{\symCardinalityB - \symCardinalityA\symJaccard}{1+\symJaccard} \symExponentialRate \symBase^{-\symRegVal_\symSetB}
}
e^{
-
\frac{(\symCardinalityA+\symCardinalityB)\symJaccard}{1+\symJaccard} \symExponentialRate \symBase^{-\min(\symRegVal_\symSetA, \symRegVal_\symSetB)}
}
\\
&=
\begin{cases}
e^{
-\symExponentialRate(\symBase^{-\symRegVal_\symSetA}\frac{\symCardinalityA - \symCardinalityB\symJaccard}{1+\symJaccard} 
+\symBase^{-\symRegVal_\symSetB}\symCardinalityB)} & \symRegVal_\symSetA \geq \symRegVal_\symSetB,\\
e^{
-\symExponentialRate(\symBase^{-\symRegVal_\symSetB}\frac{\symCardinalityB - \symCardinalityA\symJaccard}{1+\symJaccard}  + \symBase^{-\symRegVal_\symSetA}\symCardinalityA )} & \symRegVal_\symSetA \leq \symRegVal_\symSetB.
\end{cases}
\end{align*}
It can be used to calculate the probability that $\symRegValVariate_{\symSetA\symIndexI} > \symRegValVariate_{\symSetB\symIndexI}$
\begin{align*}
&\symProbability(\symRegValVariate_{\symSetA\symIndexI} > \symRegValVariate_{\symSetB\symIndexI})
=
{\textstyle\sum_{\symRegVal=-\infty}^\infty
\symProbability(\symRegValVariate_{\symSetA\symIndexI} = \symRegVal+1 \wedge \symRegValVariate_{\symSetB\symIndexI} \leq \symRegVal)}
\\
&={\textstyle\sum_{\symRegVal=-\infty}^\infty
\symProbability(\symRegValVariate_{\symSetA\symIndexI} \leq \symRegVal+1 \wedge \symRegValVariate_{\symSetB\symIndexI} \leq \symRegVal)
 -
\symProbability(\symRegValVariate_{\symSetA\symIndexI} \leq \symRegVal\wedge \symRegValVariate_{\symSetB\symIndexI} \leq \symRegVal)}
\\
&={\textstyle\sum_{\symRegVal=-\infty}^\infty
e^{
-\symExponentialRate\symBase^{-\symRegVal}(\frac{\symCardinalityA - \symCardinalityB\symJaccard}{\symBase(1+\symJaccard)}  + \symCardinalityB)}
-
e^{
-\symExponentialRate\symBase^{-\symRegVal}\frac{\symCardinalityA + \symCardinalityB}{1+\symJaccard}}}
\\
&=
\symHelperFunc_\symBase\!\left(
\log_\symBase\!\left(\symExponentialRate\left({\textstyle\frac{\symCardinalityA - \symCardinalityB\symJaccard}{\symBase(1+\symJaccard)}}+ \symCardinalityB\right)\right),
\log_\symBase\!\left(\symExponentialRate{\textstyle\frac{\symCardinalityA + \symCardinalityB}{1+\symJaccard}}\right)
\right)
\\
&\approx
\log_\symBase\!\left(\symExponentialRate{\textstyle\frac{\symCardinalityA + \symCardinalityB}{1+\symJaccard}}\right)
-
\log_\symBase\!\left(\symExponentialRate\left({\textstyle\frac{\symCardinalityA - \symCardinalityB\symJaccard}{\symBase(1+\symJaccard)}}+ \symCardinalityB\right)\right)
=
\symProbFunc_{\symBase}\!\left({\textstyle\frac{\symCardinalityA - \symCardinalityB\symJaccard}{\symCardinalityA + \symCardinalityB}}\right).
\end{align*}
Here we used the approximation \eqref{equ:approx_zeta} which is asymptotically equal as $\symBase\rightarrow1$ and introduces a negligible error if $\symBase\leq 2$. $\symProbFunc_{\symBase}$ is defined as $\symProbFunc_{\symBase}(\symX):=-\log_\symBase(
1-\symX\frac{\symBase-1}{\symBase}
)$. Complemented by $\symProbability(\symRegValVariate_{\symSetA\symIndexI} < \symRegValVariate_{\symSetB\symIndexI})$, which is obtained analogously, and $\symProbability(\symRegValVariate_{\symSetA\symIndexI} = \symRegValVariate_{\symSetB\symIndexI}) = 1 - \symProbability(\symRegValVariate_{\symSetA\symIndexI} > \symRegValVariate_{\symSetB\symIndexI}) -\symProbability(\symRegValVariate_{\symSetA\symIndexI} < \symRegValVariate_{\symSetB\symIndexI})$ we have
\begin{align}
\symProbability(\symRegValVariate_{\symSetA\symIndexI} > \symRegValVariate_{\symSetB\symIndexI})
&\approx
\symProbFunc_{\symBase}(\symCardinalityANorm - \symCardinalityBNorm\symJaccard),
\quad
\symProbability(\symRegValVariate_{\symSetA\symIndexI} < \symRegValVariate_{\symSetB\symIndexI})
\approx
\symProbFunc_{\symBase}(\symCardinalityBNorm - \symCardinalityANorm\symJaccard),
\nonumber
\\
\symProbability(\symRegValVariate_{\symSetA\symIndexI} = \symRegValVariate_{\symSetB\symIndexI})
&\approx
1
-
\symProbFunc_{\symBase}(\symCardinalityANorm - \symCardinalityBNorm\symJaccard)
-
\symProbFunc_{\symBase}(\symCardinalityBNorm - \symCardinalityANorm\symJaccard),\label{equ:probequal}
\end{align}
where we introduced the relative cardinalities $\symCardinalityANorm = \frac{\symCardinalityA}{\symCardinalityA+\symCardinalityB}$ and $\symCardinalityBNorm = \frac{\symCardinalityB}{\symCardinalityA+\symCardinalityB}$ with $\symCardinalityANorm+\symCardinalityBNorm=1$.
The approximation of $\symProbability(\symRegValVariate_{\symSetA\symIndexI} = \symRegValVariate_{\symSetB\symIndexI})$ is also always nonnegative \ifextended(see \cref{lem:prob_consistency})\else\cite{Ertl2021}\fi.

If the cardinalities $\symCardinalityA$ and $\symCardinalityB$ are known, the \ac{ML} method can be used to estimate $\symJaccard$. The corresponding log-likelihood function as function of $\symJaccard$ using the approximated probabilities is
\begin{multline*}
\log\symLikelihood(\symJaccard) = \symDiffCountVariate_{+} \log(\symProbFunc_{\symBase}(\symCardinalityANorm - \symCardinalityBNorm\symJaccard)) + \symDiffCountVariate_{-}\log(\symProbFunc_{\symBase}(\symCardinalityBNorm - \symCardinalityANorm\symJaccard))
\\
+\symDiffCountVariate_{0}\log(
1
-
\symProbFunc_{\symBase}(\symCardinalityANorm - \symCardinalityBNorm\symJaccard)
-
\symProbFunc_{\symBase}(\symCardinalityBNorm - \symCardinalityANorm\symJaccard)
),
\end{multline*}
where $\symDiffCountVariate_{+} := |\lbrace \symIndexI : \symRegValVariate_{\symSetA\symIndexI} > \symRegValVariate_{\symSetB\symIndexI}\rbrace|$, $\symDiffCountVariate_{-} := |\lbrace \symIndexI : \symRegValVariate_{\symSetA\symIndexI} < \symRegValVariate_{\symSetB\symIndexI} \rbrace|$, and $\symDiffCountVariate_{0} := |\lbrace \symIndexI : \symRegValVariate_{\symSetA\symIndexI} = \symRegValVariate_{\symSetB\symIndexI} \rbrace|$ are the number of registers in the sketch of $\symSetA$ that are greater than, less than, or equal to those in the sketch of $\symSetB$, respectively. 
Since this log-likelihood function is cheap to evaluate requiring only 5 logarithm evaluations and is strictly concave on the domain $\symJaccard \in [0, \min(\frac{\symCardinalityA}{\symCardinalityB}, \frac{\symCardinalityB}{\symCardinalityA})]=[0, \min(\frac{\symCardinalityANorm}{\symCardinalityBNorm}, \frac{\symCardinalityBNorm}{\symCardinalityANorm})]$ at least for all $\symBase\leq e\approx 2.718$ \ifextended(see \cref{lem:concavity})\else\cite{Ertl2021}\fi, the \ac{ML} estimate for $\symJaccard$ can be quickly and robustly found using standard univariate optimization algorithms like Brent's method \cite{Brent1973}. 

Since the \ac{ML} estimator is asymptotically efficient, the \ac{RMSE} of the \ac{ML} estimate is expected to be equal to $\symFisher^{-1/2}(\symJaccard)$ as $\symNumReg\rightarrow\infty$. $\symFisher(\symJaccard)$ denotes the Fisher information with respect to $\symJaccard$ for known cardinalities $\symCardinalityA$ and $\symCardinalityB$, which can be derived as \ifextended(see \cref{lem:fisher})\else\cite{Ertl2021}\fi
\begin{equation*}
\symFisher(\symJaccard)
= 
{\scriptstyle 
\frac{\symNumReg
(\symBase - 1)^2}{\symBase^2\log^2(\symBase)}
\left(
\frac{
\left(\symCardinalityBNorm
\symBase^{\symProbFunc_\symBase(\symCardinalityANorm - \symCardinalityBNorm \symJaccard)}
\right)^2
}{
\symProbFunc_\symBase(\symCardinalityANorm - \symCardinalityBNorm \symJaccard)}
+
\frac{
\left(\symCardinalityANorm
\symBase^{\symProbFunc_\symBase(\symCardinalityBNorm - \symCardinalityANorm \symJaccard)}
\right)^2
}{
\symProbFunc_\symBase(\symCardinalityBNorm - \symCardinalityANorm \symJaccard)
}
+
\frac{
\left(
\symCardinalityBNorm
\symBase^{\symProbFunc_\symBase(\symCardinalityANorm - \symCardinalityBNorm \symJaccard)}
+
\symCardinalityANorm
\symBase^{\symProbFunc_\symBase(\symCardinalityBNorm - \symCardinalityANorm \symJaccard)}
\right)^2
}{
1
-
\symProbFunc_\symBase(\symCardinalityANorm - \symCardinalityBNorm \symJaccard)
-
\symProbFunc_\symBase(\symCardinalityBNorm - \symCardinalityANorm \symJaccard)
}
\right)
}.
\end{equation*}
The asymptotic \ac{RMSE} can be compared with the Jaccard estimator of \ac{MH} for which the \ac{RMSE} is $\sqrt{\symJaccard(1-\symJaccard)/\symNumReg}$. The corresponding ratio is shown in \cref{fig:theoretical_variance} for the cases $\symCardinalityA=\symCardinalityB$ and $\symCardinalityA=0.5\symCardinalityB$, and various values of $\symBase$. For $\symCardinalityA=\symCardinalityB$ the curves approach 1 as $\symBase\rightarrow 1$. If $\symBase=1.001$ the difference from 1 is already very small and hence the estimation error for $\symJaccard$ will be almost the same as for \ac{MH} with the same number of registers $\symNumReg$. Given that 2-byte registers are sufficient for $\symBase=1.001$ as discussed in \cref{sec:parameter_config}, this significantly improves the space efficiency compared to \ac{MH} which uses 4- or even 8-byte components.
As shown in \cref{fig:theoretical_variance} for $\symCardinalityA=0.5\symCardinalityB$, the error can be even smaller than that of the \ac{MH} Jaccard estimator. The reason is that our estimation approach incorporates $\symCardinalityA$ and $\symCardinalityB$ and also if register values are greater or smaller while the classic \ac{MH} estimator only considers the number of equal registers.

If the true cardinalities $\symCardinalityA$ and $\symCardinalityB$ are not known, we simply replace them by corresponding estimates $\symCardinalityAEstimate$ and $\symCardinalityBEstimate$ using \eqref{equ:raw_cardinality_estimator}.
As a consequence, the \ac{RMSE} will apparently increase and $\symFisher^{-1/2}(\symJaccard)$ will only represent an asymptotic lower bound as $\symNumReg\rightarrow\infty$. Nevertheless, as supported by our experiments presented later, the estimation error is indistinguishable from the case with known cardinalities if the sets have equal size, thus if $\symCardinalityA=\symCardinalityB$. The reason is that $\symDiffCountVariate_{+}$ and $\symDiffCountVariate_{-}$ are identically distributed in this case, which makes the log-likelihood function symmetric with respect to $\symCardinalityANorm$ and $\symCardinalityBNorm$. Since this also implies $\symCardinalityANorm=\symCardinalityBNorm=\frac{1}{2}$ and estimates $\symCardinalityANormEstimate=\symCardinalityAEstimate/(\symCardinalityAEstimate+\symCardinalityBEstimate)$ and $\symCardinalityBNormEstimate=\symCardinalityBEstimate/(\symCardinalityAEstimate+\symCardinalityBEstimate)$ for $\symCardinalityANorm$ and $\symCardinalityBNorm$ satisfy $\symCardinalityANormEstimate+\symCardinalityBNormEstimate=1$ by definition, they can be written as $\symCardinalityANormEstimate=\frac{1}{2} + \symError$ and $\symCardinalityBNormEstimate=\frac{1}{2} - \symError$, respectively, with some estimation error $\symError$. Therefore and due to the symmetry of the log-likelihood function, the \ac{ML} estimator is an even function of $\symError$. The estimated Jaccard similarity will therefore differ only by an $\symBigO(\symError^2)$ term from the estimate based on the true cardinalities $\symCardinalityA$ and $\symCardinalityB$. Since $\symError\ll 1$ for sufficiently large $\symNumReg$, this second order term can be ignored in practice.

\begin{figure}[t]
  \centering
  \includegraphics[width=\linewidth]{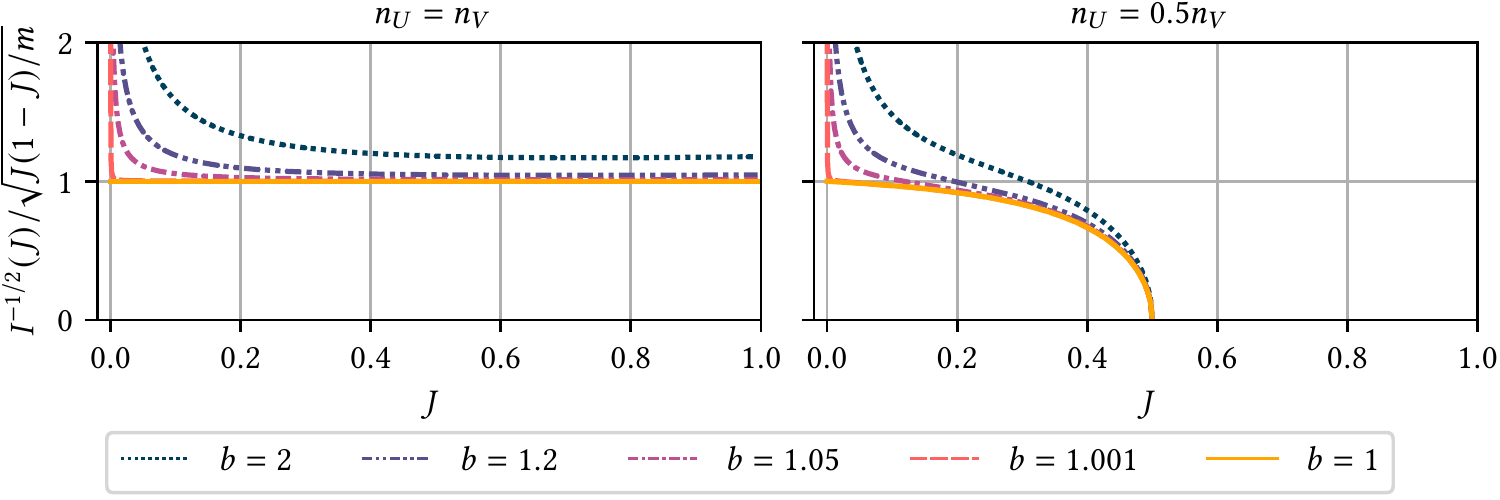}
  \caption{\boldmath Asymptotic \acs*{RMSE} of the new estimation approach with known cardinalities $\symCardinalityA$ and $\symCardinalityB$ relative to the \acs*{RMSE} of \acs*{MH} $\sqrt{\symJaccard(1-\symJaccard)/\symNumReg}$ with equal number of registers $\symNumReg$.}
  \label{fig:theoretical_variance}
\end{figure}

\subsection{Locality Sensitivity}
A collision probability, that is a monotonic function of some similarity measure, is the foundation of \ac{LSH} \cite{Indyk1998, Bawa2005, Lv2007,Zhu2016}. Rewriting of \eqref{equ:probequal} and using $\symCardinalityANorm+\symCardinalityBNorm=1$ gives for the probability of a register to be equal in two different SetSketches
\begin{equation*}
\symProbability(\symRegValVariate_{\symSetA\symIndexI} = \symRegValVariate_{\symSetB\symIndexI})
\approx
\textstyle
\log_\symBase(
1+\symJaccard(\symBase-1)
+
{\textstyle\frac{(\symBase-1)^2}{\symBase}(\symCardinalityANorm-\symCardinalityBNorm\symJaccard)(\symCardinalityBNorm-\symCardinalityANorm\symJaccard)}
).
\end{equation*}
Using $0\leq (\symCardinalityANorm-\symCardinalityBNorm\symJaccard)(\symCardinalityBNorm-\symCardinalityANorm\symJaccard)\leq \frac{1}{4}(1-\symJaccard)^2$ \ifextended(see \cref{lem:lsh_inequality}) \else\cite{Ertl2021} \fi
leads to
\begin{equation*}
\log_\symBase(1+\symJaccard(\symBase-1))\lesssim
\symProbability(\symRegValVariate_{\symSetA\symIndexI} = \symRegValVariate_{\symSetB\symIndexI})
\lesssim
\log_\symBase(\scriptstyle 1+\symJaccard(\symBase-1)
+
(1-\symJaccard)^2{\frac{(\symBase-1)^2}{4\symBase}}
).
\end{equation*}
These bounds are illustrated in \cref{fig:collision_probability} for $\symBase\in\lbrace 2,1.2, 1.001\rbrace$. They are very tight for large $\symJaccard$. Both bounds approach $\symProbability(\symRegValVariate_{\symSetA\symIndexI} = \symRegValVariate_{\symSetB\symIndexI})=\symJaccard$ as $\symBase\rightarrow 1$. Estimating $\symProbability(\symRegValVariate_{\symSetA\symIndexI} = \symRegValVariate_{\symSetB\symIndexI})$ by $\symDiffCountVariate_0/\symNumReg$ and resolving for $\symJaccard$ results in corresponding lower and upper bound estimators
\begin{equation}
\label{equ:jaccard_estimator}
\symJaccardEstimate_\textnormal{low}
:=
\max\!\left(
0,
2
\textstyle
\frac{
\symBase^{\frac{\symDiffCountVariate_0/\symNumReg+1}{2}}
-1
}
{\symBase -1}
-
1
\right),
\quad
\symJaccardEstimate_\textnormal{up}
:=
\frac{\symBase^{\symDiffCountVariate_0/\symNumReg}-1}{\symBase -1}.
\end{equation}

\begin{figure}[t]
  \centering
  \includegraphics[width=\linewidth]{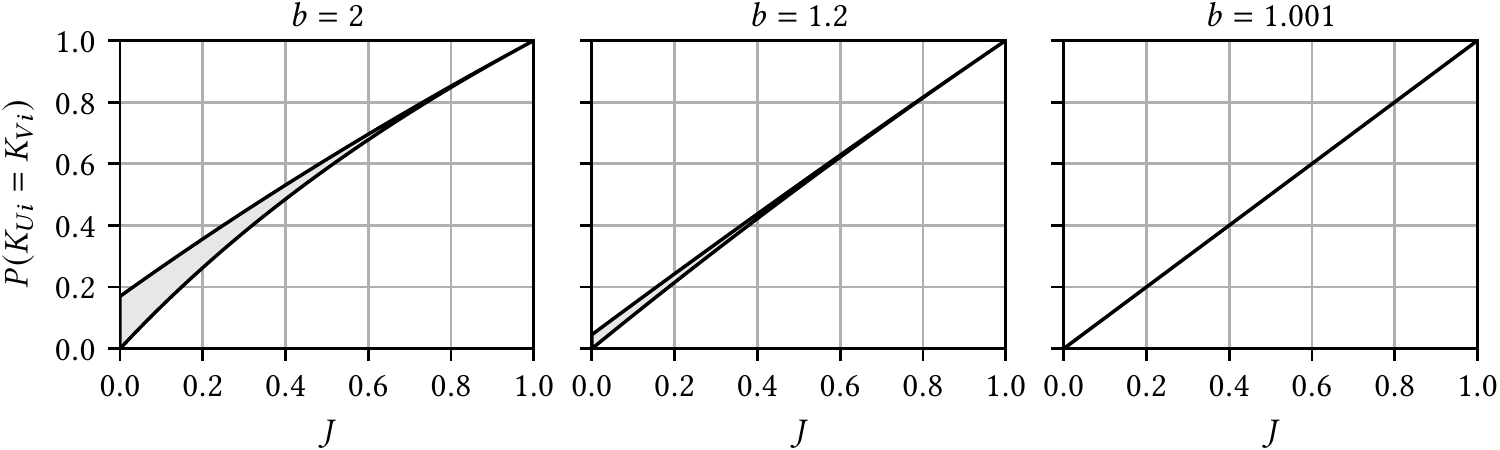}
  \caption{\boldmath The range of possible collision probabilities of SetSketch registers as function of $\symJaccard$.}
  \label{fig:collision_probability}
\end{figure}

\begin{figure}[t]
  \centering
  \includegraphics[width=\linewidth]{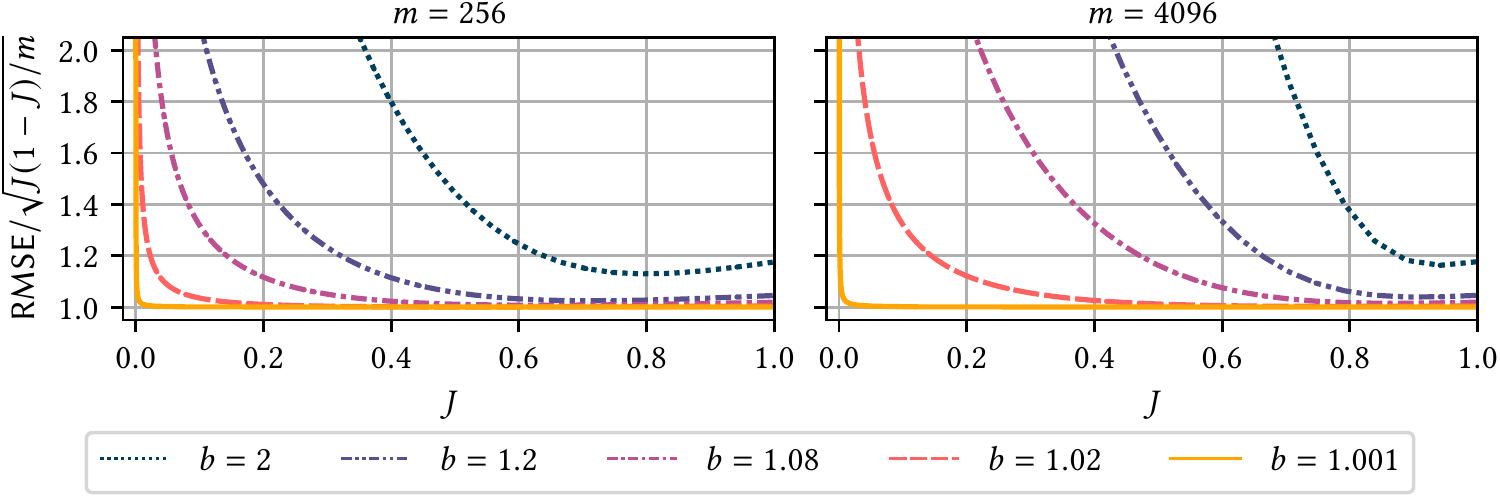}
  \caption{\boldmath The \acs*{RMSE} of $\symJaccardEstimate_\textnormal{\normalfont up}$ for the case $\symCardinalityA = \symCardinalityB$ relative to the \acs*{RMSE} of \acs*{MH} $\sqrt{\symJaccard(1-\symJaccard)/\symNumReg}$.}
  \label{fig:mse_upperbound_estimation.pdf}
\end{figure}

The tight boundaries, especially for large $\symJaccard$, make SetSketch an interesting alternative to \ac{MH} for \ac{LSH}. Like b-bit minwise hashing \cite{Li2010} and odd sketches \cite{Mitzenmacher2014}, SetSketches are more memory efficient than MinHash, but unlike them, SetSketches can be further aggregated. The \ac{RMSE} of \ac{MH} is $\sqrt{\symJaccard(1-\symJaccard)/\symNumReg}$. 
As long as the \ac{RMSE} of SetSketch is not significantly greater than that, the additional error from using SetSketches instead of \ac{MH} is negligible. For comparison, we consider the worst case scenario with $|\symSetA|=|\symSetB|\Leftrightarrow \symCardinalityANorm = \symCardinalityBNorm = \frac{1}{2}$, which maximizes the collision probability, while using $\symJaccardEstimate_\textnormal{up}$, which is based on the minimum possible collision probability, for estimation.
\cref{fig:mse_upperbound_estimation.pdf} shows the theoretical results for different values of $\symBase$ and $\symNumReg$. Even for large values like $\symBase=2$ the \ac{RMSE} is only increased by less than $20\%$, for all similarities greater than $0.7$ or $0.9$ for the cases $\symNumReg=256$ and $\symNumReg=4096$, respectively. The difference becomes smaller with decreasing $\symBase$. The \ac{RMSE} almost matches that of \ac{MH} for $\symBase=1.001$, for which 2-byte registers suffice as mentioned earlier, again showing great potential for space savings compared to \ac{MH}.

When searching for nearest neighbors with \ac{LSH}, a set of candidates is determined first, which is then further filtered. For filtering, the presented more precise joint estimation approach can be used instead of \eqref{equ:jaccard_estimator} to reduce the false positive rate.

\section{Comparison}
This section compares \ac{MH} and \ac{HLL} with SetSketch, and in particular shows that they relate to SetSketches with $\symBase=1$ and $\symBase=2$, respectively. 

\subsection{MinHash}
\label{sec:comp_minhash}
\ac{MH} can be regarded as an extreme case of SetSketch1 with $\symBase\rightarrow 1$ for which the statistic $e^{-\symExponentialRate\symBase^{-\symRegValVariate}}$ approaches a continuous uniform distribution according to \eqref{equ:reg_val_distribution}. Therefore, since the registers of SetSketch1 are also statistically independent, the transformation 
$\symRegValVariate'_\symIndexI = 1 - e^{-\symExponentialRate\symBase^{-\symRegValVariate_\symIndexI}}$
makes SetSketch1 equivalent to \ac{MH} with values $\symRegValVariate'_\symIndexI$ as $\symBase\rightarrow 1$. 
Applying this transformation to \eqref{equ:raw_cardinality_estimator} for $\symBase\rightarrow1$ and using $\lim_{\symBase\rightarrow 1} (1-1/\symBase)/\log(\symBase)=1$ indeed leads to the cardinality estimator of  \ac{MH} \cite{Clifford2012, Cohen2015} with a \ac{RSD} of $1/\sqrt{\symNumReg}$
\begin{equation}
\label{equ:mh_cardinality_estimator}
\textstyle\symCardinalityEstimate = \frac{\symNumReg}{ \sum_{\symIndexI=1}^\symNumReg -\log(1-\symRegValVariate'_\symIndexI)}.
\end{equation}

The presented joint estimation method can also be applied to \ac{MH}. However, as \ac{MH} uses the minimum instead of the maximum for state updates we have to redefine $\symDiffCountVariate_{+}$ and $\symDiffCountVariate_{-}$ as $\symDiffCountVariate_{+} := |\lbrace \symIndexI : \symRegValVariate'_{\symSetA\symIndexI} < \symRegValVariate'_{\symSetB\symIndexI}\rbrace|$, $\symDiffCountVariate_{-} := |\lbrace \symIndexI : \symRegValVariate'_{\symSetA\symIndexI} > \symRegValVariate'_{\symSetB\symIndexI} \rbrace|$. For $\symBase\rightarrow 1$, the \ac{ML} estimate can be explicitly expressed as \ifextended(see \cref{lem:ml_estimate_mh})\else\cite{Ertl2021}\fi
\begin{equation}
\label{equ:new_mh_estimate}
\textstyle
\symJaccardEstimate
=
\frac{
\symCardinalityANorm^2
(\symDiffCountVariate_0+\symDiffCountVariate_{-})
+
\symCardinalityBNorm^2
(\symDiffCountVariate_0+\symDiffCountVariate_{+})
-\sqrt{
\left(
\symCardinalityANorm^2
(\symDiffCountVariate_0+\symDiffCountVariate_{-})
-
\symCardinalityBNorm^2
(\symDiffCountVariate_0+\symDiffCountVariate_{+})
\right)^2
+
4
\symDiffCountVariate_{-}
\symDiffCountVariate_{+}
\symCardinalityANorm^2
\symCardinalityBNorm^2
}
}{
2\symNumReg\symCardinalityANorm
\symCardinalityBNorm
}
\end{equation}
and has an asymptotic \ac{RMSE} of \ifextended(see \cref{lem:fisher_limit})\else\cite{Ertl2021}\fi
\begin{equation*}
\symRMSE(\symJaccardEstimate)
\stackrel{\symNumReg\rightarrow\infty}{=}
\symFisher^{-1/2}(\symJaccard)
= 
\textstyle
\sqrt{\frac{\symJaccard(1-\symJaccard)}{\symNumReg}}
\sqrt{
1-\frac{(\symCardinalityANorm-\symCardinalityBNorm)^2\symJaccard}{\symCardinalityANorm\symCardinalityBNorm(1-\symJaccard)^2}
}
\leq
\sqrt{\frac{\symJaccard(1-\symJaccard)}{\symNumReg}}.
\end{equation*}
This shows that this estimator outperforms the state-of-the-art Jaccard estimator, which has a \ac{RMSE} of $\sqrt{\symJaccard(1-\symJaccard)/\symNumReg}$. Our experiments showed that this estimator also works better, if the cardinalities are unknown and need to be estimated using \eqref{equ:new_mh_estimate}. Therefore, this approach is a much less expensive alternative to the \ac{ML} approach described in \cite{Cohen2017}, which considers the full likelihood as a function of all 3 parameters and requires solving a three-dimensional optimization problem.

The state of SetSketch2 is logically equivalent to that of SuperMinHash as $\symBase\rightarrow 1$. SuperMinHash also uses correlation between components and is able to reduce the variance of the standard estimator for $\symJaccard$ by up to a factor of 2 for small sets \cite{Ertl2017b}.

\subsection{HyperLogLog}
\label{sec:comp_hll}
The \acf{GHLL} with stochastic averaging, which includes \ac{HLL} as a special case with $\symBase=2$, is similar to SetSketch with $\symExponentialRate=1/\symNumReg$. Under the Poisson model \cite{Flajolet2007, Ertl2017}, which assumes that the cardinality $\symCardinality$ is not fixed but Poisson distributed with mean $\symPoissonRate$, the register values will be distributed as $\symProbability(\symRegValVariate_\symIndexI\leq \symRegVal) = e^{-\symPoissonRate\symNumReg^{-1}\symBase^{-\symRegVal}}$ for $\symRegVal\geq 0$ \ifextended(see \cref{lem:poisson_approx})\else\cite{Ertl2021}\fi.
Comparison with \eqref{equ:set_sketch_distribution} shows that the register values of \ac{GHLL} are distributed as those of SetSketch for $\symExponentialRate=1/\symNumReg$ and $\symCardinality=\symPoissonRate$ provided that all register values are nonzero. Therefore, the cardinality estimator \eqref{equ:raw_cardinality_estimator} can be used to estimate $\symPoissonRate$ in this case. An unbiased estimator for the Poisson parameter $\symPoissonRate$ is also an unbiased estimator for the true cardinality $\symCardinality$ \ifextended(see \cref{lem:depoissonization})\else\cite{Ertl2021}\fi. Since \eqref{equ:raw_cardinality_estimator} is asymptotically unbiased as $\symNumReg\rightarrow \infty$, it can be used to estimate $\symCardinality$. And indeed, \eqref{equ:raw_cardinality_estimator} corresponds to the cardinality estimator with a \ac{RSD} of $\sqrt{3\log(2)-1}/\sqrt{\symNumReg} \approx 1.04/\sqrt{\symNumReg}$ presented in \cite{Flajolet2007} for \ac{HLL} with base $\symBase=2$. 

For small cardinalities, when many registers are zero, the value distribution will differ significantly from \eqref{equ:set_sketch_distribution} and the estimator will fail. 
However, it can be fixed by applying the same trick as presented for the case $\symBase=2$ in \cite{Ertl2017}. If registers values are limited to the range $\lbrace 0, 1, \ldots, \symMaxRegularValue+1\rbrace$, and $\symCountVariate_\symRegVal := |\lbrace \symIndexI : \symRegValVariate_\symIndexI = \symRegVal\rbrace|$ is the histogram of register values, the corrected estimator can be written as \ifextended(see \cref{sec:corrected_cardinality_estimator})\else\cite{Ertl2021}\fi
\begin{equation}
\label{equ:corrected_cardinality_estimator}
\symCardinalityCorrectedEstimate = \textstyle\frac{\symNumReg(1-1/\symBase)}{
\symExponentialRate \log(\symBase)
{
\scriptstyle
\left(
\symNumReg\symSmallCorrectionFunc_\symBase(\symCountVariate_0/\symNumReg)
+
(\sum_{\symRegVal=1}^\symMaxRegularValue \symCountVariate_\symRegVal \symBase^{-\symRegVal})
+
\symNumReg
\symBase^{-\symMaxRegularValue}
\symLargeCorrectionFunc_\symBase(1-\symCountVariate_{\symMaxRegularValue+1}/\symNumReg)
\right)}}.
\end{equation}
The functions $\symSmallCorrectionFunc_\symBase$ and $\symLargeCorrectionFunc_\symBase$ are defined as converging series
\begin{align*}
\symSmallCorrectionFunc_\symBase(\symX)
&:=
\textstyle
\symX + (\symBase-1)\sum_{\symRegVal=1}^\infty \symBase^{\symRegVal-1}\symX^{\symBase^\symRegVal} ,
\\
\symLargeCorrectionFunc_\symBase(\symX)
&:=
\textstyle
1-\symX
+
(\symBase-1)
\sum_{\symRegVal=0}^\infty \symBase^{-\symRegVal-1}(\symX^{\symBase^{-\symRegVal}}-1).
\end{align*}
$\symCardinalityCorrectedEstimate$ also includes a correction for very large cardinalities, for which register update values greater than $\symMaxRegularValue+1$, exceeding the value range of registers, are more likely to occur.
The corrected estimator could also be applied to misconfigured SetSketches that do not allow to ignore the probability of registers with values less than 0 or greater than $\symMaxRegularValue+1$ as discussed in \cref{sec:parameter_config}. To save computation time, the values of $\symSmallCorrectionFunc_\symBase(\symCountVariate_0/\symNumReg)$ and $\symLargeCorrectionFunc_\symBase(\symCountVariate_{\symMaxRegularValue+1}/\symNumReg)$ can be tabulated for all possible values of $\symCountVariate_0,\symCountVariate_{\symMaxRegularValue+1}\in\lbrace 0,1,\ldots,\symNumReg\rbrace$ for fixed $\symNumReg$ and $\symBase$.

Due to the similarity of the register value distribution, the proposed joint estimation method also works well for \ac{GHLL}, provided that all register values are from $\lbrace 1,2,\ldots,\symMaxRegularValue\rbrace$. Since the estimator relies only on the relative order, it can be even applied as long as there are no registers having value 0 or $\symMaxRegularValue+1$ simultaneously in both \ac{GHLL} sketches. Registers equal to $\symMaxRegularValue+1$ can be easily avoided by choosing $\symMaxRegularValue$ sufficiently large. Registers, that have never been updated and thus are zero in both sketches, are expected for union cardinalities $|\symSetA\cup\symSetB|$ smaller than $\symNumReg \symHarmonic_\symNumReg$, where $\symHarmonic_\symNumReg:=\sum_{\symIndexI=1}^\symNumReg 1/\symIndexI$ is the $\symNumReg$-th harmonic number. This follows directly from the coupon collector's problem \cite{Mitzenmacher2005}. If the registers do not satisfy the prerequisites for applying the new estimation approach, the inclusion-exclusion principle \eqref{equ:joint_incl_excl} could still be used as fallback. 

\subsection{HyperMinHash}
The similarity with \ac{GHLL} as discussed in \cref{sec:hyperminhash_intro} allows using the proposed joint estimation approach for HyperMinHash as well. If the condition of not too small cardinalities is satisfied, the experiments presented below have shown that the new approach also outperforms the original estimator of HyperMinHash, especially when the cardinalities of the two sets are very different.

\section{Experiments}

For the sake of reproducibility, we used synthetic data for our experiments to verify SetSketch and the proposed estimators. 
The wide and successful application of \ac{MH}, \ac{HLL}, and many other probabilistic data structures has proven that the output of high-quality hash functions \cite{Urban2020} is indistinguishable from uniform random values in practice. This justifies to simply use sets consisting of random 64-bit integers for our experiments. In this way arbitrary many random sets of predefined cardinality can be easily generated, which is fundamental for the rigorous verification of the presented estimators. Real-world data sets usually do not include thousands of different sets with the same predefined cardinality.

\subsection{Implementation}
Both SetSketch variants as well as \ac{GHLL} have been implemented in C++. The corresponding source code, including scripts to reproduce all the presented results, has been made available at \url{https://github.com/dynatrace-research/set-sketch-paper}. We used the Wyrand pseudorandom number generator \cite{Yi2021} which is extremely fast, has a state of 64 bits, and passes a series of statistical quality tests \cite{Lemire2019a}. Seeded with the element of a set, it is able to return a sequence of 64-bit pseudorandom integers. The random bits are used very economically. Only if all 64 bits are consumed, the next bunch of 64 bits will be generated. Sampling with replacement was realized as constant-time operation as described in \cite{Ertl2020} based on Fisher-Yates shuffling \cite{Fisher1938}. The algorithm proposed in \cite{Lemire2019} was used to sample random integers from intervals. The ziggurat method \cite{Marsaglia2000} as implemented in the Boost.Random C++ library \cite{Watanabe2020} was used to obtain exponentially distributed random values. The algorithm described in \cite{Ertl2020} was applied to efficiently sample from a truncated exponential distribution as needed for SetSketch2. 
Our implementation precalculates and stores all relevant powers of $\symBase^{-1}$ in a sorted array, which is then used to find register update values as defined by \eqref{equ:setsketch_update} using binary search instead of expensive logarithm evaluations. This also allows to limit the search space to values greater than $\symRegValVariateLow$, which further saves time with increasing cardinality.

\subsection{Cardinality Estimation}
\label{sec:card_experiment}

\begin{figure}
\centering
\includegraphics[width=\columnwidth]{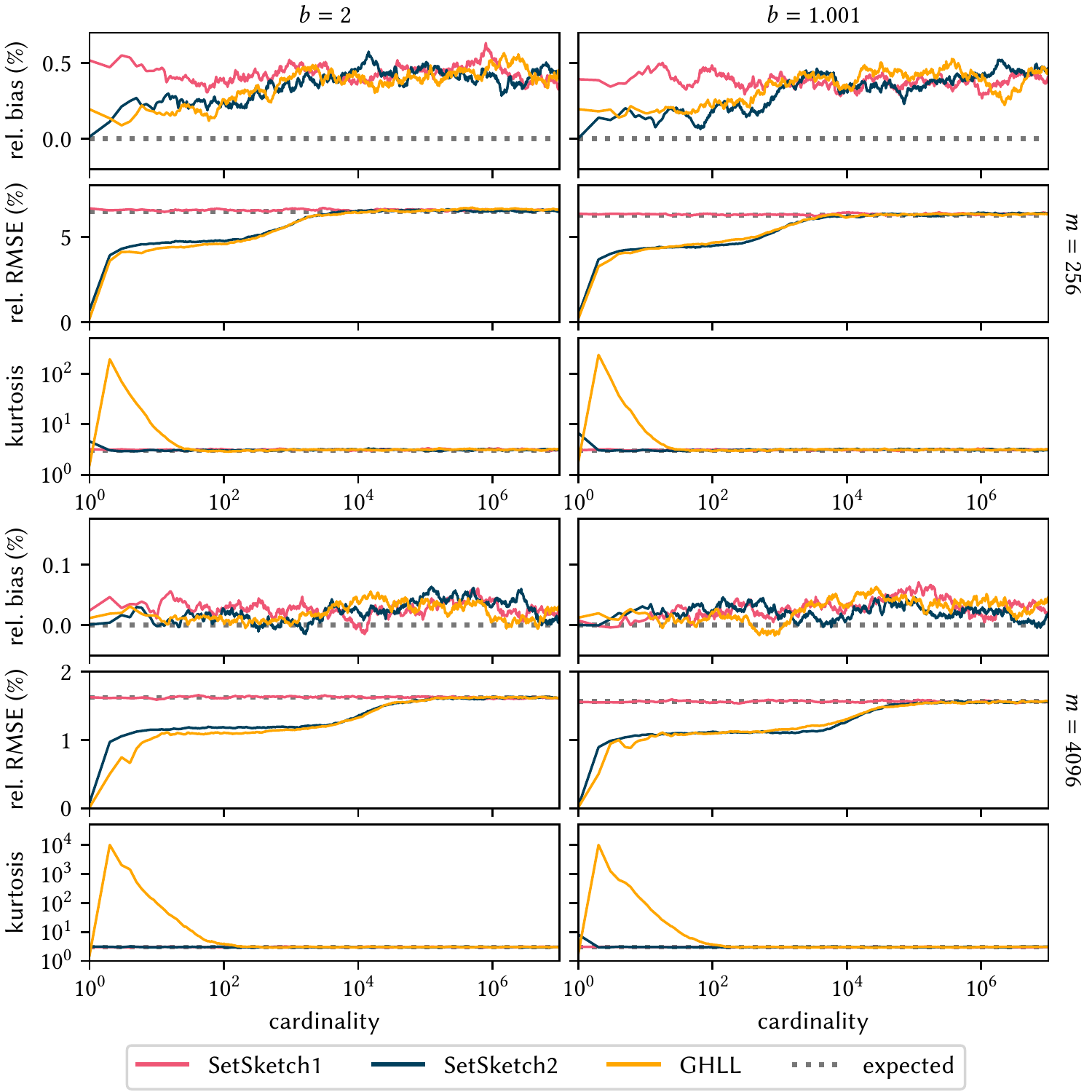}
\caption{\boldmath The relative bias, the relative \acs*{RMSE}, and the kurtosis of $\symCardinalityEstimate$ \eqref{equ:raw_cardinality_estimator} for SetSketch1 and SetSketch2 and of $\symCardinalityCorrectedEstimate$ \eqref{equ:corrected_cardinality_estimator} for \acs*{GHLL}.}
\label{fig:cardinality}
\end{figure}

\begin{figure*}
\centering
\includegraphics[width=\linewidth]{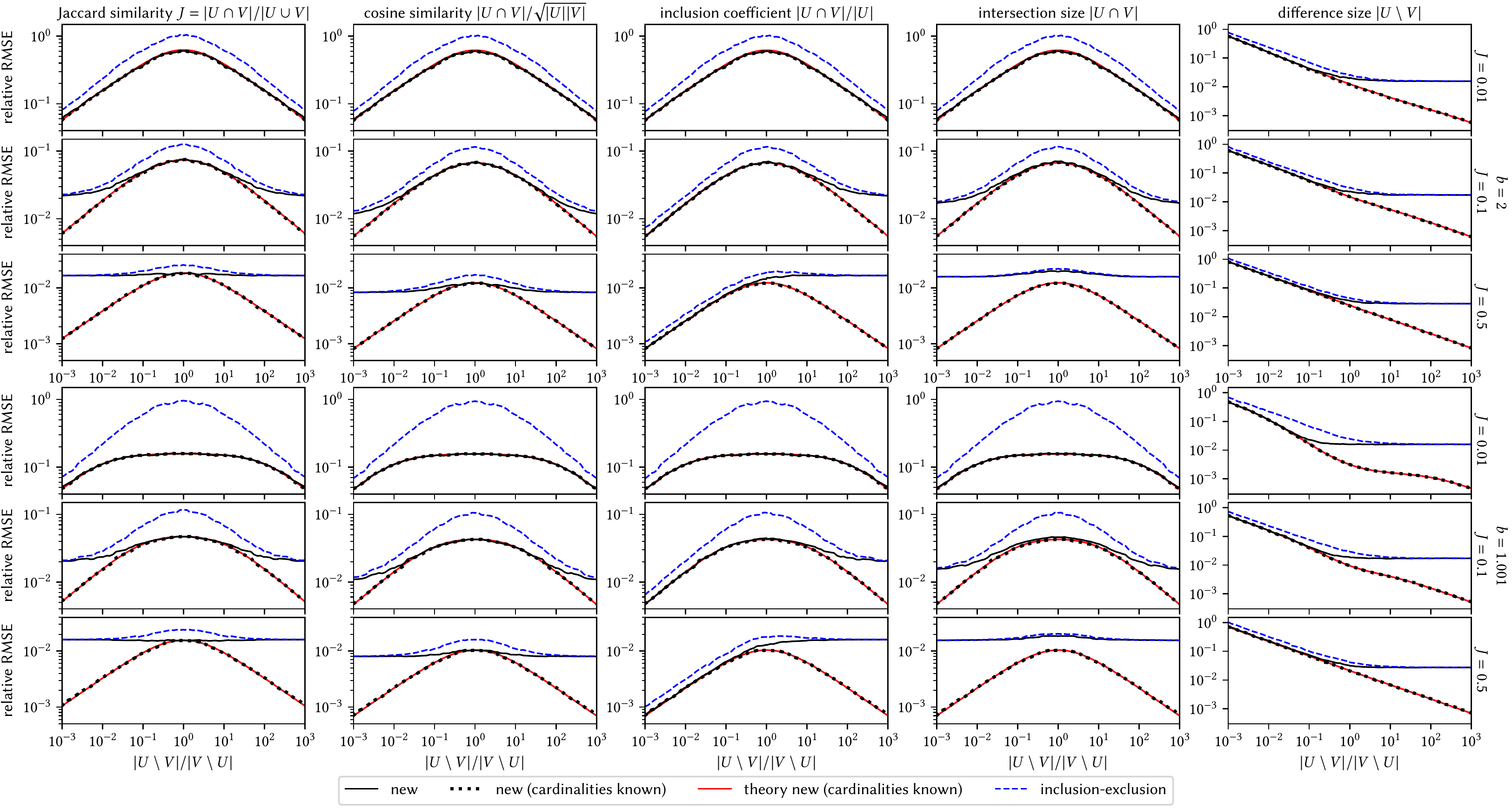}
\caption{\boldmath The relative \acs*{RMSE} of various estimated joint quantities when using SetSketch1 with $\symBase\in\lbrace1.001, 2\rbrace$ and $\symNumReg=4096$ for sets with a fixed union cardinality of $|\symSetA\cup\symSetB|=10^6$.}
\label{fig:joint_setsketch1}
\end{figure*}

\begin{figure*}
\centering
\includegraphics[width=\linewidth]{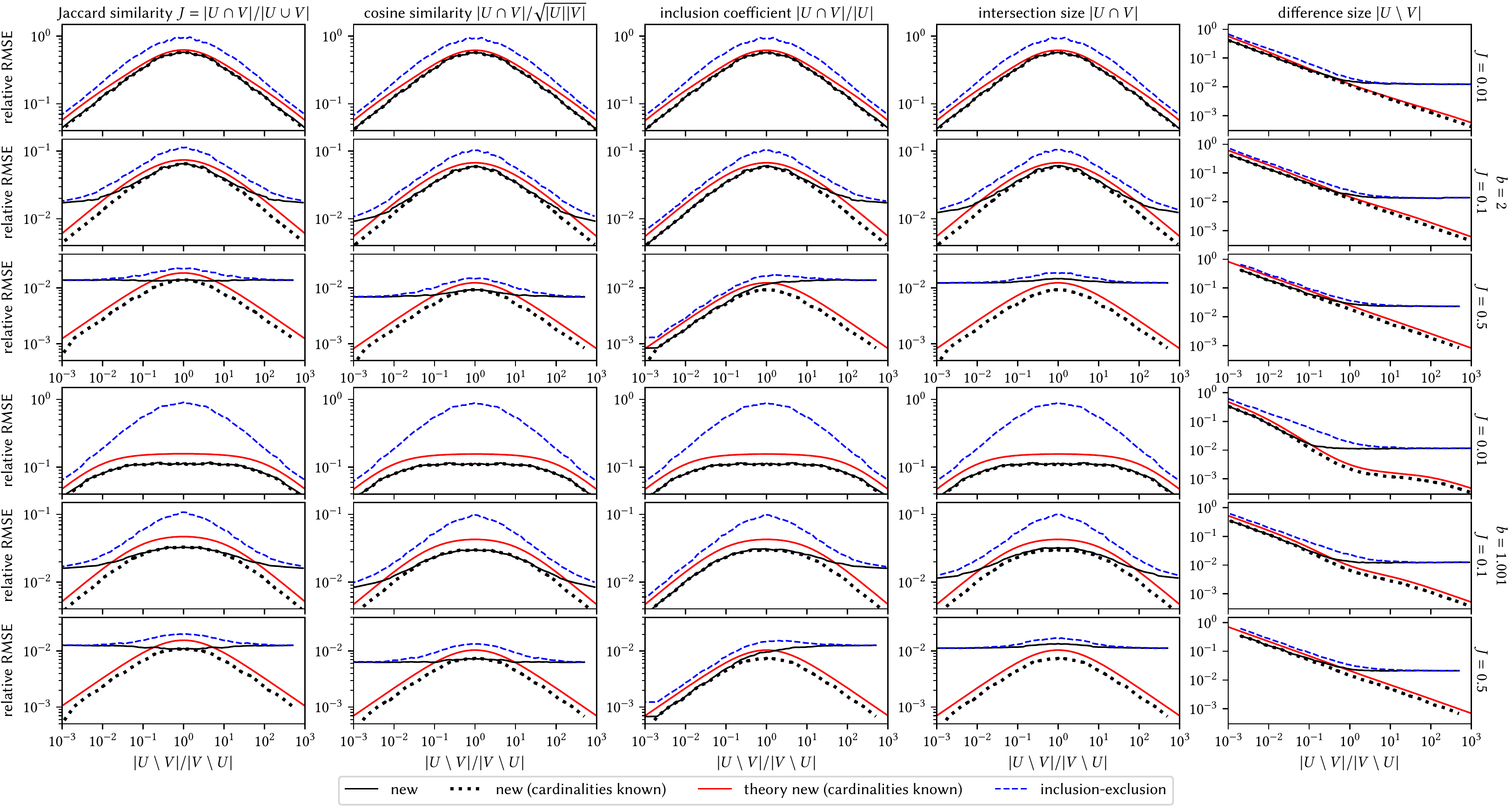}
\caption{\boldmath The relative \acs*{RMSE} of various estimated joint quantities when using SetSketch2 with $\symBase\in\lbrace1.001, 2\rbrace$ and $\symNumReg=4096$ for sets with a fixed union cardinality of $|\symSetA\cup\symSetB|=10^3$.}
\label{fig:joint_setsketch2}
\end{figure*}

\begin{figure*}
\centering
\includegraphics[width=\linewidth]{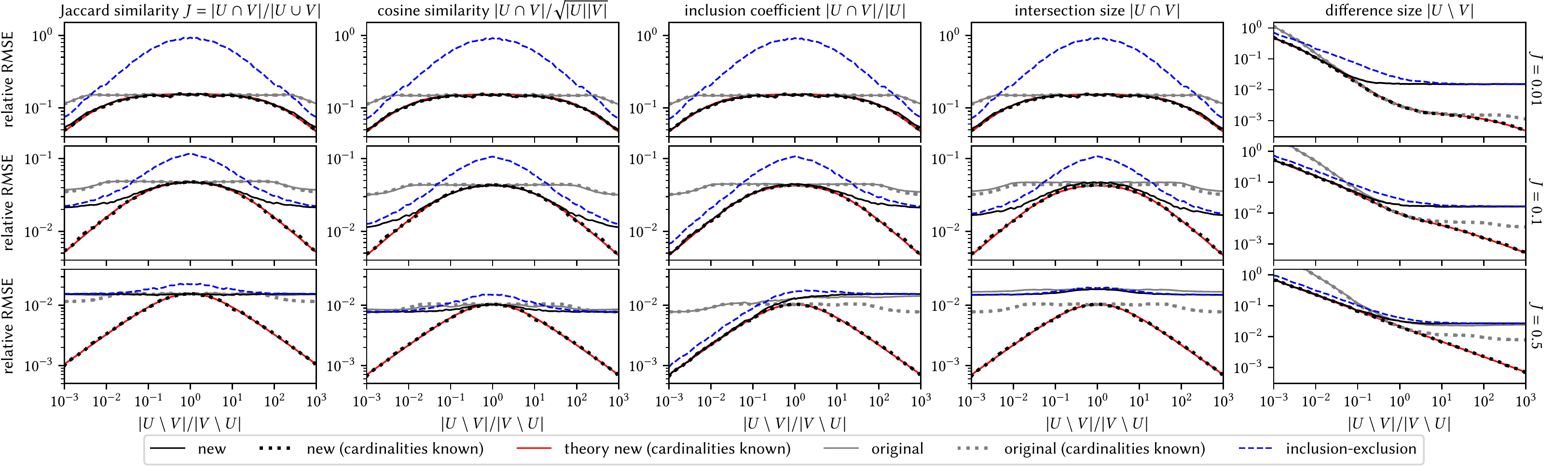}
\caption{\boldmath The relative \acs*{RMSE} of various estimated joint quantities when using \acs*{MH} with $\symNumReg=4096$ for sets with a fixed union cardinality of $|\symSetA\cup\symSetB|=10^6$.}
\label{fig:joint_minhash}
\end{figure*}

\begin{figure*}
\centering
\includegraphics[width=\linewidth]{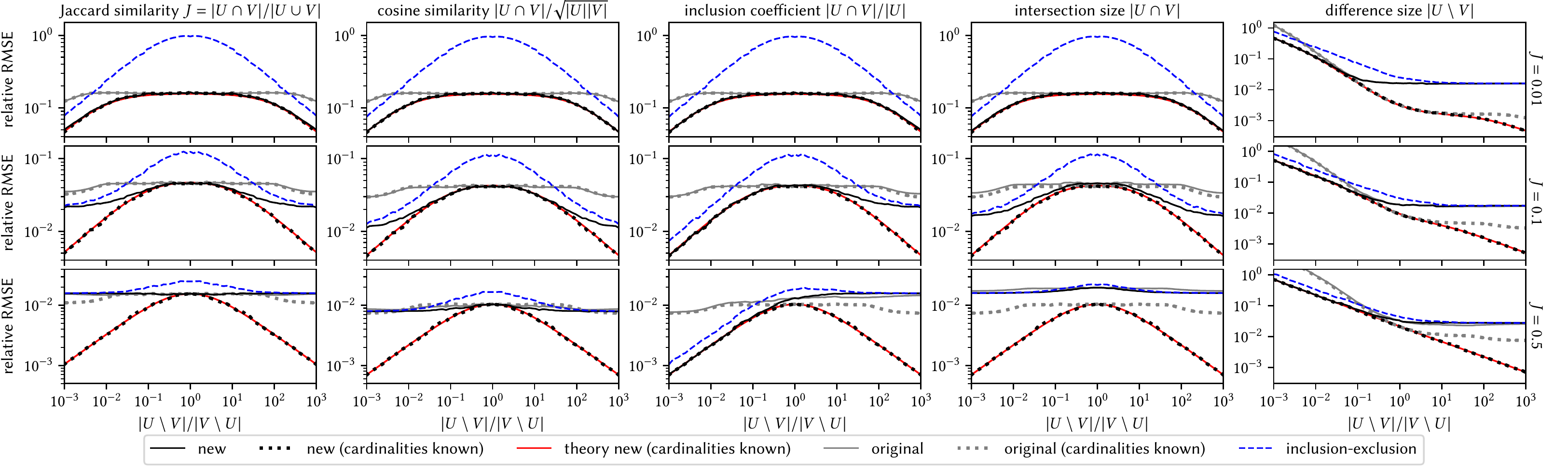}
\caption{\boldmath The relative \acs*{RMSE} of various estimated joint quantities when using HyperMinHash with $\symNumReg=4096$ and $\symHyperMinHashParameter=10$, which corresponds to $\symBase=2^{-2^{10}}\approx 1.000677$, for sets with a fixed union cardinality of $|\symSetA\cup\symSetB|=10^6$.}
\label{fig:joint_hyperminhash}
\end{figure*}

Our cardinality estimation approach was tested by running \num{10000} simulation cycles for different SetSketch and \ac{GHLL} configurations. In each cycle we generated 10 million 64-bit random integer values and inserted them into the data structure. The use of 64 bits makes collisions very unlikely and allows regarding the number of inserted random values as the cardinality of the recorded set. During each simulation cycle the cardinality was estimated for a predefined set of true cardinality values which have been chosen to be roughly equally spaced on the log-scale.

We investigated SetSketch1, SetSketch2, and \ac{GHLL} with stochastic averaging using $\symNumReg=256$ and $\symNumReg=4096$ registers. We considered different bases $\symBase=1.001$ and $\symBase=2$, for which we used 6-bit ($\symMaxRegularValue=62$) and 2-byte registers ($\symMaxRegularValue = 2^{16}-2 = 65534$), respectively. For both SetSketch variants, we used $\symExponentialRate=20$ and \eqref{equ:raw_cardinality_estimator} for cardinality estimation. The corrected estimator \eqref{equ:corrected_cardinality_estimator} was used for \ac{GHLL}. 

\cref{fig:cardinality} shows the relative bias, the relative \ac{RMSE}, and the kurtosis of the cardinality estimates over the true cardinality. The bias is significantly smaller than the \ac{RMSE} in all cases. Furthermore, when comparing $\symNumReg=256$ and $\symNumReg=4096$, the bias decreases more quickly than the \ac{RMSE} as $\symNumReg\rightarrow\infty$. Hence, the bias can be ignored in practice for reasonable values of $\symNumReg$.

The independence of register values in SetSketch1 leads to a constant error over the entire cardinality range, that matches well with the theoretically derived \ac{RSD} from \cref{sec:cardinality_estimation}. In accordance with the discussion in \cref{sec:cardinality_estimation}, the error is only marginally improved when decreasing $\symBase$ from $2$ to $1.001$. The error curves of SetSketch2 and \ac{GHLL} are very similar and show a significant improvement for cardinalities smaller than $\symNumReg$. 
This is attributed to the correlation of register values of SetSketch2 and the stochastic averaging of \ac{GHLL}, respectively, both of which have a similar positive effect. However, when comparing the kurtosis, which is expected to be equal to 3 for normal-like distributions, we observed huge values for \ac{GHLL}. This means that \ac{GHLL} is very susceptible to outliers for small cardinalities. This is not surprising, as stochastic averaging only updates a single register per element. For a set with cardinality $\symCardinality=2$, for example, it happens with $1/\symNumReg$ probability that both elements update the same register, which makes the state indistinguishable from that of a set with just a single element and obviously results in a large \SI{50}{\percent} error. We also tested the \ac{ML} method for cardinality estimation and obtained visually almost identical results \ifextended (see \cref{fig:cardinality_ml}) \else\cite{Ertl2021} \fi which are therefore omitted in \cref{fig:cardinality} and which prove the efficiency of estimators \eqref{equ:raw_cardinality_estimator} and \eqref{equ:corrected_cardinality_estimator}.

\subsection{Joint Estimation}

To verify the presented joint estimation approach, we generated many pairs of random sets with predefined relationship.
Each pair of sets $(\symSetA, \symSetB)$ was constructed by generating three sets $\symSetS_1$, $\symSetS_2$, and $\symSetS_3$ of 64-bit random numbers with fixed cardinalities $\symCardinality_1$, $\symCardinality_2$, and $\symCardinality_3$, respectively, which are merged according to $\symSetA = \symSetS_1\cup \symSetS_3$ and $\symSetB = \symSetS_2\cup \symSetS_3$.
The use of 64-bit random numbers allows ignoring collisions and all three sets can be considered to be distinct. By construction, we have $\symJaccard = \frac{\symCardinality_3}{\symCardinality_1+\symCardinality_2+\symCardinality_3}$, $\symCardinalityA=\symCardinality_1+\symCardinality_3$, and $\symCardinalityB=\symCardinality_2+\symCardinality_3$, which guarantees both cardinalities, the Jaccard similarity, and hence also other joint quantities to be the same for all generated pairs $(\symSetA,\symSetB)$.
After recording both sets using the data structure under consideration and applying the proposed joint estimation approach, we finally compared the estimates with the prescribed true values for various joint quantities.

\cref{fig:joint_setsketch1} shows the relative \ac{RMSE} when estimating the Jaccard similarity, cosine similarity, inclusion coefficients, intersection size, and difference sizes using different approaches from SetSketch1 with $\symExponentialRate=20$ and $\symBase\in\lbrace 1.001,2\rbrace$. The union cardinality was fixed $|\symSetA\cup\symSetB|=\num{e6}$ and the Jaccard similarity was selected from $\symJaccard \in\lbrace 0.01, 0.1, 0.5\rbrace$. Each data point shows the result after evaluating 1000 pairs of randomly generated sets. The charts were obtained by varying the ratio $|\symSetA\setminus\symSetB|/|\symSetB\setminus\symSetA|$ over $[\num{e-3},\num{e3}]$ while keeping $\symJaccard$ and $|\symSetA\cup\symSetB|$ fixed. The symmetry allowed us to perform the experiments only for ratios from $[1,\num{e3}]$.

\cref{fig:joint_setsketch1} clearly shows that the proposed joint estimator dominates the inclusion-exclusion principle for all investigated joint quantities. The difference is more significant for small set overlaps like for $\symJaccard=0.01$. Comparing the results for $\symBase=1.001$ and $\symBase=2$ shows that joint estimation can be significantly improved when using smaller bases $\symBase$. Knowing the true values of $\symCardinalityA$ and $\symCardinalityB$ further reduces the estimation error for which we observed perfect agreement with the theoretically derived \ac{RMSE}. If $\symCardinalityA$ and $\symCardinalityB$ are known, any other joint quantity $\symQuantity$ can be expressed as a function of $\symJaccard$. The corresponding \ac{ML} estimate is therefore given by $\symQuantityEstimate = \symQuantity(\symJaccardEstimate)$. Variable transformation of the Fisher information given in \cref{sec:joint_estimation}, allows to calculate the corresponding asymptotic \ac{RMSE} of $\symQuantity$ which is $\symFisher^{-1/2}(\symJaccard)|\symQuantity'(\symJaccard)|$ as $\symNumReg\rightarrow \infty$.
As also predicted for the case $\symCardinalityA=\symCardinalityB$, the estimation error of $\symJaccard$ is the same regardless of whether the true cardinalities are known or not.

When running the same simulations for SetSketch2 and \ac{GHLL} also with bases $\symBase\in\lbrace 1.001,2\rbrace$, we got almost identical charts \ifextended(see \cref{fig:joint_set_sketch2_1000000} and \cref{fig:joint_ghll_1000000})\else\cite{Ertl2021}\fi, which therefore have been omitted here. The reason why our estimation approach also worked for \ac{GHLL} is that the union cardinality was fixed at \num{e6} which is large enough to not have any registers that are zero in both sketches as discussed in \cref{sec:comp_hll}. In particular, this shows that our approach can significantly improve joint estimation from \ac{HLL} sketches with $\symBase=2$ for which the inclusion-exclusion principle is the state-of-the-art approach for joint estimation. 

We repeated all simulations with a fixed union cardinality of $|\symSetA\cup\symSetB|=\num{e3}$.
For SetSketch1, we obtained the same results as for $|\symSetA\cup\symSetB|=\num{e6}$ \ifextended(see \cref{fig:joint_set_sketch1_1000})\else\cite{Ertl2021}\fi. 
However, the errors for SetSketch2 shown in \cref{fig:joint_setsketch2} are significantly smaller and also lower than theoretically predicted. The reason for this improvement is, as before for cardinality estimation, the statistical dependence between register values in case of small sets. For $\symBase=1.001$, the error is reduced by a factor of up to $\sqrt{2}$ when estimating the Jaccard similarity using our new approach. This is expected as SetSketch2 corresponds to SuperMinHash \cite{Ertl2017b} as $\symBase\rightarrow 1$, for which the variance is known to be approximately 50\% smaller than for \ac{MH}, if $|\symSetA\cup\symSetB|$ is smaller than $\symNumReg$. 
Our joint estimator failed for \ac{GHLL} in this case \ifextended(see \cref{fig:joint_ghll_1000})\else\cite{Ertl2021}\fi, because $|\symSetA\cup\symSetB|=\num{e3}$ is significantly smaller than $\symHarmonic_\symNumReg \symNumReg$ with $\symNumReg=4096$, and hence, the condition for its applicability is not satisfied as discussed in \cref{sec:comp_hll}.

We also applied the new approach to \ac{MH} and used the explicit estimation formula \eqref{equ:new_mh_estimate}. \cref{fig:joint_minhash} only shows the results for $|\symSetA\cup\symSetB|=10^6$, because the results are very similar for $|\symSetA\cup\symSetB|=10^3$, as expected \ifextended(see \cref{fig:joint_minhash_1000})\else\cite{Ertl2021}\fi. The results are also almost indistinguishable from those obtained for SetSketch with $\symBase=1.001$ shown in \cref{fig:joint_setsketch1}. 
Thus, SetSketch is able to give almost the same estimation accuracy using significantly less space as 2-byte registers are sufficient for $\symBase=1.001$. 
We also analyzed the state-of-the-art \ac{MH} estimator based on the fraction of equal components. For the case that the cardinalities are not known, they have been estimated using \eqref{equ:mh_cardinality_estimator}. Our new estimator has a significantly better overall performance in both cases with known and unknown cardinalities, respectively. Only for inclusion coefficients and difference sizes the original method led to a slightly smaller error for $\symJaccard=0.5$ and $|\symSetA\setminus\symSetB|/|\symSetB\setminus\symSetA| > 1$. However, the new estimator clearly dominated, if sets have a small overlap. In contrast to the original estimator, it also dominates the inclusion-exclusion principle in all cases.

Finally, due to its similarity to \ac{GHLL} for which our approach was able to improve joint estimation, we also considered HyperMinHash. \cref{fig:joint_hyperminhash} shows the results for a HyperMinHash with parameter $\symHyperMinHashParameter=10$ which corresponds to a base of $\symBase=2^{-2^{10}}\approx 1.000677$ as discussed in \cref{sec:hyperminhash_intro}. The original estimator of HyperMinHash led to similar results as the original estimator of \ac{MH} (compare \cref{fig:joint_minhash}). 
In contrast to the original estimator, which is based on empirically determined constants, our approach is solely based on theory and never performed worse than the inclusion-exclusion principle. Furthermore, since the estimation error was significantly reduced in many cases, our estimator seems to be superior to the original estimation approach. 
The results for the case that the cardinalities are known show perfect agreement with the theoretical \ac{RMSE} originally derived for SetSketch1. Therefore, at least for large sets, HyperMinHash seems to encode joint information equally well as SetSketch with corresponding base. However, the big advantage of SetSketch is that the same estimator can be applied for any cardinalities, while estimation from \ac{GHLL} or HyperMinHash sketches requires special treatment of small sets. The original HyperMinHash estimator delegates to a second estimator in this case.

\subsection{Performance}
\begin{figure}
  \centering
  \includegraphics[width=\linewidth]{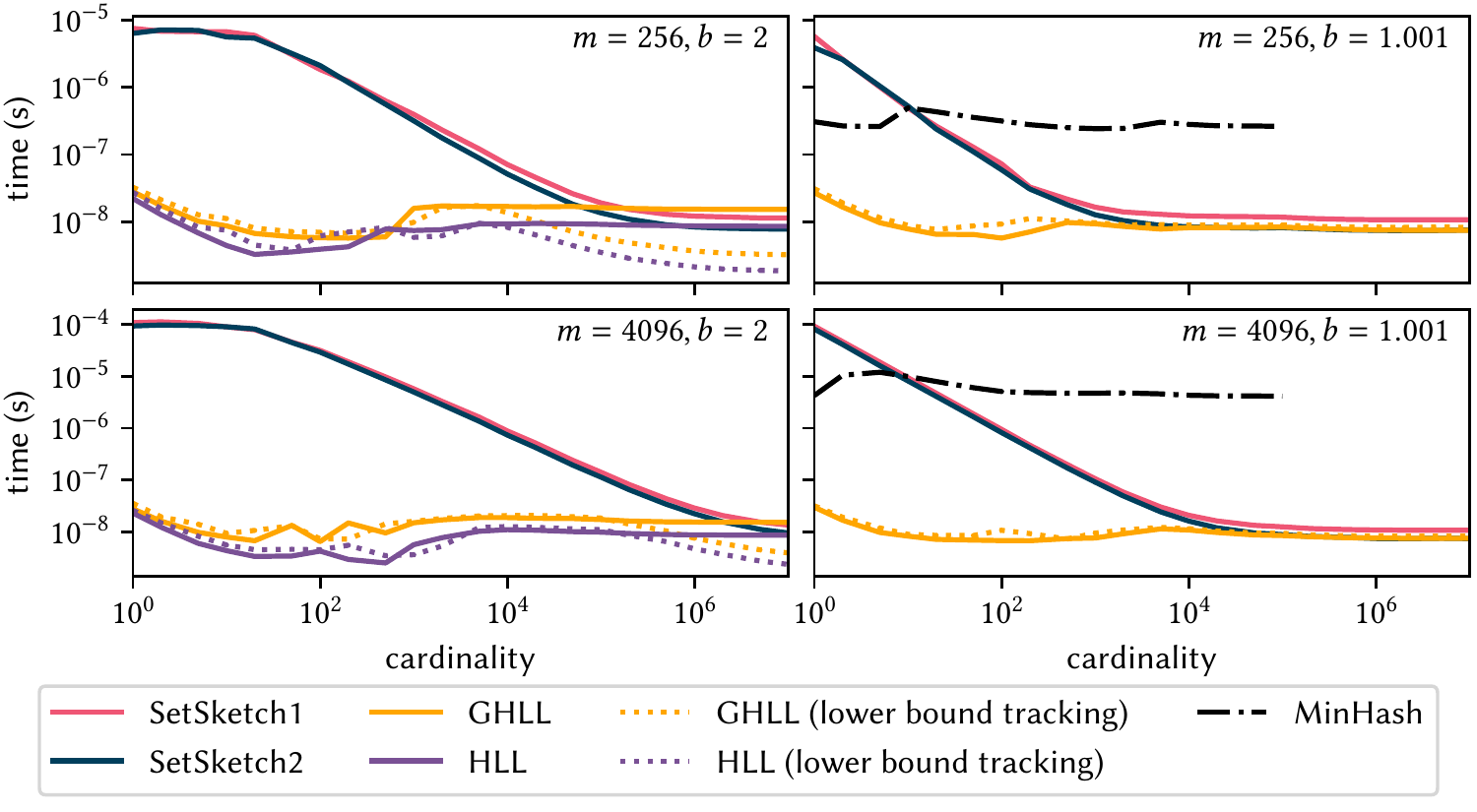}
  \caption{\boldmath The average recording time per element as function of set cardinality.}
  \label{fig:performance}
\end{figure}
The runtime behavior of a sketch is crucial for its practicality. Therefore, we measured the recording time for sets with cardinalities up to $10^7$. Instead of generating a set first, we simply generated 64-bit random numbers on the fly using the very fast Wyrand generator \cite{Yi2021}. As hashing of more complex items is usually more expensive than generating random values, this experimental setup amplifies the runtime differences compared to reality, where elements also need to be loaded from main memory. Furthermore, we excluded initialization costs, which are comparable for all considered data structures and which are a negligible factor for large cardinalities.
For each cardinality considered, we performed 1000 simulation runs. In each run, we measured the time needed to generate the corresponding number of random elements and to insert them into the data structure. Afterwards, we calculated the average recording time per element. We measured the performance for SetSketch1, SetSketch2, \ac{GHLL} with $\symBase\in\lbrace 1.001, 2 \rbrace$, \ac{MH}, and \ac{HLL} and sketch sizes $\symNumReg\in\lbrace 256,4096\rbrace$ on a Dell Precision 5530 notebook with an Intel Core i9-8950HK processor.

\cref{fig:performance} summarizes the results. The recording speed was roughly independent of the set size for \ac{HLL} and \ac{GHLL} due to stochastic averaging. We also implemented variants of both algorithms that use lower bound tracking as described in \cref{sec:lower_bound_tracking}. If update values are not greater than the current lower bound, they will not be able to change any register. This can avoid many relatively costly random accesses to individual registers. This optimization, which is simpler than other approaches \cite{Ertl2017, Reviriego2020}, led to a significant performance improvement for \ac{HLL} and \ac{GHLL} with $\symBase=2$ at large cardinalities. However, this optimization did not speed up recording for \ac{GHLL} with $\symBase=1.001$, because the calculation of register update values is the bottleneck here as it is more expensive for small $\symBase$.

For \ac{MH} the recording time was also independent of the cardinality, but many orders of magnitude slower due to its $\symBigO(\symNumReg)$ insert operation, which was also the reason why we only simulated sets up to a size of $10^5$.
The insert operations of both SetSketch variants have almost identical performance characteristics. As expected, insertions are quite slow for small sets. However, with increasing cardinality, the tracked lower bound $\symRegValVariateLow$ will also increase and finally leads to better and better recording speeds. For large sets, SetSketch is several orders of magnitude faster than \ac{MH} and SetSketch2 even achieves the performance of the non-optimized (without lower bound tracking) versions of \ac{HLL} and \ac{GHLL}. We observed a quicker decay for $\symBase=1.001$ than for $\symBase=2$. The reason is that register values are updated more frequently for smaller $\symBase$, which allows $\symRegValVariateLow$ to be raised earlier.

\section{Conclusion}
We have presented a new data structure for sets that combines the properties of \ac{MH} and \ac{HLL} and allows fine-tuning between estimation accuracy and memory efficiency. The presented estimators for cardinality and joint quantities do not require empirical calibration, give consistent errors over the full cardinality range without having to consider special cases such as small sets, and can be evaluated in a numerically robust fashion. The simple estimation from SetSketches, plus locality sensitivity as a bonus, compensate for the slower recording speed compared to \ac{HLL} and \ac{GHLL} for small sets. The developed joint estimator can also be straightforwardly applied to existing \ac{MH}, \ac{HLL}, \ac{GHLL}, and HyperMinHash data structures to obtain more accurate results for joint quantities than with the corresponding state-of-the-art estimators in many cases. We expect that our estimation approach will also work for other set similarity measures or joint quantities that have not been covered by our experiments.

\bibliographystyle{ACM-Reference-Format}

\ifextended

\appendix 

\section{Proofs}

\begin{lemma}
\label{lem:max_binary_op}
The only commutative ($\symBinaryOperation(\symX,\symY) = \symBinaryOperation(\symY,\symX)$) and idempotent ($\symBinaryOperation(\symX,\symX)=\symX$) binary operation $\symBinaryOperation$ on $\mathbb{R}$, that satisfies $\symBinaryOperation(\symX,\symY)\geq \symX$ and $\symY\geq\symZ\Rightarrow\symBinaryOperation(\symX,\symY)\geq\symBinaryOperation(\symX,\symZ)$, is the maximum function $\symBinaryOperation(\symX,\symY) = \max(\symX,\symY)$.
\end{lemma}
\begin{proof}
Assume first $\symX\geq\symY$. Then we have $\symX = \symBinaryOperation(\symX,\symX) \geq \symBinaryOperation(\symX, \symY) \geq \symX$. Therefore $\symBinaryOperation(\symX, \symY) = \symBinaryOperation(\symY, \symX) = \symX$. For the case $\symX\leq\symY$ we analogously get $\symBinaryOperation(\symY, \symX) = \symBinaryOperation(\symX, \symY) = \symY$. Hence, combining both cases gives $\symBinaryOperation(\symX,\symY) = \max(\symX,\symY)$.
\end{proof}

\begin{lemma}
\label{lem:func_equation}
Any nonconstant function $\symDistributionFunc:\mathbb{Z}\times(0,\infty) \rightarrow[0,1]$, that satisfies $\symDistributionFunc(\symRegVal; \symCardinality_1 + \symCardinality_2) 
=
\symDistributionFunc(\symRegVal; \symCardinality_1) 
\cdot
\symDistributionFunc(\symRegVal; \symCardinality_2)$
and
$\symDistributionFunc(\symRegVal;\symCardinality) = \symDistributionFunc(\symRegVal + 1;\symCardinality\symBase)$
for all 
$\symRegVal\in\mathbb{Z}$ and $\symCardinality, \symCardinality_1, \symCardinality_2\in\mathbb{R}_{>0}$
and constant $\symBase>1$, has the shape $\symDistributionFunc(\symRegVal;\symCardinality) = e^{-\symCardinality\symExponentialRate \symBase^{-\symRegVal}}$ with some constant $\symExponentialRate> 0$.
\end{lemma}
\begin{proof}
Setting $\symRegVal=0$ gives $\symDistributionFunc(0; \symCardinality_1 + \symCardinality_2) 
=
\symDistributionFunc(0; \symCardinality_1) 
\cdot
\symDistributionFunc(0; \symCardinality_2)$. This corresponds, to the exponential Cauchy equation, which has either the solution $\symDistributionFunc(0; \symCardinality) = 0$ or $\symDistributionFunc(0; \symCardinality) = e^{-\symCardinality\symExponentialRate}$. The constant $\symExponentialRate$ must be nonnegative, because $\symDistributionFunc$ has domain $[0,1]$. Repeated application of the second equation yields $\symDistributionFunc(\symRegVal;\symCardinality)=\symDistributionFunc(0;\symCardinality \symBase^{-\symRegVal})$. Therefore, the only potential nonconstant solutions are given by $\symDistributionFunc(\symRegVal;\symCardinality)=e^{-\symCardinality\symExponentialRate\symBase^{-\symRegVal}}$ with $\symExponentialRate>0$. It can be easily verified that these are indeed solutions of the given system of equations. 
\end{proof}

\begin{lemma}
\label{lem:setsketch2}
If $\symStatisticX\sim\symExponential(\symExponentialRate)$ is exponentially distributed with rate $\symExponentialRate$, $\symNumReg\in\mathbb{N}$, and $\symGamma_\symIndexJ:=\frac{1}{\symExponentialRate}\log(1 + \frac{\symIndexJ}{\symNumReg-\symIndexJ})$, the probability that $\symStatisticX\in[\symGamma_{\symIndexJ-1},\symGamma_{\symIndexJ})$ with $1\leq \symIndexJ \leq \symNumReg$ is equal to $\frac{1}{\symNumReg}$.
\end{lemma}
\begin{proof}
Using the cumulative distribution function of the exponential distribution $\symProbability(\symStatisticX<\symX) = 
1 - e^{-\symExponentialRate \symX}$ we get
\begin{multline*}
\symProbability(\symStatisticX\in[\symGamma_{\symIndexJ-1},\symGamma_{\symIndexJ}))
=
\symProbability(\symStatisticX<\symGamma_{\symIndexJ}) - \symProbability(\symStatisticX<\symGamma_{\symIndexJ-1})
\\
\begin{aligned}
&=
(1 - e^{-\symExponentialRate \symGamma_{\symIndexJ}})
-
(1 - e^{-\symExponentialRate \symGamma_{\symIndexJ-1}})
=
e^{-\symExponentialRate \symGamma_{\symIndexJ-1}}
-
e^{-\symExponentialRate \symGamma_{\symIndexJ}}
\\
&=
\frac{1}{1 + \frac{\symIndexJ-1}{\symNumReg-\symIndexJ+1}}
-\frac{1}{1 + \frac{\symIndexJ}{\symNumReg-\symIndexJ}}
=
\frac{\symNumReg-\symIndexJ+1}{\symNumReg}
-\frac{\symNumReg-\symIndexJ}{\symNumReg}
=
\frac{1}{\symNumReg}.
\end{aligned}
\end{multline*}
\end{proof}

\begin{lemma}
\label{lem:singleton}
If $\symExponentialRate \geq \log(\symNumReg/\symSmallProbability)/\symBase$ with $\symSmallProbability>0$, the probability, that any register value of a SetSketch representing some nonempty set is negative, is bounded by $\symSmallProbability$, hence 
$\symProbability(\min\nolimits_\symIndexI \symRegValVariate_\symIndexI < 0) \leq \symSmallProbability$.
\end{lemma}
\begin{proof}
By definition, the register values are smallest when the data structure represents just a set with a single element. 
Therefore, and because $\symExponentialRate \geq \log(\symNumReg/\symSmallProbability)/\symBase\Leftrightarrow \symNumReg e^{-\symExponentialRate\symBase} \leq \symSmallProbability$ it is sufficient to show that 
$\symProbability(\min\nolimits_\symIndexI \symRegValVariate_\symIndexI < 0 \mid \symCardinality = 1) \leq \symNumReg e^{-\symExponentialRate\symBase}$
holds for SetSketch1 and SetSketch2.

For SetSketch1 with $\symCardinality = 1$, the register values $\symRegValVariate_\symIndexI$ are independent and distributed according to \eqref{equ:reg_val_distribution} as $\symProbability(\symRegValVariate_\symIndexI \leq \symRegVal\mid \symCardinality = 1) = e^{-\symExponentialRate\symBase^{-\symRegVal}}$. Therefore,
$
\symProbability(\min\nolimits_\symIndexI \symRegValVariate_\symIndexI < 0\mid \symCardinality = 1)
=
1 - \symProbability(\min\nolimits_\symIndexI \symRegValVariate_\symIndexI \geq 0\mid \symCardinality = 1)
=
1 - (\symProbability(\symRegValVariate_\symIndexI \geq 0\mid \symCardinality = 1))^\symNumReg
=
1 - (1-\symProbability(\symRegValVariate_\symIndexI \leq -1\mid \symCardinality = 1))^\symNumReg
=
1 - (1-e^{-\symExponentialRate\symBase})^{\symNumReg}
\leq
1 - (1-\symNumReg e^{-\symExponentialRate\symBase})
=
\symNumReg e^{-\symExponentialRate\symBase}.
$
Here we used Bernoulli's inequality $(1+\symX)^\symNumReg\geq 1+\symNumReg\symX$ with $\symX = -e^{-\symExponentialRate\symBase}\geq -1$.

For SetSketch2 with $\symCardinality = 1$, the smallest register value is distributed according to \eqref{equ:correlated} as 
$\min\nolimits_\symIndexI \symRegValVariate_\symIndexI \sim \lfloor 1 - \log_\symBase \symPoint_\symNumReg\rfloor$
with $\symPoint_\symNumReg\sim \symExponential(\symExponentialRate;\symGamma_{\symNumReg-1},\symGamma_{\symNumReg})\sim\symExponential(\symExponentialRate;\log(\symNumReg)/\symExponentialRate,\infty)\Leftrightarrow \symExponentialRate\symPoint_\symNumReg-\log\symNumReg\sim\symExponential(1)$. Hence,
$
\symProbability(\min\nolimits_\symIndexI \symRegValVariate_\symIndexI < 0\mid \symCardinality = 1)
=
\symProbability(\lfloor 1 - \log_\symBase \symPoint_\symNumReg\rfloor < 0)
=
\symProbability(1 - \log_\symBase \symPoint_\symNumReg < 0)
=
\symProbability(\symPoint_\symNumReg > \symBase)
=
\symProbability(\symExponentialRate\symPoint_\symNumReg-\log\symNumReg > \symExponentialRate\symBase-\log\symNumReg)
=
\min(1, \symNumReg^{-\symExponentialRate\symBase})
\leq
\symNumReg e^{-\symExponentialRate\symBase}
$
where we used that $\symStatisticX = \symExponentialRate\symPoint_\symNumReg-\log\symNumReg$ is exponentially distributed with rate 1 and that $\symProbability(\symStatisticX>\symX) = \min(1, e^{-\symX})$ for all $\symX\in\mathbb{R}$.
\end{proof}

\begin{lemma}
\label{lem:max}
If $\symMaxRegularValue \geq \lfloor \log_\symBase \frac{\symNumReg\symCardinality_\textnormal{max}\symExponentialRate}{\symSmallProbability}\rfloor$ with $\symSmallProbability>0$, the probability, that any register value of a SetSketch representing a set with a maximum cardinality of $\symCardinality_\textnormal{max}\geq 1$ is greater than $\symMaxRegularValue+1$, is bounded by $\symSmallProbability$, hence
$\symProbability(\max\nolimits_\symIndexI \symRegValVariate_\symIndexI > \symMaxRegularValue+1) \leq \symSmallProbability$.
\end{lemma}
\begin{proof}
It is sufficient to prove the statement for the extreme case where the cardinality $\symCardinality$ equals $\symCardinality_\textnormal{max}$. $\symMaxRegularValue \geq \lfloor \log_\symBase \frac{\symNumReg\symCardinality_\textnormal{max}\symExponentialRate}{\symSmallProbability}\rfloor$ implies $\symMaxRegularValue + 1 > \log_\symBase \frac{\symNumReg\symCardinality_\textnormal{max}\symExponentialRate}{\symSmallProbability}$ and further $\symCardinality_\textnormal{max}\symNumReg\symExponentialRate\symBase^{-\symMaxRegularValue-1} < \symSmallProbability$. 
Since the claimed statement is obvious for $\symSmallProbability\geq 1$, we can assume $\symSmallProbability<1$ and hence $\symNumReg\symExponentialRate\symBase^{-\symMaxRegularValue-1} < \frac{\symSmallProbability}{\symCardinality_\textnormal{max}}< 1$.
First, we consider $\symProbability(\max\nolimits_\symIndexI \symRegValVariate_\symIndexI > \symMaxRegularValue+1\mid \symCardinality = 1)$ for a set with cardinality 1. 

For SetSketch1, this probability is given by 
$
\symProbability(\max\nolimits_\symIndexI \symRegValVariate_\symIndexI > \symMaxRegularValue+1\mid \symCardinality = 1)
=
1 - \symProbability(\max\nolimits_\symIndexI \symRegValVariate_\symIndexI \leq \symMaxRegularValue+1\mid \symCardinality = 1)
=
1 - \symProbability(\symRegValVariate_\symIndexI \leq \symMaxRegularValue+1\mid\symCardinality=1)^{\symNumReg}
=
1 - (e^{-\symExponentialRate\symBase^{-\symMaxRegularValue-1}})^\symNumReg
=
1 - e^{-\symNumReg\symExponentialRate\symBase^{-\symMaxRegularValue-1}}
\leq
\symNumReg\symExponentialRate\symBase^{-\symMaxRegularValue-1}
<
\frac{\symSmallProbability}{\symCardinality_\textnormal{max}}
$ where we used the inequality $1 - e^{-\symX} \leq \symX$ for $\symX\geq 0$. 

We can find the same upper bound for SetSketch2. For $\symCardinality=1$ $\max\nolimits_\symIndexI \symRegValVariate_\symIndexI$ is distributed according to \eqref{equ:correlated} as
$\max\nolimits_\symIndexI \symRegValVariate_\symIndexI \sim \lfloor 1 - \log_\symBase \symPoint_1\rfloor$
with $\symPoint_1\sim \symExponential(\symExponentialRate;\symGamma_{0},\symGamma_{1})\sim\symExponential(\symExponentialRate;0,\frac{1}{\symExponentialRate}\log(1+ \frac{1}{\symNumReg-1}))$. The probability density of this truncated exponential distribution with support $[0, \frac{1}{\symExponentialRate}\log(1+ \frac{1}{\symNumReg-1})]$ is given by $\symDensity_{\symPoint_1}(\symX) = \symNumReg\symExponentialRate e^{-\symExponentialRate\symX}$. Hence, 
$\symProbability(\max\nolimits_\symIndexI \symRegValVariate_\symIndexI>\symMaxRegularValue+1\mid\symCardinality=1)
=
\symProbability(\lfloor 1 - \log_\symBase \symPoint_1\rfloor > \symMaxRegularValue+1)
=
\symProbability(\symPoint_1 \leq \symBase^{-\symMaxRegularValue-1})
=
\min(1, \int_0^{\symBase^{-\symMaxRegularValue-1}}\symDensity_{\symPoint_1}(\symX) d\symX)
=
\min(1, \symNumReg(1-e^{-\symExponentialRate\symBase^{-\symMaxRegularValue-1}}))
\leq
\symNumReg(1-e^{-\symExponentialRate\symBase^{-\symMaxRegularValue-1}})
\leq
\symNumReg\symExponentialRate\symBase^{-\symMaxRegularValue-1}
<
\frac{\symSmallProbability}{\symCardinality_\textnormal{max}}
$. Here we used again the inequality $1 - e^{-\symX} \leq \symX$ for $\symX\geq 0$.

For a set with cardinality $\symCardinality_\textnormal{max}$ the probability that $\max\nolimits_\symIndexI \symRegValVariate_\symIndexI > \symMaxRegularValue+1$
can therefore be bounded by
$\symProbability(\max\nolimits_\symIndexI \symRegValVariate_\symIndexI > \symMaxRegularValue+1\mid\symCardinality=\symCardinality_\textnormal{max})
=
1 - \symProbability(\max\nolimits_\symIndexI \symRegValVariate_\symIndexI \leq \symMaxRegularValue+1\mid\symCardinality=\symCardinality_\textnormal{max})
=
1 - (\symProbability(\max\nolimits_\symIndexI \symRegValVariate_\symIndexI \leq \symMaxRegularValue+1\mid\symCardinality=1))^{\symCardinality_\textnormal{max}}
=
1-(1-\symProbability(\max\nolimits_\symIndexI \symRegValVariate_\symIndexI > \symMaxRegularValue+1\mid\symCardinality=1))^{\symCardinality_\textnormal{max}}
<
1  - (1 - \frac{\symSmallProbability}{\symCardinality_\textnormal{max}})^{\symCardinality_\textnormal{max}}
\leq
1  - (1 - \symCardinality_\textnormal{max}\frac{\symSmallProbability}{\symCardinality_\textnormal{max}})
=
\symSmallProbability
$ for both SetSketch variants when using Bernoulli's inequality $(1+\symX)^\symCardinality\geq 1+\symCardinality\symX$ with $\symX = -\frac{\symSmallProbability}{\symCardinality_\textnormal{max}} > -1$.
\end{proof}

\begin{lemma}
\label{lem:xi}
The function 
$\symPowerSeriesFunc_\symBase^\symIntPower(\symX):= 
\frac{\log \symBase}{\Gamma(\symIntPower)} \sum_{\symRegVal = -\infty}^\infty
\symBase^{\symIntPower (\symX-\symRegVal)}
e^{-\symBase^{\symX-\symRegVal}}
$ with $\symBase>1$ and $\symIntPower>0$ is periodic with period 1 and oscillates around 1. Its Fourier series is 
$\symPowerSeriesFunc_\symBase^\symIntPower(\symX)= 
 1 + \frac{2}{\Gamma(\symIntPower)}\operatorname{Re}\!\left(\sum_{\symIndexL=1}^\infty \Gamma\!\left(\symIntPower - \textstyle\frac{\symImaginary 2\pi \symIndexL}{\log \symBase}\right)e^{ \symImaginary 2\pi\symIndexL \symX}\right)$ where $\Gamma$ denotes the gamma function.
\end{lemma}
\begin{proof}
The periodicity follows directly from  $\symPowerSeriesFunc_\symBase^\symIntPower(\symX) = \symPowerSeriesFunc_\symBase^\symIntPower(\symX + 1)$. 
Therefore, the Fourier series can be written as 
\begin{equation*}
\symPowerSeriesFunc_\symBase^\symIntPower(\symX)
=
\frac{\symFourierCoefficient_0}{2}
+
\operatorname{Re}\!\left(
\sum_{\symIndexL = 1}^\infty
\symFourierCoefficient_\symIndexL
e^{2\pi\symImaginary\symIndexL\symX}
\right)
\quad
\end{equation*}
with coefficients
\begin{align*}
\symFourierCoefficient_\symIndexL &= 
2\int_0^1
\symPowerSeriesFunc_\symBase^\symIntPower(\symX)
e^{-2\pi\symImaginary\symIndexL\symX}
d\symX
\\
&=
\frac{2\log \symBase}{\Gamma(\symIntPower)} 
\int_0^1
\sum_{\symRegVal = -\infty}^\infty
\symBase^{\symIntPower (\symX-\symRegVal)}
e^{-\symBase^{\symX-\symRegVal}}
e^{-2\pi\symImaginary\symIndexL\symX}
d\symX
\\
&=
\frac{2\log \symBase}{\Gamma(\symIntPower)} 
\int_{-\infty}^{\infty}
\symBase^{\symIntPower \symX}
e^{-\symBase^{\symX}}
e^{-2\pi\symImaginary\symIndexL\symX}
d\symX
\\
&=
\frac{2\log \symBase}{\Gamma(\symIntPower)} 
\int_{-\infty}^{\infty}
\symBase^{\symX\left(\symIntPower -\textstyle\frac{\symImaginary 2\pi \symIndexL}{\log \symBase}\right)}
e^{-\symBase^{\symX}}
d\symX.
\end{align*}
Substitution $\symBase^\symX=\symY$ finally gives
\begin{equation*}
\symFourierCoefficient_\symIndexL
= 
\frac{2}{\Gamma(\symIntPower)} 
\int_{0}^{\infty}
\symY^{\symIntPower-\frac{\symImaginary 2\pi\symIndexL}{\log\symBase}-1}
e^{-\symY}
d\symY
=
\frac{2}{\Gamma(\symIntPower)} 
\Gamma\!\left(\symIntPower - \textstyle\frac{\symImaginary 2\pi \symIndexL}{\log \symBase}\right)
\end{equation*}
where we used the definition of the gamma function 
\begin{equation*}
\Gamma(\symZ):=\int_0^\infty \symY^{\symZ-1}e^{-\symY}d\symY.
\end{equation*}
\end{proof}

\begin{lemma}
\label{lemma:inquality_sinh}
$\frac{\symY\symIndexL}{\sinh(\symY\symIndexL)}\leq\left(\frac{\symY}{\sinh(\symY)}\right)^\symIndexL$ holds for all integers $\symIndexL \geq 1$ and all $\symY\in\mathbb{R}\setminus\lbrace 0\rbrace$.
\end{lemma}
\begin{proof}
We use induction. The case $\symIndexL=1$ is trivial. Assume that the inequality holds for $\symIndexL=\symIndexK$, which is used together with $1\leq\frac{\symY}{\tanh(\symY)}$ to prove the inequality for $\symIndexL = \symIndexK+1$:
\begin{multline*}
\frac{\symY(\symIndexK+1)}{\sinh(\symY(\symIndexK+1))}
=
\frac{\symY\symIndexK}{\sinh(\symY(\symIndexK+1))}
\left(
1
+
\frac{1}{\symIndexK}
\right)
\\
\begin{aligned}
&\leq
\frac{\symY\symIndexK}{\sinh(\symY(\symIndexK+1))}
\left(
\frac{\symY}{\tanh(\symY)} + \frac{1}{\symIndexK}\frac{\symY\symIndexK}{\tanh(\symY\symIndexK)}
\right)
\\
&=
\frac{\symY^2\symIndexK}{\sinh(\symY(\symIndexK+1))}
\left(
\frac{\cosh(\symY)}{\sinh(\symY)} + \frac{\cosh(\symY\symIndexK)}{\sinh(\symY\symIndexK)}
\right)
\\
&=
\frac{\symY^2\symIndexK}{\sinh(\symY(\symIndexK+1))}
\left(
\frac{\sinh(\symY\symIndexK)\cosh(\symY)+\cosh(\symY\symIndexK)\sinh(\symY)}{\sinh(\symY)\sinh(\symY\symIndexK)}
\right)
\\
&=
\frac{\symY^2\symIndexK}{\sinh(\symY(\symIndexK+1))}
\left(
\frac{\sinh(\symY(\symIndexK+1))}{\sinh(\symY)\sinh(\symY\symIndexK)}
\right)
\\
&=
\frac{\symY}{\sinh(\symY)}
\frac{\symY\symIndexK}{\sinh(\symY\symIndexK)}
\leq
\left(\frac{\symY}{\sinh(\symY)}\right)^{\symIndexK+1}.
\end{aligned}
\end{multline*}
\end{proof}

\begin{figure}
  \centering
  \includegraphics[width=0.5\linewidth]{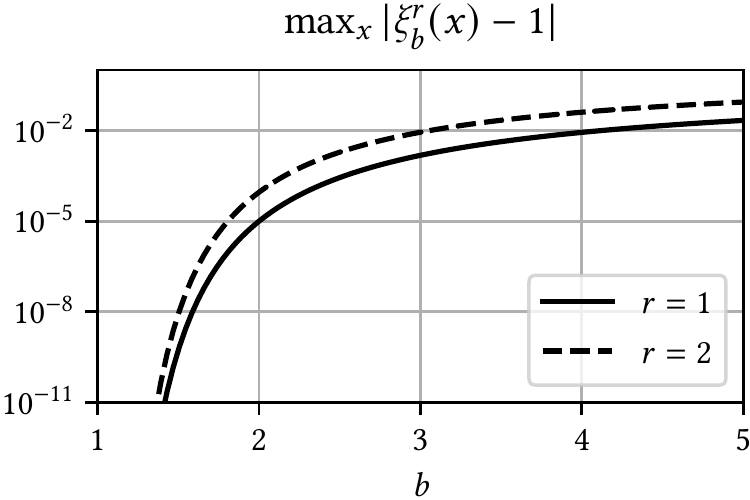}
  \caption{\boldmath Maximum deviation from 1 for $\symPowerSeriesFunc_\symBase^1$ and $\symPowerSeriesFunc_\symBase^2$ as function of $\symBase$.}
  \label{fig:helper_func_error}
\end{figure}

\begin{lemma}
\label{lem:xi1_approx}
The function $\symPowerSeriesFunc_\symBase^1$ defined in \cref{lem:xi} satisfies
\begin{equation*}
|\symPowerSeriesFunc_\symBase^1(\symX) - 1| \leq 
\frac{2}{\sqrt{
\sinh\!\left(\frac{2\pi^2}{\log \symBase}\right)
\frac{\log \symBase}{2\pi^2}
}
-1
}
\end{equation*}
for all $\symX\in\mathbb{R}$ and $\symBase>1$.
\end{lemma}
For $\symBase=2$ the right-hand side is less than \num{9.885e-6}. Since $\sinh(\symX) \sim \frac{e^\symX}{2}$ as $\symX\rightarrow\infty$, we have $|\symPowerSeriesFunc_\symBase^1(\symX) - 1| = \symBigO(e^{-\frac{\pi^2}{\log\symBase}}(\log\symBase)^{-\frac{1}{2}})$ which shows the rapid decay of the deviation from 1 as $\symBase\rightarrow 1$ (see \cref{fig:helper_func_error}). Therefore, $\symPowerSeriesFunc_\symBase^1(\symX)$ is well approximated 
by the constant $1$ for not too large values of $\symBase$, especially for $\symBase\leq 2$.
\begin{proof}
We use the Fourier series of $\symPowerSeriesFunc_\symBase^1$ (\cref{lem:xi}), the identity $|\Gamma(1-\symY\symImaginary)|^2 = \frac{\pi\symY}{\sinh(\pi\symY)}$, and \cref{lemma:inquality_sinh}:
\begin{multline*}
|\symPowerSeriesFunc_\symBase^1(\symX) - 1|
=
2\left|\operatorname{Re}\!\left(\sum_{\symIndexL=1}^\infty \Gamma\!\left(1 - \textstyle\frac{\symImaginary 2\pi \symIndexL}{\log \symBase}\right)e^{ \symImaginary 2\pi\symIndexL \symX}\right)\right|
\leq
2
\sum_{\symIndexL=1}^\infty
\left|
\Gamma\!\left(1 - \textstyle\frac{\symImaginary 2\pi \symIndexL}{\log \symBase}\right)
\right|
\\
=
2
\sum_{\symIndexL=1}^\infty
\sqrt{
\frac{\textstyle\frac{2\pi^2\symIndexL}{\log \symBase}}{\sinh\!\left(\textstyle\frac{2\pi^2  \symIndexL}{\log \symBase}\right)}
}
\leq
2
\sum_{\symIndexL=1}^\infty
\left(
\frac{\textstyle\frac{2\pi^2}{\log \symBase}}{\sinh\!\left(\textstyle\frac{2\pi^2}{\log \symBase}\right)}
\right)^\frac{\symIndexL}{2}
\\
=
\frac{2}{\sqrt{
\sinh\!\left(\frac{2\pi^2}{\log \symBase}\right)
\frac{\log \symBase}{2\pi^2}
}
-1
}.
\end{multline*}
Here we used the closed-form formula $\sum_{\symIndexL=1}^\infty \symZ^\symIndexL = \frac{1}{\frac{1}{\symZ}-1}$ for the geometric series which converges, because $\symZ=\frac{\symY}{\sinh\symY}\in(0,1)$ with $\symY = \frac{2\pi^2}{\log \symBase}>0$.
\end{proof}

\begin{lemma}
\label{lemma:inequality2}
$1+\symY\symIndexL^2\leq (1+\symY)^\symIndexL$ holds for all integers $\symIndexL \geq 1$ and all real $\symY\geq2$.
\end{lemma}
\begin{proof}
The case $\symIndexL=1$ is trivial. For $\symIndexL\geq2$ we have
\begin{multline*}
1+\symY\symIndexL^2 = 1 + \symY\symIndexL + \symY\symIndexL(\symIndexL-1) \leq 
1 + \symY\symIndexL + \frac{\symY}{2} \symY\symIndexL(\symIndexL-1)
\\
=
{\symIndexL\choose 0}+\symY {\symIndexL\choose 1}+\symY^2 {\symIndexL\choose 2}
\leq
\sum_{\symIndexJ=0}^{\symIndexL}\symY^\symIndexJ {\symIndexL\choose \symIndexJ}
=
(1+\symY)^\symIndexL.
\end{multline*}
\end{proof}
\begin{lemma}
\label{lem:xi2_approx}
The function $\symPowerSeriesFunc_\symBase^2$ defined in \cref{lem:xi} satisfies
\begin{equation*}
|\symPowerSeriesFunc_\symBase^2(\symX) - 1| 
\leq 
\frac{2}{\sqrt{
\frac{\sinh\left(\frac{2\pi^2}{\log \symBase}\right)
\log^3 \symBase
}{
2\pi^2(\log^2\symBase+4\pi^2)
}
}
-1
}
\end{equation*}
for all $\symX\in\mathbb{R}$ and all $\symBase$ with $1 < \symBase \leq e^{\sqrt{2}\pi}\approx\num{85.02}$.
\end{lemma}
For $\symBase=2$ the right-hand side is less than \num{9.015e-5}. Since $\sinh(\symX) \sim \frac{e^\symX}{2}$ as $\symX\rightarrow\infty$, we have $|\symPowerSeriesFunc_\symBase^2(\symX) - 1| = \symBigO(e^{-\frac{\pi^2}{\log\symBase}}(\log\symBase)^{-\frac{3}{2}})$ which shows the rapid decay of the deviation from 1 as $\symBase\rightarrow 1$ (see \cref{fig:helper_func_error}). Therefore, $\symPowerSeriesFunc_\symBase^2(\symX)$ is well approximated by the constant $1$ for not too large values of $\symBase$, especially for $\symBase\leq 2$.
\begin{proof}
$\symBase \leq e^{\sqrt{2}\pi}$ implies $\frac{4\pi^2}{\log^2\symBase}\geq 2$. Furthermore, 
we use the identity $|\Gamma(2-\symY\symImaginary)|^2 = (1 + \symY^2)\frac{\pi\symY}{\sinh(\pi\symY)}$, \cref{lemma:inquality_sinh}, and \cref{lemma:inequality2}:
\begin{multline*}
|\symPowerSeriesFunc_\symBase^2(\symX) - 1|
=
2\left|\operatorname{Re}\!\left(\sum_{\symIndexL=1}^\infty \Gamma\!\left(2 - \textstyle\frac{\symImaginary 2\pi \symIndexL}{\log \symBase}\right)e^{ \symImaginary 2\pi\symIndexL \symX}\right)\right|
\\
\begin{aligned}
&\leq
2
\sum_{\symIndexL=1}^\infty
\left|
\Gamma\!\left(2 - \textstyle\frac{\symImaginary 2\pi \symIndexL}{\log \symBase}\right)
\right|
=
2
\sum_{\symIndexL=1}^\infty
\sqrt{1+\frac{4\pi^2\symIndexL^2}{\log^2\symBase}}
\sqrt{
\frac{\textstyle\frac{2\pi^2\symIndexL}{\log \symBase}}{\sinh\!\left(\textstyle\frac{2\pi^2  \symIndexL}{\log \symBase}\right)}
}
\\
&\leq
2
\sum_{\symIndexL=1}^\infty
\sqrt{\left(1+\frac{4\pi^2}{\log^2\symBase}\right)^\symIndexL}
\sqrt{
\left(\frac{\textstyle\frac{2\pi^2}{\log \symBase}}{\sinh\!\left(\textstyle\frac{2\pi^2}{\log \symBase}\right)}
\right)^\symIndexL}
\\
&=
2
\sum_{\symIndexL=1}^\infty
\left(
\frac{\textstyle\left(1+\frac{4\pi^2}{\log^2\symBase}\right)\frac{2\pi^2}{\log \symBase}}{\sinh\!\left(\textstyle\frac{2\pi^2}{\log \symBase}\right)}
\right)^\frac{\symIndexL}{2}
=
\frac{2}{\sqrt{
\frac{\sinh\left(\frac{2\pi^2}{\log \symBase}\right)
}{
\left(1+\frac{4\pi^2}{\log^2\symBase}\right)\frac{2\pi^2}{\log \symBase}
}
}
-1
}
\\
&=
\frac{2}{\sqrt{
\frac{\sinh\left(\frac{2\pi^2}{\log \symBase}\right)
\log^3 \symBase
}{
2\pi^2(\log^2\symBase+4\pi^2)
}
}
-1
}.
\end{aligned}
\end{multline*}
The geometric series converges, because $\frac{(1+\symY^2/\pi^2)\symY}{\sinh\symY}<\frac{\symY+\symY^3/6}{\sinh\symY} <1$ with $\symY = \frac{2\pi^2}{\log \symBase} > 0$. The numerator $\symY+\symY^3/6$ is obviously smaller than $\sinh\symY$ as it corresponds to the first two terms of the Taylor series of the hyperbolic sine function.
\end{proof}

\begin{lemma}
\label{lem:zeta_approx}
The relative difference of the function
\begin{equation*}
\symHelperFunc_\symBase(\symX_1,\symX_2):={\textstyle\sum_{\symRegVal=-\infty}^\infty e^{-\symBase^{\symX_1-\symRegVal}} - e^{-\symBase^{\symX_2-\symRegVal}}}
\end{equation*}
from $\symX_2-\symX_1$ is bounded by
\begin{equation*}
\left|\frac{\symHelperFunc_\symBase(\symX_1,\symX_2) - (\symX_2 - \symX_1)}{\symX_2 - \symX_1}\right|\leq
\frac{2}{\sqrt{
\sinh\!\left(\frac{2\pi^2}{\log \symBase}\right)
\frac{\log \symBase}{2\pi^2}
}
-1
}
\quad
\textnormal{for $\symBase>1$}.
\end{equation*}
\end{lemma}
As this is the same upper bound as in \cref{lem:xi1_approx}, the relative error is also smaller than \num{9.885e-6} for $\symBase\leq 2$ and quickly approaches zero as $\symBase\rightarrow 1$.
\begin{proof}
The mean value theorem and \cref{lem:xi1_approx} yield
\begin{multline*}
\left|\frac{\symHelperFunc_\symBase(\symX_1,\symX_2) - (\symX_2 - \symX_1)}{\symX_2 - \symX_1}\right|
=
{\textstyle\left|\frac{\left(\sum_{\symRegVal=-\infty}^\infty e^{-\symBase^{\symX_1-\symRegVal}} - e^{-\symBase^{\symX_2-\symRegVal}}\right) - (\symX_2 - \symX_1)}{\symX_2 - \symX_1}\right|}
\\
\begin{aligned}
&=
{\textstyle\left|\frac{\int_{\symX_1}^{\symX_2}
\left(
\log(\symBase) \sum_{\symRegVal = -\infty}^\infty
\symBase^{\symX-\symRegVal}
e^{-\symBase^{\symX-\symRegVal}}
\right)
-1\,d\symX}{\symX_2 - \symX_1}\right|}
=
{\textstyle\left|\frac{\int_{\symX_1}^{\symX_2}
\symPowerSeriesFunc_\symBase^1(\symX)-1\,d\symX}{\symX_2 - \symX_1}\right|}
\\
&\leq
\max_\symX
\left|\symPowerSeriesFunc_\symBase^1(\symX)-1\right|
\leq
\frac{2}{\sqrt{
\sinh\!\left(\frac{2\pi^2}{\log \symBase}\right)
\frac{\log \symBase}{2\pi^2}
}
-1
}.
\end{aligned}
\end{multline*}
\end{proof}

\begin{lemma}
\label{lem:card_variance}
If the random values $\symRegValVariate_\symIndexI$ are independent and identically distributed such that $\symExpectation((\symExponentialRate \symBase^{-\symRegValVariate_\symIndexI})^\symIntPower) = \frac{(1-\symBase^{-\symIntPower})\,\Gamma(\symIntPower)}{\symCardinality^\symIntPower \log \symBase}$ holds for $\symIntPower\in\lbrace 1,2\rbrace$, the variance of $\symStatisticX_\symNumReg = \frac{\log \symBase}{1-\symBase^{-1}}\frac{1}{\symNumReg}\sum_{\symIndexI=1}^\symNumReg \symExponentialRate \symBase^{-\symRegValVariate_\symIndexI}$ is given by $\symVariance(\symStatisticX_\symNumReg) = \frac{1}{\symNumReg\symCardinality^2}\left(\frac{\symBase+1}{\symBase-1}\log(\symBase) - 1\right)$.
\end{lemma}
\begin{proof}
\begin{align*}
\symVariance(\symStatisticX_\symNumReg) &= 
\textstyle
\left(\frac{\log \symBase}{1-\symBase^{-1}}\frac{1}{\symNumReg}\right)^2
\sum_{\symIndexI=1}^\symNumReg
\symVariance(\symExponentialRate \symBase^{-\symRegValVariate_\symIndexI})
\\
&=
\textstyle
\frac{1}{\symNumReg}\frac{\log^2 \symBase}{(1-\symBase^{-1})^2}
\symVariance(\symExponentialRate \symBase^{-\symRegValVariate_\symIndexI})
\\
&=
\textstyle
\frac{1}{\symNumReg}\frac{\log^2 \symBase}{(1-\symBase^{-1})^2}
\left(\symExpectation((\symExponentialRate \symBase^{-\symRegValVariate_\symIndexI})^2)
-
(\symExpectation(\symExponentialRate \symBase^{-\symRegValVariate_\symIndexI}))^2
\right)
\\
&=
\textstyle
\frac{1}{\symNumReg}\frac{\log^2 \symBase}{(1-\symBase^{-1})^2}
\left(
\frac{(1-\symBase^{-2})}{\symCardinality^2 \log \symBase}
-
\frac{(1-\symBase^{-1})^2}{\symCardinality^2 \log^2 \symBase}
\right)
\\
&=
\textstyle
\frac{1}{\symNumReg\symCardinality^2}\left(\frac{\symBase+1}{\symBase-1}\log(\symBase) - 1\right).
\end{align*}
\end{proof}

\begin{lemma}
\label{lem:prob_consistency}
$
1
-\symProbFunc_{\symBase}(\symCardinalityANorm - \symCardinalityBNorm\symJaccard)
-\symProbFunc_{\symBase}(\symCardinalityBNorm - \symCardinalityANorm\symJaccard)
> 0$ holds for all $\symCardinalityANorm,\symCardinalityBNorm>0$ with $\symCardinalityANorm+\symCardinalityBNorm=1$, $\symBase>1$, and $\symJaccard\in[0, \min(\frac{\symCardinalityANorm}{\symCardinalityBNorm}, \frac{\symCardinalityBNorm}{\symCardinalityANorm})]$, where $\symProbFunc_{\symBase}(\symX)$ is defined as $\symProbFunc_{\symBase}(\symX):=-\log_\symBase(
1-\symX\frac{\symBase-1}{\symBase}
)$.
\end{lemma}
\begin{proof}
 $\symJaccard\in[0, \min(\frac{\symCardinalityANorm}{\symCardinalityBNorm}, \frac{\symCardinalityBNorm}{\symCardinalityANorm})]$ implies $\symCardinalityANorm - \symCardinalityBNorm\symJaccard\geq 0$ and $\symCardinalityBNorm - \symCardinalityANorm\symJaccard\geq 0$ and therefore
\begin{multline*}
1
-\symProbFunc_{\symBase}(\symCardinalityANorm - \symCardinalityBNorm\symJaccard)
-\symProbFunc_{\symBase}(\symCardinalityBNorm - \symCardinalityANorm\symJaccard)
\\
\begin{aligned}
&=
\log_\symBase(\symBase)
+
\log_\symBase(1-(\symCardinalityANorm - \symCardinalityBNorm\symJaccard){\textstyle\frac{\symBase-1}{\symBase}})
+
\log_\symBase(1-(\symCardinalityBNorm - \symCardinalityANorm\symJaccard){\textstyle\frac{\symBase-1}{\symBase}})
\\
&=
\log_\symBase\!\left(
\symBase
(1-(\symCardinalityANorm - \symCardinalityBNorm\symJaccard){\textstyle\frac{\symBase-1}{\symBase}})
(1-(\symCardinalityBNorm - \symCardinalityANorm\symJaccard){\textstyle\frac{\symBase-1}{\symBase}})
\right)
\\
&=
\log_\symBase\!\left(
\symBase
\left(
1
-
(\symCardinalityANorm+\symCardinalityBNorm)(1-\symJaccard){\textstyle\frac{\symBase-1}{\symBase}}
+
(\symCardinalityANorm - \symCardinalityBNorm\symJaccard)
(\symCardinalityBNorm - \symCardinalityANorm\symJaccard)
({\textstyle\frac{\symBase-1}{\symBase}})^2
\right)
\right)
\\
&=
\log_\symBase\!\left(
1
+
\symJaccard(\symBase-1)
+
(\symCardinalityANorm - \symCardinalityBNorm\symJaccard)
(\symCardinalityBNorm - \symCardinalityANorm\symJaccard)
{\textstyle\frac{(\symBase-1)^2}{\symBase}}
\right)
>
\log_\symBase(1)
=
0.
\end{aligned}
\end{multline*}
\end{proof}

\begin{lemma}
\label{lem:concavity}
Assume $1<\symBase\leq e\approx\num{2.718}$, $\symDiffCountVariate_{+},\symDiffCountVariate_{-}, \symDiffCountVariate_{0}\geq 0$ with $\symDiffCountVariate_{+}+\symDiffCountVariate_{-}+\symDiffCountVariate_{0}> 0$, and $\symCardinalityANorm,\symCardinalityBNorm>0$ with $\symCardinalityANorm+\symCardinalityBNorm=1$. Then the function 
\begin{multline*}
\log\symLikelihood(\symJaccard)
=
\symDiffCountVariate_{+} \log(\symProbFunc_{\symBase}(\symCardinalityANorm - \symCardinalityBNorm\symJaccard)) + \symDiffCountVariate_{-}\log(\symProbFunc_{\symBase}(\symCardinalityBNorm - \symCardinalityANorm\symJaccard))
\\
+\symDiffCountVariate_{0}\log(
1
-
\symProbFunc_{\symBase}(\symCardinalityANorm - \symCardinalityBNorm\symJaccard)
-
\symProbFunc_{\symBase}(\symCardinalityBNorm - \symCardinalityANorm\symJaccard)
)
\end{multline*}
with $\symProbFunc_{\symBase}(\symX):=-\log_\symBase(
1-\symX\frac{\symBase-1}{\symBase}
)$
is strictly concave on the domain $\symJaccard\in [0, \min(\frac{\symCardinalityANorm}{\symCardinalityBNorm}, \frac{\symCardinalityBNorm}{\symCardinalityANorm})]$.
\end{lemma}
\begin{proof}
We show that all three terms are strictly concave themselves. For the first two terms, it is sufficient to show that the function $\symAnyFunc(\symX)=\log(\symProbFunc_\symBase(\symX))$ is strictly concave on $\symX\in[0,1)$.
The second derivative is 
\begin{equation*}
\symAnyFunc''(\symX) = -\frac{(\frac{\symBase-1}{\symBase})^2(\log(1-\symX\frac{\symBase-1}{\symBase})+1)}{(1-\symX\frac{\symBase-1}{\symBase})^2\log^2(1-\symX\frac{\symBase-1}{\symBase})}
\end{equation*}
which is negative for all $\symX\in(0,1)$ as long as $\log(1-\symX\frac{\symBase-1}{\symBase})+1>0\Leftrightarrow 1-\frac{1}{e}>\symX(1-\frac{1}{\symBase})$, which is obviously the case for $\symBase\leq e$.

The last term is $\log(\symAnyFuncTwo(\symJaccard))$ with 
\begin{align*}
\symAnyFuncTwo(\symJaccard) &= 1
-
\symProbFunc_{\symBase}(\symCardinalityANorm - \symCardinalityBNorm\symJaccard)
-
\symProbFunc_{\symBase}(\symCardinalityBNorm - \symCardinalityANorm\symJaccard)
\\
&=
1+
\log_\symBase(
1-(\symCardinalityANorm - \symCardinalityBNorm \symJaccard){\textstyle\frac{\symBase-1}{\symBase}}
)
+
\log_\symBase(
1-(\symCardinalityBNorm - \symCardinalityANorm \symJaccard){\textstyle\frac{\symBase-1}{\symBase}}
).
\end{align*}
Its second derivative is $(\log(\symAnyFuncTwo(\symJaccard)))'' = \frac{\symAnyFuncTwo''(\symJaccard)}{\symAnyFuncTwo(\symJaccard)}-\left(\frac{\symAnyFuncTwo'(\symJaccard)}{\symAnyFuncTwo(\symJaccard)}\right)^2$ which is negative because $\symAnyFuncTwo(\symJaccard)>0$ according to \cref{lem:prob_consistency}, $\symAnyFuncTwo'(\symJaccard)>0$ as $\symProbFunc_{\symBase}$ and hence also $\symAnyFuncTwo$ are strictly increasing, and $\symAnyFuncTwo''(\symJaccard)\leq 0$ as all terms of $\symAnyFuncTwo(\symJaccard)$ are concave. In particular, $\log_\symBase(1-(\symCardinalityANorm - \symCardinalityBNorm \symJaccard)(1 - 1/ \symBase))$ and $\log_\symBase(1-(\symCardinalityBNorm - \symCardinalityANorm \symJaccard)(1 - 1/ \symBase))$ are both concave, because the logarithm of a linear function is concave.
\end{proof}

\begin{lemma}
\label{lem:fisher}
If $\symDiffCountVariate_{+}$, $\symDiffCountVariate_{-}$, and $\symDiffCountVariate_{0}$ are multinomially distributed with $\symDiffCountVariate_{+}+\symDiffCountVariate_{-}+\symDiffCountVariate_{0}=\symNumReg$ trials and probabilities $\symProbFunc_{+}(\symJaccard):=\symProbFunc_{\symBase}(\symCardinalityANorm - \symCardinalityBNorm\symJaccard)$, $\symProbFunc_{-}(\symJaccard):=\symProbFunc_{\symBase}(\symCardinalityBNorm - \symCardinalityANorm\symJaccard)$, and $\symProbFunc_{0}(\symJaccard):=1
-
\symProbFunc_{\symBase}(\symCardinalityANorm - \symCardinalityBNorm\symJaccard)
-
\symProbFunc_{\symBase}(\symCardinalityBNorm - \symCardinalityANorm\symJaccard)
)$, respectively, with $\symBase > 1$, $\symCardinalityANorm,\symCardinalityBNorm>0$, $\symCardinalityANorm+\symCardinalityBNorm=1$, and $\symProbFunc_{\symBase}(\symX):=-\log_\symBase(
1-\symX\frac{\symBase-1}{\symBase}
)$, the Fisher information with respect to $\symJaccard$ is given by
\begin{equation*}
\symFisher(\symJaccard)
= 
{\scriptstyle 
\frac{\symNumReg
(\symBase - 1)^2}{\symBase^2\log^2(\symBase)}
\left(
\frac{
\left(\symCardinalityBNorm
\symBase^{\symProbFunc_\symBase(\symCardinalityANorm - \symCardinalityBNorm \symJaccard)}
\right)^2
}{
\symProbFunc_\symBase(\symCardinalityANorm - \symCardinalityBNorm \symJaccard)}
+
\frac{
\left(\symCardinalityANorm
\symBase^{\symProbFunc_\symBase(\symCardinalityBNorm - \symCardinalityANorm \symJaccard)}
\right)^2
}{
\symProbFunc_\symBase(\symCardinalityBNorm - \symCardinalityANorm \symJaccard)
}
+
\frac{
\left(
\symCardinalityBNorm
\symBase^{\symProbFunc_\symBase(\symCardinalityANorm - \symCardinalityBNorm \symJaccard)}
+
\symCardinalityANorm
\symBase^{\symProbFunc_\symBase(\symCardinalityBNorm - \symCardinalityANorm \symJaccard)}
\right)^2
}{
1
-
\symProbFunc_\symBase(\symCardinalityANorm - \symCardinalityBNorm \symJaccard)
-
\symProbFunc_\symBase(\symCardinalityBNorm - \symCardinalityANorm \symJaccard)
}
\right)
}
\end{equation*}
for $\symJaccard\in[0, \min(\frac{\symCardinalityANorm}{\symCardinalityBNorm}, \frac{\symCardinalityBNorm}{\symCardinalityANorm}))$ and diverges for $\symJaccard=\min(\frac{\symCardinalityANorm}{\symCardinalityBNorm}, \frac{\symCardinalityBNorm}{\symCardinalityANorm})$.
\end{lemma}
\begin{proof}
For $\symJaccard\in[0, \min(\frac{\symCardinalityANorm}{\symCardinalityBNorm}, \frac{\symCardinalityBNorm}{\symCardinalityANorm}))$ all probabilities are positive due to \cref{lem:prob_consistency}.
By definition, $\symProbFunc_{+}(\symJaccard)+\symProbFunc_{-}(\symJaccard)+\symProbFunc_{0}(\symJaccard)=1$ and hence the sum of their derivatives vanishes $\symProbFunc'_{+}(\symJaccard)+\symProbFunc'_{-}(\symJaccard)+\symProbFunc'_{0}(\symJaccard)=0$. The first derivative of the log-likelihood function
\begin{equation*}
\log\symLikelihood(\symJaccard)
=
\symDiffCountVariate_{+} \log(\symProbFunc_{+}(\symJaccard)) + 
\symDiffCountVariate_{-} \log(\symProbFunc_{-}(\symJaccard)) + 
\symDiffCountVariate_{0} \log(\symProbFunc_{0}(\symJaccard))
\end{equation*}
is given by
\begin{equation*}
(\log\symLikelihood(\symJaccard))'
=
\symDiffCountVariate_{+} {\textstyle\frac{\symProbFunc'_{+}(\symJaccard)}{\symProbFunc_{+}(\symJaccard)}} + 
\symDiffCountVariate_{-} {\textstyle\frac{\symProbFunc'_{-}(\symJaccard)}{\symProbFunc_{-}(\symJaccard)}} + 
\symDiffCountVariate_{0} {\textstyle\frac{\symProbFunc'_{0}(\symJaccard)}{\symProbFunc_{0}(\symJaccard)}}.
\end{equation*}
Since $\symExpectation(\symDiffCountVariate_{+}) = \symNumReg\symProbFunc_{+}(\symJaccard)$, $\symExpectation(\symDiffCountVariate_{-}) = \symNumReg\symProbFunc_{-}(\symJaccard)$, and $\symExpectation(\symDiffCountVariate_{0}) = \symNumReg\symProbFunc_{0}(\symJaccard)$, we have $\symExpectation((\log\symLikelihood(\symJaccard))') = \symNumReg(\symProbFunc'_{+}(\symJaccard) + \symProbFunc'_{-}(\symJaccard) + \symProbFunc'_{0}(\symJaccard)) = 0$.
Therefore, the Fisher information can be computed as
\begin{equation*}
\begin{aligned}
\symFisher(\symJaccard)
&=
\symExpectation(((\log\symLikelihood(\symJaccard))')^2)
=
\symVariance((\log\symLikelihood(\symJaccard))') + \symExpectation((\log\symLikelihood(\symJaccard))')^2
\\
&=
\symVariance((\log\symLikelihood(\symJaccard))')
\\
&=
\symVariance(\symDiffCountVariate_{+}) \left({\textstyle\frac{\symProbFunc'_{+}(\symJaccard)}{\symProbFunc_{+}(\symJaccard)}}\right)^2 + 
\symVariance(\symDiffCountVariate_{-}) \left({\textstyle\frac{\symProbFunc'_{-}(\symJaccard)}{\symProbFunc_{-}(\symJaccard)}}\right)^2 + 
\symVariance(\symDiffCountVariate_{0}) \left({\textstyle\frac{\symProbFunc'_{0}(\symJaccard)}{\symProbFunc_{0}(\symJaccard)}}\right)^2
\\
\textstyle
&\quad+
2\symCovariance(\symDiffCountVariate_{+}, \symDiffCountVariate_{-}) {\textstyle\frac{\symProbFunc'_{+}(\symJaccard)}{\symProbFunc_{+}(\symJaccard)}\frac{\symProbFunc'_{-}(\symJaccard)}{\symProbFunc_{-}(\symJaccard)}}
+
2\symCovariance(\symDiffCountVariate_{+}, \symDiffCountVariate_{0}) {\textstyle\frac{\symProbFunc'_{+}(\symJaccard)}{\symProbFunc_{+}(\symJaccard)}\frac{\symProbFunc'_{0}(\symJaccard)}{\symProbFunc_{0}(\symJaccard)}}
\\
&\quad +
2\symCovariance(\symDiffCountVariate_{0}, \symDiffCountVariate_{-}) {\textstyle\frac{\symProbFunc'_{0}(\symJaccard)}{\symProbFunc_{0}(\symJaccard)}\frac{\symProbFunc'_{-}(\symJaccard)}{\symProbFunc_{-}(\symJaccard)}}.
\end{aligned}
\end{equation*}
Using the formulas for variance e.g. $\symVariance(\symDiffCountVariate_{+}) = \symNumReg \symProbFunc_{+}(\symJaccard)(1-\symProbFunc_{+}(\symJaccard))$ and covariance e.g. $\symCovariance(\symDiffCountVariate_{+}, \symDiffCountVariate_{0})=-\symNumReg \symProbFunc_{+}(\symJaccard)\symProbFunc_{0}(\symJaccard)$ of multinomially distributed variables we obtain
\begin{multline*}
\symFisher(\symJaccard)
=
\symNumReg {\textstyle\frac{(\symProbFunc'_{+}(\symJaccard))^2}{\symProbFunc_{+}(\symJaccard)}}
+
\symNumReg {\textstyle\frac{(\symProbFunc'_{-}(\symJaccard))^2}{\symProbFunc_{-}(\symJaccard)}}
+
\symNumReg {\textstyle\frac{(\symProbFunc'_{0}(\symJaccard))^2}{\symProbFunc_{0}(\symJaccard)}}
\\
-
\symNumReg
\left(
\symProbFunc'_{+}(\symJaccard) + \symProbFunc'_{-}(\symJaccard) + \symProbFunc'_{0}(\symJaccard)
\right)^2
\\
=
\symNumReg
{\textstyle\frac{(\symProbFunc'_{+}(\symJaccard))^2}{\symProbFunc_{+}(\symJaccard)}}
+
\symNumReg{\textstyle\frac{(\symProbFunc'_{-}(\symJaccard))^2}{\symProbFunc_{-}(\symJaccard)}}
+
\symNumReg{\textstyle\frac{(\symProbFunc'_{0}(\symJaccard))^2}{\symProbFunc_{0}(\symJaccard)}}
\\
=
\symNumReg
\left(
{\textstyle\frac{(\symCardinalityBNorm\symProbFunc'_{\symBase}(\symCardinalityANorm - \symCardinalityBNorm\symJaccard))^2}{\symProbFunc_{\symBase}(\symCardinalityANorm - \symCardinalityBNorm\symJaccard)}}
+
{\textstyle\frac{(\symCardinalityANorm\symProbFunc'_{\symBase}(\symCardinalityBNorm - \symCardinalityANorm\symJaccard))^2}{\symProbFunc_{\symBase}(\symCardinalityBNorm - \symCardinalityANorm\symJaccard)}}
+
{\textstyle\frac{(
\symCardinalityBNorm\symProbFunc'_{\symBase}(\symCardinalityANorm - \symCardinalityBNorm\symJaccard)
+
\symCardinalityANorm\symProbFunc'_{\symBase}(\symCardinalityBNorm - \symCardinalityANorm\symJaccard)
)^2}{1
-
\symProbFunc_{\symBase}(\symCardinalityANorm - \symCardinalityBNorm\symJaccard)
-
\symProbFunc_{\symBase}(\symCardinalityBNorm - \symCardinalityANorm\symJaccard)
}}
\right).
\end{multline*}
The first derivative of $\symProbFunc_\symBase$ can be expressed as
\begin{equation*}
\symProbFunc_\symBase'(\symX) = \frac{\symBase-1}{\symBase\log\symBase}\symBase^{\symProbFunc_\symBase(\symX)}
\end{equation*}
which finally gives the desired expression for the Fisher information.
If $\symJaccard=\min(\frac{\symCardinalityANorm}{\symCardinalityBNorm}, \frac{\symCardinalityBNorm}{\symCardinalityANorm})$, either $\symProbFunc_{\symBase}(\symCardinalityANorm - \symCardinalityBNorm\symJaccard)=0$ or $\symProbFunc_{\symBase}(\symCardinalityBNorm - \symCardinalityANorm\symJaccard)=0$, which shows the divergence in this case.
\end{proof}

\begin{lemma}
\label{lem:lsh_inequality}
If $\symCardinalityANorm>0$, $\symCardinalityBNorm>0$, $\symCardinalityANorm+\symCardinalityBNorm=1$, and $\symJaccard\in[0, \min(\frac{\symCardinalityANorm}{\symCardinalityBNorm}, \frac{\symCardinalityBNorm}{\symCardinalityANorm})]$, the inequality 
\begin{equation*}
0\leq (\symCardinalityANorm-\symCardinalityBNorm\symJaccard)(\symCardinalityBNorm-\symCardinalityANorm\symJaccard)\leq \frac{1}{4}(1-\symJaccard)^2
\end{equation*}
is satisfied. Left and right equality are obtained for $\lbrace\symCardinalityANorm, \symCardinalityBNorm\rbrace = \lbrace 1/(1+\symJaccard), \symJaccard/(1+\symJaccard)\rbrace$ and $\symCardinalityANorm = \symCardinalityBNorm = \frac{1}{2}$, respectively.
\end{lemma}

\begin{proof}
$\symJaccard\in[0, \min(\frac{\symCardinalityANorm}{\symCardinalityBNorm}, \frac{\symCardinalityBNorm}{\symCardinalityANorm})]$ implies $\symCardinalityANorm - \symCardinalityBNorm\symJaccard\geq 0$ and $\symCardinalityBNorm - \symCardinalityANorm\symJaccard\geq 0$. Hence, the left inequality clearly holds and is equal if either $\symCardinalityANorm - \symCardinalityBNorm\symJaccard = 0$ or $\symCardinalityBNorm - \symCardinalityANorm\symJaccard = 0$, which gives together with $\symCardinalityANorm + \symCardinalityBNorm = 1$ the two solutions.

For the right inequality we have
\begin{equation*}
(\symCardinalityANorm-\symCardinalityBNorm\symJaccard)(\symCardinalityBNorm-\symCardinalityANorm\symJaccard)
\leq
\left(\frac{(\symCardinalityANorm-\symCardinalityBNorm\symJaccard) + (\symCardinalityBNorm-\symCardinalityANorm\symJaccard)}{2}\right)^2
=
\frac{1}{4}(1-\symJaccard)^2
\end{equation*}
where we used the inequality of arithmetic and geometric means $\sqrt{\symX\symY}\leq \frac{\symX+\symY}{2}$. Equality is achieved, if $\symCardinalityANorm-\symCardinalityBNorm\symJaccard = \symCardinalityBNorm-\symCardinalityANorm\symJaccard$ which implies $\symCardinalityANorm=\symCardinalityBNorm=\frac{1}{2}$.
\end{proof}

\begin{lemma}
\label{lem:pb_lim}
$\lim_{\symBase\rightarrow 1}\symProbFunc_{\symBase}(\symX) = \symX$
where the function $\symProbFunc_{\symBase}(\symX)$ is defined for $\symBase>1$ and $\symX\in[0,1]$ as
$\symProbFunc_{\symBase}(\symX):=-\log_\symBase(
1-\symX\frac{\symBase-1}{\symBase}
)$.
\end{lemma}
\begin{proof}
\begin{align*}
\lim_{\symBase\rightarrow 1}\symProbFunc_{\symBase}(\symX) 
&= 
\lim_{\symBase\rightarrow 1} -\log_\symBase(
1-\symX{\textstyle
\frac{\symBase-1}{\symBase}}
)
\\
&=
\lim_{\symBase\rightarrow 1}
\log_\symBase(
\symBase)
- \log_\symBase(
\symBase-\symX(\symBase-1)
)
\\
&= 
1-\lim_{\symBase\rightarrow 1} \log_\symBase(
\symBase-\symX(\symBase-1)
)
\\
&= 
1-\lim_{\symBase\rightarrow 1} {\textstyle
\frac{\log(
\symBase-\symX(\symBase-1)
)}{\log\symBase}}
\\
&=
1-\lim_{\symBase\rightarrow 1} {\textstyle
\frac{\frac{1-\symX}{
\symBase-\symX(\symBase-1)
}}{\frac{1}{\symBase}}}
=
1-(1-\symX)
=
\symX,
\end{align*}
where we used L'Hospital's rule. 
\end{proof}

\begin{lemma}
\label{lem:ml_estimate_mh}
The maximum of the function
\begin{multline*}
\log\symLikelihood(\symJaccard) = \lim_{\symBase\rightarrow 1} \symDiffCountVariate_{+} \log(\symProbFunc_{\symBase}(\symCardinalityANorm - \symCardinalityBNorm\symJaccard)) + \symDiffCountVariate_{-}\log(\symProbFunc_{\symBase}(\symCardinalityBNorm - \symCardinalityANorm\symJaccard))
\\
+\symDiffCountVariate_{0}\log(
1
-
\symProbFunc_{\symBase}(\symCardinalityANorm - \symCardinalityBNorm\symJaccard)
-
\symProbFunc_{\symBase}(\symCardinalityBNorm - \symCardinalityANorm\symJaccard)
)
\end{multline*}
with $\symCardinalityANorm,\symCardinalityBNorm>0$, $\symCardinalityANorm+\symCardinalityBNorm=1$, $\symDiffCountVariate_{+}, \symDiffCountVariate_{-}, \symDiffCountVariate_{0}\geq 0$, $\symDiffCountVariate_{+}+\symDiffCountVariate_{-}+\symDiffCountVariate_{0}=\symNumReg>0$, and $\symProbFunc_{\symBase}(\symX):=-\log_\symBase(
1-\symX\frac{\symBase-1}{\symBase}
)$ and $\symJaccard\in [0, \min(\frac{\symCardinalityANorm}{\symCardinalityBNorm}, \frac{\symCardinalityBNorm}{\symCardinalityANorm})]$ is given by 
\begin{equation*}
\textstyle
\symJaccardEstimate
=
\frac{
\symCardinalityANorm^2
(\symDiffCountVariate_0+\symDiffCountVariate_{-})
+
\symCardinalityBNorm^2
(\symDiffCountVariate_0+\symDiffCountVariate_{+})
-\sqrt{
\left(
\symCardinalityANorm^2
(\symDiffCountVariate_0+\symDiffCountVariate_{-})
-
\symCardinalityBNorm^2
(\symDiffCountVariate_0+\symDiffCountVariate_{+})
\right)^2
+
4
\symDiffCountVariate_{-}
\symDiffCountVariate_{+}
\symCardinalityANorm^2
\symCardinalityBNorm^2
}
}{
2\symNumReg\symCardinalityANorm
\symCardinalityBNorm
}.
\end{equation*}
\end{lemma}

\begin{proof}
Using \cref{lem:pb_lim} and $\symCardinalityANorm+\symCardinalityBNorm=1$ we have
\begin{equation*}
\log\symLikelihood(\symJaccard) = \symDiffCountVariate_{+} \log(\symCardinalityANorm - \symCardinalityBNorm\symJaccard) + \symDiffCountVariate_{-}\log(\symCardinalityBNorm - \symCardinalityANorm\symJaccard)
+\symDiffCountVariate_{0}\log(\symJaccard).
\end{equation*}
Assume first $\symDiffCountVariate_{+}, \symDiffCountVariate_{-}, \symDiffCountVariate_{0}>0$, which implies 
$\lim_{\symJaccard\rightarrow 0} \log\symLikelihood(\symJaccard) = -\infty$ and $\lim_{\symJaccard\rightarrow \min(\frac{\symCardinalityANorm}{\symCardinalityBNorm}, \frac{\symCardinalityBNorm}{\symCardinalityANorm})} \log\symLikelihood(\symJaccard) = -\infty$. Therefore, and because $\log\symLikelihood(\symJaccard)$ is strictly concave according to \cref{lem:concavity}, the \ac{ML} estimate is unique and from $(0, \min(\frac{\symCardinalityANorm}{\symCardinalityBNorm}, \frac{\symCardinalityBNorm}{\symCardinalityANorm}))$ and can be found as the stationary point.

Setting the first derivative of $\log\symLikelihood(\symJaccard)$ equal to 0 gives 
\begin{equation*}
\textstyle
(\log\symLikelihood(\symJaccard))' = 
\symDiffCountVariate_{+} \frac{\symCardinalityBNorm}{\symCardinalityBNorm\symJaccard - \symCardinalityANorm}
+ \symDiffCountVariate_{-}\frac{\symCardinalityANorm}{\symCardinalityANorm\symJaccard- \symCardinalityBNorm}
+\symDiffCountVariate_{0}\frac{1}{\symJaccard}
=
0,
\end{equation*}
which can be transformed into the quadratic equation when using $\symDiffCountVariate_{+}+\symDiffCountVariate_{-}+\symDiffCountVariate_{0}=\symNumReg$
\begin{equation*}
\symNumReg
\symCardinalityANorm
\symCardinalityBNorm
\symJaccard^2
-
(
\symCardinalityANorm^2
(\symDiffCountVariate_0+\symDiffCountVariate_{-})
+
\symCardinalityBNorm^2
(\symDiffCountVariate_0+\symDiffCountVariate_{+})
)
\symJaccard
 +
\symDiffCountVariate_0
\symCardinalityANorm
\symCardinalityBNorm
=
0.
\end{equation*}
Using the quadratic formula, the smaller solution is given by
\begin{equation*}
\textstyle
\symJaccardEstimate
=
\frac{
\symCardinalityANorm^2
(\symDiffCountVariate_0+\symDiffCountVariate_{-})
+
\symCardinalityBNorm^2
(\symDiffCountVariate_0+\symDiffCountVariate_{+})
-\sqrt{
\left(
\symCardinalityANorm^2
(\symDiffCountVariate_0+\symDiffCountVariate_{-})
-
\symCardinalityBNorm^2
(\symDiffCountVariate_0+\symDiffCountVariate_{+})
\right)^2
+
4
\symDiffCountVariate_{-}
\symDiffCountVariate_{+}
\symCardinalityANorm^2
\symCardinalityBNorm^2
}
}{
2\symNumReg\symCardinalityANorm
\symCardinalityBNorm
}.
\end{equation*}
For this solution the following inequalities hold
\begin{multline*}
0
=
\scriptstyle
\frac{
\symCardinalityANorm^2
(\symDiffCountVariate_0+\symDiffCountVariate_{-})
+
\symCardinalityBNorm^2
(\symDiffCountVariate_0+\symDiffCountVariate_{+})
-\sqrt{
\left(
\symCardinalityANorm^2
(\symDiffCountVariate_0+\symDiffCountVariate_{-})
-
\symCardinalityBNorm^2
(\symDiffCountVariate_0+\symDiffCountVariate_{+})
\right)^2
+
4
(\symDiffCountVariate_0+\symDiffCountVariate_{-})
(\symDiffCountVariate_0+\symDiffCountVariate_{+})
\symCardinalityANorm^2
\symCardinalityBNorm^2
}
}{
2\symNumReg\symCardinalityANorm
\symCardinalityBNorm
}
\\
<
\symJaccardEstimate
<
\textstyle
\frac{
\symCardinalityANorm^2
(\symDiffCountVariate_0+\symDiffCountVariate_{-})
+
\symCardinalityBNorm^2
(\symDiffCountVariate_0+\symDiffCountVariate_{+})
-\sqrt{
\left(
\symCardinalityANorm^2
(\symDiffCountVariate_0+\symDiffCountVariate_{-})
-
\symCardinalityBNorm^2
(\symDiffCountVariate_0+\symDiffCountVariate_{+})
\right)^2
}
}{
2\symNumReg\symCardinalityANorm
\symCardinalityBNorm
}
\\
\textstyle
=
\min(
\frac{\symCardinalityANorm}{\symCardinalityBNorm}
\frac{\symDiffCountVariate_0+\symDiffCountVariate_{-}}{\symNumReg}
,
\frac{\symCardinalityBNorm}{\symCardinalityANorm}
\frac{\symDiffCountVariate_0+\symDiffCountVariate_{+}}{\symNumReg}
)
\leq
\min(
\frac{\symCardinalityANorm}{\symCardinalityBNorm}
,
\frac{\symCardinalityBNorm}{\symCardinalityANorm}
)
\end{multline*}
which shows that the smaller solution satisfies $\symJaccardEstimate\in(0, \min(\frac{\symCardinalityANorm}{\symCardinalityBNorm}, \frac{\symCardinalityBNorm}{\symCardinalityANorm}))$. Since we know that the solution is unique, $\symJaccardEstimate$ is indeed the searched solution.

The cases for which at least one of $\symDiffCountVariate_{-}$, $\symDiffCountVariate_{+}$, $\symDiffCountVariate_{0}$ is zero, need separate treatment. However, it is easy to verify that the above formula for $\symJaccardEstimate$ also holds for those special cases.
\end{proof}

\begin{lemma}
\label{lem:fisher_limit}
The Fisher information given in \cref{lem:fisher} has the limit
\begin{equation*}
\lim_{\symBase\rightarrow 1} \symFisher(\symJaccard) = 
\textstyle
\frac{\symNumReg}{\symJaccard(1-\symJaccard)}
\frac{
1
}{
1
-
  \frac{
    (\symCardinalityANorm-\symCardinalityBNorm)^2\symJaccard
  }{
    \symCardinalityANorm\symCardinalityBNorm(1-\symJaccard)^2
  }
}
\end{equation*}
as $\symBase\rightarrow 1$.
\end{lemma}

\begin{proof}
Using $\lim_{\symBase\rightarrow 1}\frac{(\symBase - 1)^2}{\symBase^2\log^2(\symBase)}=1$ and \cref{lem:pb_lim}, which implies $\lim_{\symBase\rightarrow 1}\symProbFunc_{\symBase}(\symX)=\symX$ and $\lim_{\symBase\rightarrow 1}\symBase^{\symProbFunc_{\symBase}(\symX)}=1$, the Fisher information as given in \cref{lem:fisher} becomes
\begin{align*}
\lim_{\symBase\rightarrow 1}\symFisher(\symJaccard)
&= 
\textstyle
\symNumReg
\left(
\frac{
\symCardinalityBNorm^2
}{
\symCardinalityANorm - \symCardinalityBNorm \symJaccard}
+
\frac{
\symCardinalityANorm^2
}{
\symCardinalityBNorm - \symCardinalityANorm \symJaccard
}
+
\frac{
\left(
\symCardinalityBNorm
+
\symCardinalityANorm
\right)^2
}{
1
-
(\symCardinalityANorm - \symCardinalityBNorm \symJaccard)
-
(\symCardinalityBNorm - \symCardinalityANorm \symJaccard)
}
\right)
\\
&=
\textstyle
\symNumReg
\left(
\frac{
\symCardinalityBNorm^2
}{
\symCardinalityANorm - \symCardinalityBNorm \symJaccard}
+
\frac{
\symCardinalityANorm^2
}{
\symCardinalityBNorm - \symCardinalityANorm \symJaccard
}
+
\frac{
\symCardinalityANorm+\symCardinalityBNorm}{
\symJaccard}
\right)
=
\symNumReg
\frac{\symCardinalityANorm\symCardinalityBNorm(1-\symJaccard)}{\symJaccard(\symCardinalityANorm - \symCardinalityBNorm \symJaccard)(\symCardinalityBNorm - \symCardinalityANorm \symJaccard)}
\\
&=
\textstyle
\frac{\symNumReg}{\symJaccard(1-\symJaccard)}
\frac{1}{\frac{(\symCardinalityANorm - \symCardinalityBNorm \symJaccard)(\symCardinalityBNorm-\symCardinalityANorm\symJaccard)}{\symCardinalityANorm\symCardinalityBNorm(1-\symJaccard)^2}}
=
\frac{\symNumReg}{\symJaccard(1-\symJaccard)}
\frac{
1
}{
1
-
  \frac{
    (\symCardinalityANorm-\symCardinalityBNorm)^2\symJaccard
  }{
    \symCardinalityANorm\symCardinalityBNorm(1-\symJaccard)^2
  }
},
\end{align*}
where we used $\symCardinalityANorm+\symCardinalityBNorm=1$ multiple times.
\end{proof}

\begin{lemma}
\label{lem:poisson_approx}
If the cardinality $\symCardinality$ is not fixed, but follows a Poisson distribution with mean $\symPoissonRate$, hence $\symCardinality\sim\symPoisson(\symPoissonRate)$, nonzero register values of a \ac{GHLL} sketch using stochastic averaging will be distributed as the register values of a SetSketch with parameter $\symExponentialRate=1/\symNumReg$ representing a set with cardinality $\symPoissonRate$.
\end{lemma}
\begin{proof}
Stochastic averaging means that each distinct element is used for updating only a single register. Since the number of distinct elements is Poisson distributed with mean $\symPoissonRate$, the number of elements for updating one particular register is also Poisson distributed with mean $\symPoissonRate\symNumReg^{-1}$ where $\symNumReg$ is the number of registers.

For \ac{GHLL} the update value $\symRegValVariate_\text{update}$ is distributed as 
$\symRegValVariate_\text{update}\sim \lfloor 1 - \log_\symBase \symHashFunc_2(\symElement)\rfloor$ 
with
$\symHashFunc_2(\symElement)\sim\symUniform(0,1)$ (see \cref{sec:intro_hyperloglog}).
The corresponding cumulative distribution function is given by
$\symProbability(\symRegValVariate_\text{update}\leq \symRegVal) = \symProbability(\lfloor 1 - \log_\symBase \symHashFunc_2(\symElement)\rfloor \leq \symRegVal) = \symProbability(\symHashFunc_2(\symElement) > \symBase^{-\symRegVal}) = 1 - \symBase^{-\symRegVal}$
for $\symRegVal\geq 0$.
As a consequence, the cumulative distribution function of some register value $\symRegValVariate_\symIndexI$ that is updated with a frequency that is Poisson distributed with mean $\symPoissonRate\symNumReg^{-1}$ and probability mass $\symDensity(\symIndexJ) = \frac{e^{-\symPoissonRate\symNumReg^{-1}}(\symPoissonRate\symNumReg^{-1})^\symIndexJ}{\symIndexJ!}$ is given by 
\begin{align*}
&\symProbability(\symRegValVariate_\symIndexI\leq\symRegVal)
= 
\sum_{\symIndexJ=0}^\infty \symDensity(\symIndexJ) 
\symProbability(\symRegValVariate_\text{update}\leq \symRegVal)^\symIndexJ
\\
&=
\sum_{\symIndexJ=0}^\infty 
{\frac{e^{-\symPoissonRate\symNumReg^{-1}}(\symPoissonRate\symNumReg^{-1})^\symIndexJ}{\symIndexJ!}}(1 - \symBase^{-\symRegVal})^\symIndexJ
=
e^{-\symPoissonRate\symNumReg^{-1}}
\sum_{\symIndexJ=0}^\infty 
\frac{
(
\symPoissonRate\symNumReg^{-1}(1- \symBase^{-\symRegVal})
)^\symIndexJ}{\symIndexJ!}
\\
&=
e^{-\symPoissonRate\symNumReg^{-1}}
e^{\symPoissonRate\symNumReg^{-1}(1- \symBase^{-\symRegVal})}
=
e^{-\symPoissonRate\symNumReg^{-1}\symBase^{-\symRegVal}},
\end{align*}
where we used the identity $e^\symX =\sum_{\symIndexJ=0}^\infty \frac{\symX^\symIndexJ}{\symIndexJ!}$.
Comparing $\symProbability(\symRegValVariate_\symIndexI\leq\symRegVal)=e^{-\symPoissonRate\symNumReg^{-1}\symBase^{-\symRegVal}}$ with \eqref{equ:set_sketch_distribution} completes the proof.
\end{proof}

\begin{lemma}
\label{lem:depoissonization}
Assume the distribution of some discrete statistic $\symStatisticX$ depends on some parameter $\symCardinality$, which is Poisson distributed with mean $\symPoissonRate$, hence $\symCardinality\sim\symPoisson(\symPoissonRate)$. Furthermore, assume that $\symStatisticX$ is an unbiased estimator for $\symPoissonRate$ for all $\symPoissonRate > 0$, hence $\symExpectation(\symStatisticX\vert \symPoissonRate) = \symPoissonRate$. Then, $\symStatisticX$ is also an unbiased estimator for $\symCardinality$, if $\symStatisticX$ is conditioned on $\symCardinality$, hence $\symExpectation(\symStatisticX\vert \symCardinality) = \symCardinality$ for all $\symCardinality\geq 0$.
\end{lemma}

\begin{proof}
\begin{align*}
\textstyle 0 &= \symPoissonRate - \symExpectation(\symStatisticX\vert \symPoissonRate)
=
\textstyle
\symPoissonRate - \sum_{\symCardinality=0}^\infty \frac{e^{-\symPoissonRate}\symPoissonRate^\symCardinality}{\symCardinality!} \symExpectation(\symStatisticX\vert \symCardinality)
\\
&=
\textstyle  
e^{-\symPoissonRate}\sum_{\symCardinality=0}^\infty \frac{\symPoissonRate^\symCardinality}{\symCardinality!} \symCardinality - \sum_{\symCardinality=0}^\infty \frac{e^{-\symPoissonRate}\symPoissonRate^\symCardinality}{\symCardinality!} \symExpectation(\symStatisticX\vert \symCardinality)
\\
&=
\textstyle
e^{-\symPoissonRate} \sum_{\symCardinality=0}^\infty \frac{\symPoissonRate^\symCardinality}{\symCardinality!} (\symCardinality - \symExpectation(\symStatisticX\vert \symCardinality)).
\end{align*}
Therefore, $\sum_{\symCardinality=0}^\infty \frac{\symCardinality - \symExpectation(\symStatisticX\vert \symCardinality)}{\symCardinality!}\symPoissonRate^\symCardinality  = 0$. This power series with respect to $\symPoissonRate$ is zero for all $\symPoissonRate>0$ only if all coefficients are zero, which implies $\symExpectation(\symStatisticX\vert \symCardinality) = \symCardinality$ and shows that $\symStatisticX$ is an unbiased estimator for $\symCardinality$. 
\end{proof}

\section{Corrected Cardinality Estimator}
\label{sec:corrected_cardinality_estimator}
The derivation of cardinality estimator \eqref{equ:raw_cardinality_estimator} assumed that register values are distributed according to \eqref{equ:set_sketch_distribution}.
In practice, however, the possible register values are limited to a range $\lbrace 0,1,2,\ldots,\symMaxRegularValue, \symMaxRegularValue+1\rbrace$ with $\symMaxRegularValue\geq 0$. So the assumption is valid only if the register values are concentrated within this range, which is the case for properly configured SetSketches (see \cref{sec:parameter_config}). 
Otherwise, it is necessary to incorporate that register values are actually distributed as 
\begin{equation}
\label{equ:transformation}
\symRegValVariate_\symIndexI\sim\max(0,\min(\symMaxRegularValue+1, \symRegValVariate^{*}_\symIndexI))
\end{equation}
where $\symRegValVariate^{*}_\symIndexI$ is distributed according to \eqref{equ:set_sketch_distribution}, hence
\begin{equation}
\label{equ:distx}
\symProbability(\symRegValVariate^{*}_\symIndexI \leq \symRegVal) = e^{-\symCardinality\symExponentialRate \symBase^{-\symRegVal}}.
\end{equation}
Obviously, knowing $\symRegValVariate^{*}_\symIndexI$ would still allow using the cardinality estimator \eqref{equ:raw_cardinality_estimator}
\begin{equation}
\label{equ:card_est}
\textstyle\symCardinalityCorrectedEstimate
=
\frac{\symNumReg(1-1/\symBase)}{\symExponentialRate\log(\symBase)\sum_{\symIndexI=1}^\symNumReg  \symBase^{-\symRegValVariate^{*}_\symIndexI}}
=
\frac{\symNumReg(1-1/\symBase)}{\symExponentialRate\log(\symBase)\sum_{\symRegVal=-\infty}^{\infty} \symCountVariate^{*}_\symRegVal \symBase^{-\symRegVal}},
\end{equation}
where $\symCountVariate^{*}_\symRegVal := |\lbrace \symIndexI : \symRegValVariate^{*}_\symIndexI = \symRegVal\rbrace|$ is the histogram of values $ \symRegValVariate^{*}_\symIndexI$. Hence, knowing the histogram $\symCountVariate^{*}_\symRegVal$ would be sufficient. Due to \eqref{equ:transformation}, this histogram relates to the observed histogram  $\symCountVariate_\symRegVal := |\lbrace \symIndexI : \symRegValVariate_\symIndexI = \symRegVal\rbrace|$ as
\begin{equation}
\label{equ:histogram_relation}
\symCountVariate_\symRegVal = 
\begin{cases}
\sum_{\symIndexL=-\infty}^0 \symCountVariate^{*}_\symIndexL & \symRegVal = 0,\\
\symCountVariate^{*}_\symRegVal & 1 \leq \symRegVal \leq \symMaxRegularValue,\\
\sum_{\symIndexL=\symMaxRegularValue + 1}^\infty \symCountVariate^{*}_\symIndexL & \symRegVal = \symMaxRegularValue + 1.
\end{cases}
\end{equation}
Therefore, $\symCountVariate^{*}_\symRegVal$ is known for $1\leq \symRegVal\leq \symMaxRegularValue$ and we just need to find a way to estimate $\symCountVariate^{*}_\symRegVal$ for all $\symRegVal\leq 0$ and $\symRegVal\geq \symMaxRegularValue+1$ so that we can use \eqref{equ:card_est}. 

For $\symRegVal\leq 0$, the identity 
\begin{equation*}
\symProbability(\symRegValVariate_\symIndexI^{*} \leq \symRegVal)
 =
 (
 \symProbability(\symRegValVariate_\symIndexI^{*} \leq 0)
 )^{\symBase^{-\symRegVal}},
 \end{equation*}
which directly follows from \eqref{equ:distx}, suggests that the estimates $\symCountVariateEstimate^{*}_\symRegVal$ for $\symCountVariate^{*}_\symRegVal$ should satisfy
\begin{equation*}
\textstyle\sum_{\symIndexL=-\infty}^\symRegVal \symCountVariateEstimate^{*}_\symIndexL/\symNumReg
=
\left(\sum_{\symIndexL=-\infty}^0 \symCountVariateEstimate^{*}_\symIndexL/\symNumReg\right)^{\symBase^{-\symRegVal}}.
\end{equation*}
Furthermore, \eqref{equ:histogram_relation} implies $\sum_{\symIndexL=-\infty}^0 \symCountVariateEstimate^{*}_\symIndexL = \symCountVariate_0$, and therefore we require
\begin{equation*}
\textstyle\sum_{\symIndexL=-\infty}^\symRegVal \symCountVariateEstimate^{*}_\symIndexL/\symNumReg
=
\left(\symCountVariate_{0}/\symNumReg\right)^{\symBase^{-\symRegVal}}.
\end{equation*}
This system of equations with $\symRegVal\leq 0$ has the solution 
\begin{equation*}
\symCountVariateEstimate^{*}_\symRegVal = \symNumReg\left(\left(\symCountVariate_0/\symNumReg\right)^{\symBase^{-\symRegVal}}-\left(\symCountVariate_0/\symNumReg\right)^{\symBase^{1-\symRegVal}}\right).
\end{equation*}

Similarly, for $\symRegVal\geq \symMaxRegularValue + 1$, the identity 
\begin{equation*}
\symProbability(\symRegValVariate_\symIndexI^{*} \geq \symRegVal)
=
1
-
(
1
-
\symProbability(\symRegValVariate_\symIndexI^{*} \geq \symMaxRegularValue + 1)
)^{\symBase^{\symMaxRegularValue+1-\symRegVal}},
\end{equation*}
which also follows from \eqref{equ:distx}, leads to
\begin{equation*}
\textstyle
\sum_{\symIndexL=\symRegVal}^\infty \symCountVariateEstimate^{*}_\symIndexL/\symNumReg
=
1-
\left(
1
-
\sum_{\symIndexL=\symMaxRegularValue+1}^\infty \symCountVariateEstimate^{*}_\symIndexL/\symNumReg
\right)^{\symBase^{\symMaxRegularValue+1-\symRegVal}}
=
1-
(
1
-
\symCountVariate_{\symMaxRegularValue + 1}/\symNumReg
)^{\symBase^{\symMaxRegularValue+1-\symRegVal}}.
\end{equation*}
This system of equations with $\symRegVal\geq \symMaxRegularValue+1$ has the solution
\begin{equation*}
\symCountVariateEstimate^{*}_\symRegVal = 
\symNumReg\left(\left(1 - \symCountVariate_{\symMaxRegularValue+1}/\symNumReg\right)^{\symBase^{\symMaxRegularValue - \symRegVal}}-\left(1 - \symCountVariate_{\symMaxRegularValue+1}/\symNumReg\right)^{\symBase^{1+\symMaxRegularValue-\symRegVal}}\right).
\end{equation*}

Complemented by the trivial estimator for the case $1\leq \symRegVal\leq \symMaxRegularValue$ $\symCountVariateEstimate^{*}_\symRegVal=\symCountVariate^{*}_\symRegVal = \symCountVariate_\symRegVal$, the estimator for the whole histogram can be written as
\begin{equation*}
\symCountVariateEstimate^{*}_\symRegVal =
\begin{cases}
\symNumReg\left(\left(\symCountVariate_0/\symNumReg\right)^{\symBase^{-\symRegVal}}-\left(\symCountVariate_0/\symNumReg\right)^{\symBase^{1-\symRegVal}}\right) & \symRegVal \leq 0,
\\
\symCountVariate_\symRegVal & 1\leq \symRegVal \leq \symMaxRegularValue, 
\\
\symNumReg\left(\left(1 - \symCountVariate_{\symMaxRegularValue+1}/\symNumReg\right)^{\symBase^{\symMaxRegularValue - \symRegVal}}-\left(1 - \symCountVariate_{\symMaxRegularValue+1}/\symNumReg\right)^{\symBase^{1+\symMaxRegularValue-\symRegVal}}\right)
& \symRegVal\geq \symMaxRegularValue+1.
\end{cases}
\end{equation*}
Using these estimates instead of $\symCountVariate^{*}_\symRegVal$ in \eqref{equ:card_est} gives 
\begin{equation*}
\symCardinalityCorrectedEstimate = \textstyle\frac{\symNumReg(1-1/\symBase)}{
\symExponentialRate \log(\symBase)
{
\scriptstyle
\left(
\symNumReg\symSmallCorrectionFunc_\symBase(\symCountVariate_0/\symNumReg)
+
(\sum_{\symRegVal=1}^\symMaxRegularValue \symCountVariate_\symRegVal \symBase^{-\symRegVal})
+
\symNumReg
\symBase^{-\symMaxRegularValue}
\symLargeCorrectionFunc_\symBase(1-\symCountVariate_{\symMaxRegularValue+1}/\symNumReg)
\right)}},
\end{equation*}
where 
\begin{align*}
\symSmallCorrectionFunc_\symBase(\symX)
&:=
\textstyle
\sum_{\symRegVal=-\infty}^0
\symBase^{-\symRegVal}\left(\symX^{\symBase^{-\symRegVal}}
-
\symX^{\symBase^{1-\symRegVal}}\right)
=
\sum_{\symRegVal=1}^\infty
\symBase^{\symRegVal-1}\left(\symX^{\symBase^{\symRegVal-1}}
-
\symX^{\symBase^{\symRegVal}}\right)
\\
&=
\textstyle
\sum_{\symRegVal=1}^\infty
\symBase^{\symRegVal-1}\symX^{\symBase^{\symRegVal-1}}
-
\symBase^{\symRegVal-1}\symX^{\symBase^{\symRegVal}}
=
\symX
+
\sum_{\symRegVal=1}^\infty
\symBase^{\symRegVal}\symX^{\symBase^{\symRegVal}}
-
\symBase^{\symRegVal-1}\symX^{\symBase^{\symRegVal}}
\\
&=
\textstyle
\symX
+
(\symBase-1)
\sum_{\symRegVal=1}^\infty
\symBase^{\symRegVal-1}\symX^{\symBase^{\symRegVal}}
\end{align*}
and
\begin{align*}
\symLargeCorrectionFunc_\symBase(\symX)
&:=
\textstyle
\sum_{\symRegVal=\symMaxRegularValue+1}^\infty \symBase^{\symMaxRegularValue-\symRegVal} \left(\symX^{\symBase^{\symMaxRegularValue-\symRegVal}} - \symX^{\symBase^{1 + \symMaxRegularValue-\symRegVal}}\right)
\\
&=
\textstyle
\sum_{\symRegVal=0}^\infty \symBase^{-\symRegVal-1} \left(\symX^{\symBase^{-1-\symRegVal}} - \symX^{\symBase^{-\symRegVal}}\right)
\\
&=
\textstyle
\sum_{\symRegVal=0}^\infty \symBase^{-\symRegVal-1} \symX^{\symBase^{-1-\symRegVal}} - \symBase^{-\symRegVal-1} \symX^{\symBase^{-\symRegVal}}
\\
&=
\textstyle
-\symX+
\sum_{\symRegVal=0}^\infty \symBase^{-\symRegVal} \symX^{\symBase^{-\symRegVal}} - \symBase^{-\symRegVal-1} \symX^{\symBase^{-\symRegVal}}
\\
&=
\textstyle
-\symX+
(\symBase-1)\sum_{\symRegVal=0}^\infty \symBase^{-\symRegVal-1} \symX^{\symBase^{-\symRegVal}}
\\
&=
\textstyle
1-\symX
+
(\symBase-1)
\sum_{\symRegVal=0}^\infty \symBase^{-\symRegVal-1}(\symX^{\symBase^{-\symRegVal}}-1).
\end{align*}
We used the identity $(\symBase-1)\sum_{\symRegVal=0}^\infty \symBase^{-\symRegVal-1} = 1$ for the last transformation, which improves the convergence, because $\symX^{\symBase^{-\symRegVal}}\rightarrow 1$ as $\symRegVal\rightarrow\infty$ for $\symX>0$.

\section{More Experimental Results}

\begin{figure}[H]
  \centering
  \includegraphics[width=\columnwidth]{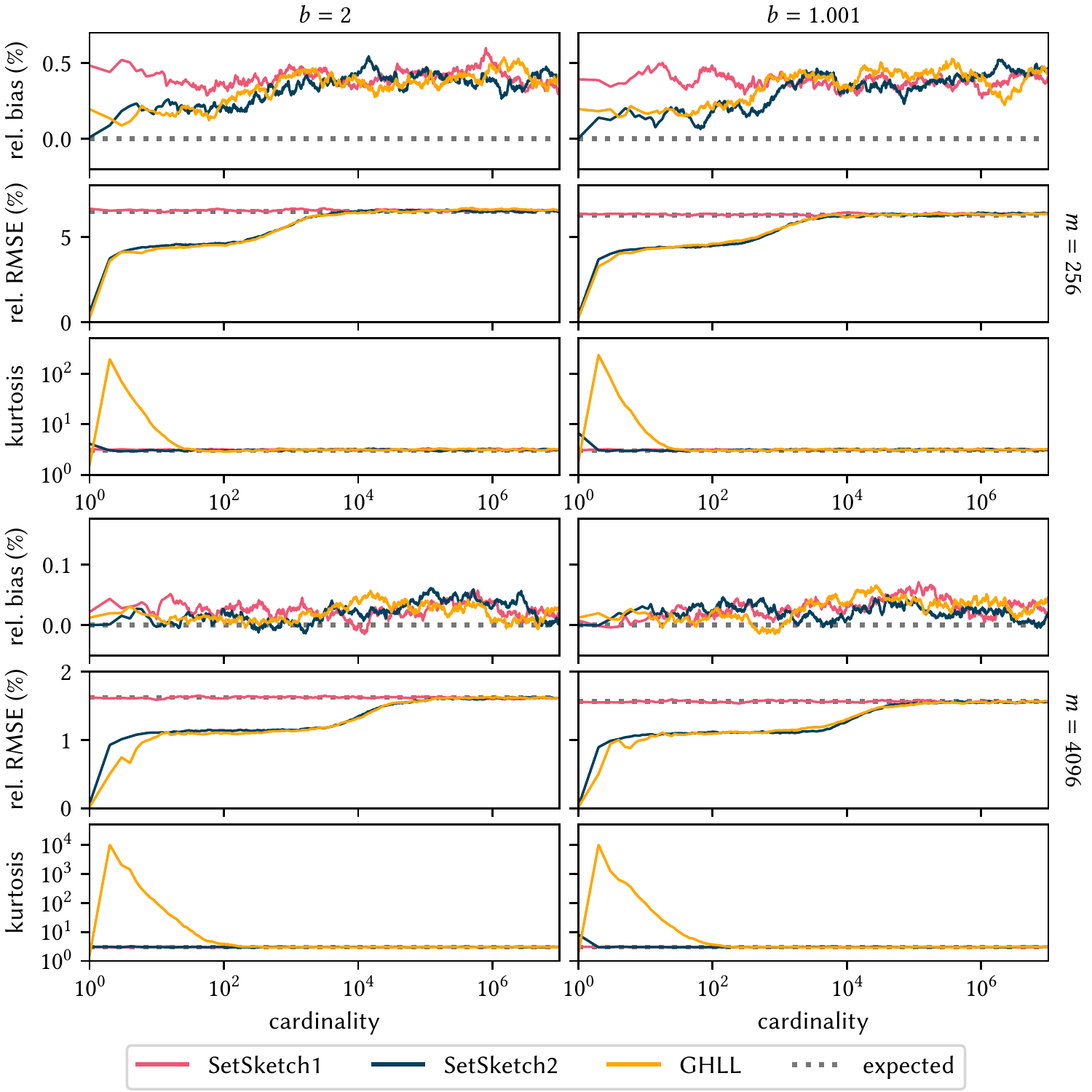}
  \caption{\boldmath The relative bias, the relative \acs*{RMSE}, and the kurtosis for SetSketch1, SetSketch2, and \acs*{GHLL} when using \acs*{ML} estimation based on \eqref{equ:set_sketch_distribution}.}
  \label{fig:cardinality_ml}
\end{figure}

\begin{figure*}[h]
\centering
\includegraphics[width=\textwidth]{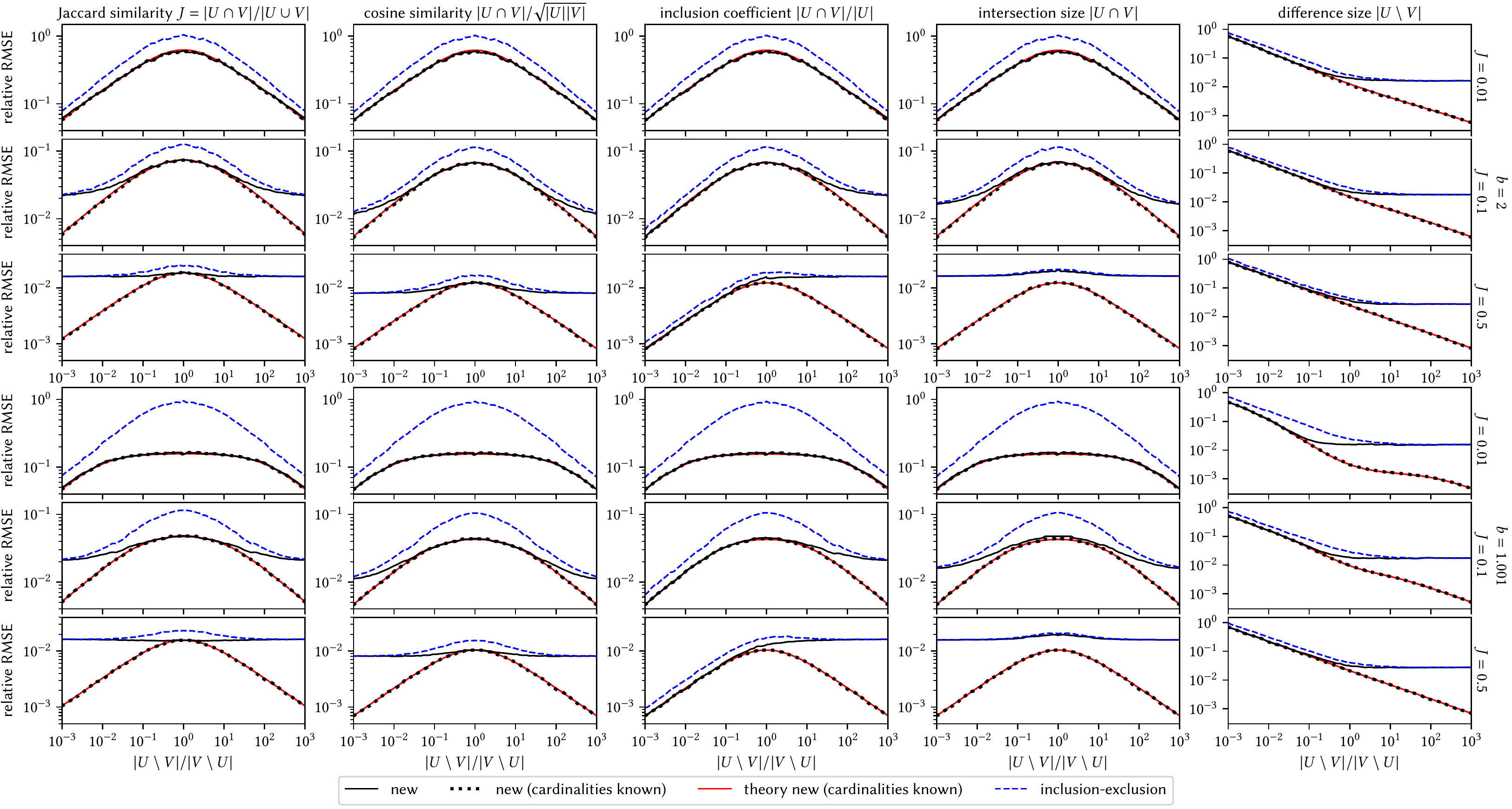}
\caption{\boldmath The relative \acs*{RMSE} of various estimated joint quantities when using SetSketch2 with $\symBase\in\lbrace1.001, 2\rbrace$ and $\symNumReg=4096$ for sets with a fixed union cardinality of $|\symSetA\cup\symSetB|=10^6$.}
\label{fig:joint_set_sketch2_1000000}
\end{figure*}

\begin{figure*}[h]
\centering
\includegraphics[width=\textwidth]{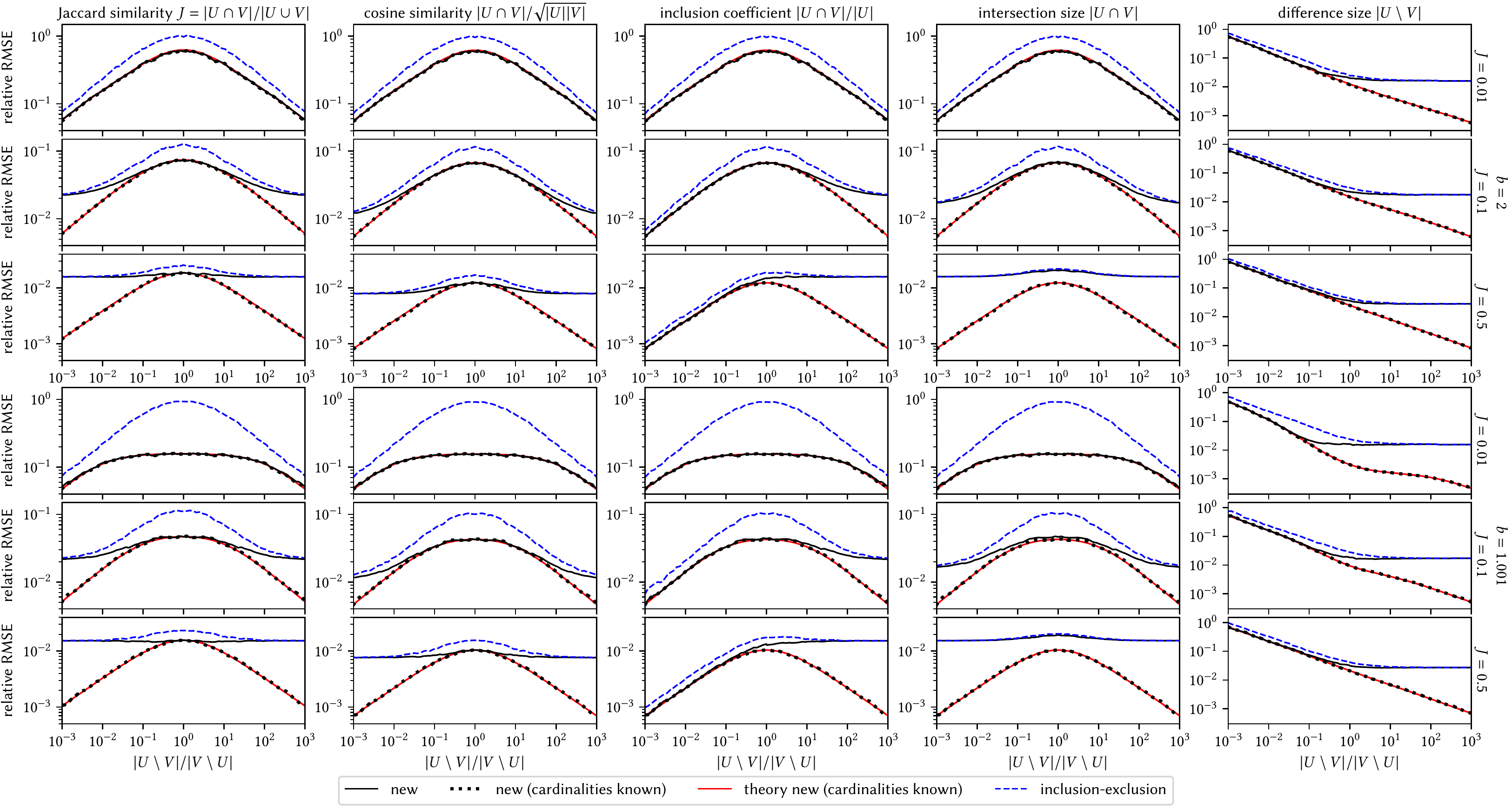}
\caption{\boldmath The relative \acs*{RMSE} of various estimated joint quantities when using \acs*{GHLL} with $\symBase\in\lbrace1.001, 2\rbrace$ and $\symNumReg=4096$ for sets with a fixed union cardinality of $|\symSetA\cup\symSetB|=10^6$.}
\label{fig:joint_ghll_1000000}
\end{figure*}

\begin{figure*}[h]
\centering
\includegraphics[width=\textwidth]{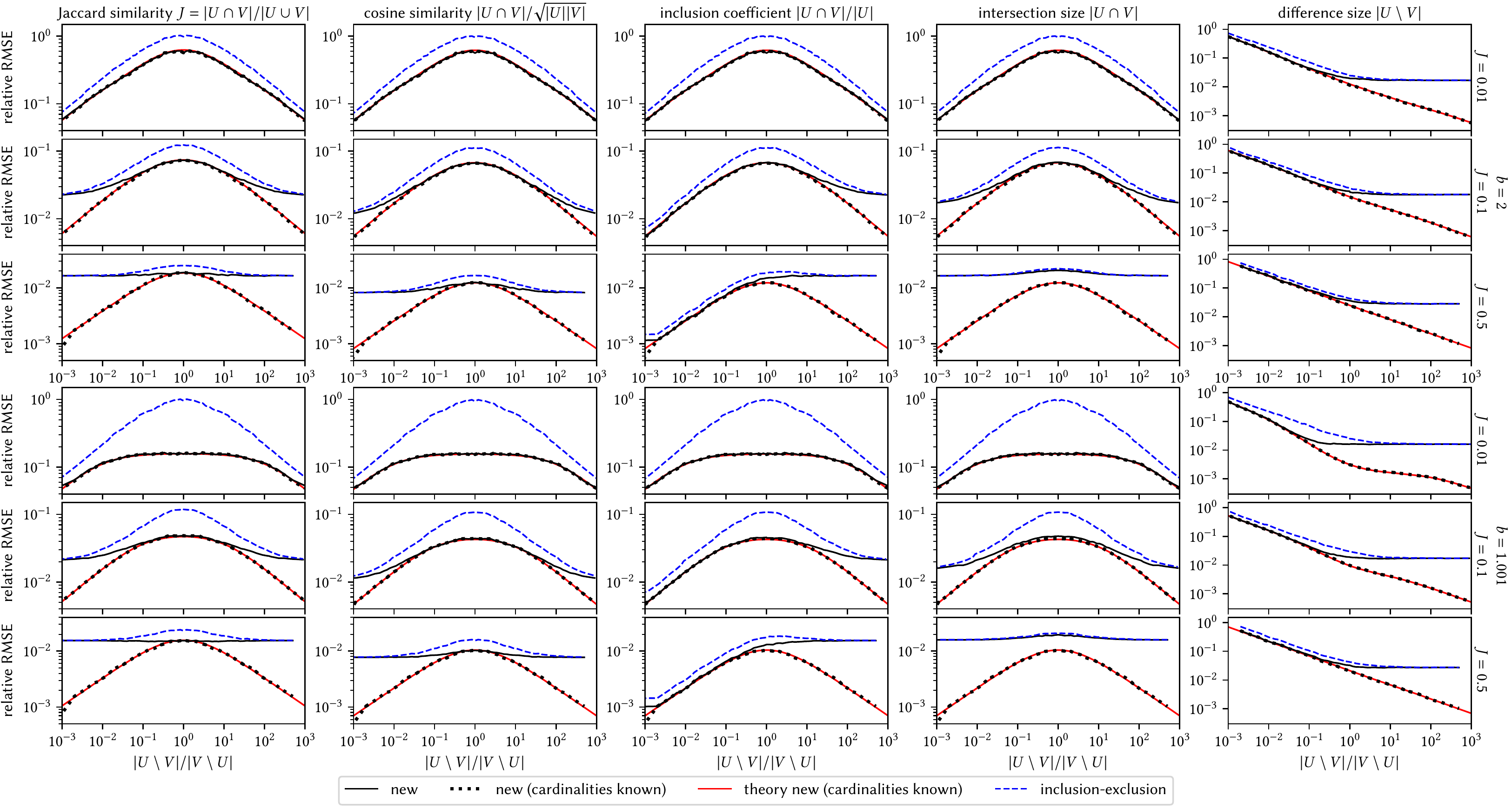}
\caption{\boldmath The relative \acs*{RMSE} of various estimated joint quantities when using SetSketch1 with $\symBase\in\lbrace1.001, 2\rbrace$ and $\symNumReg=4096$ for sets with a fixed union cardinality of $|\symSetA\cup\symSetB|=10^3$.}
\label{fig:joint_set_sketch1_1000}
\end{figure*}

\begin{figure*}[h]
\centering
\includegraphics[width=\textwidth]{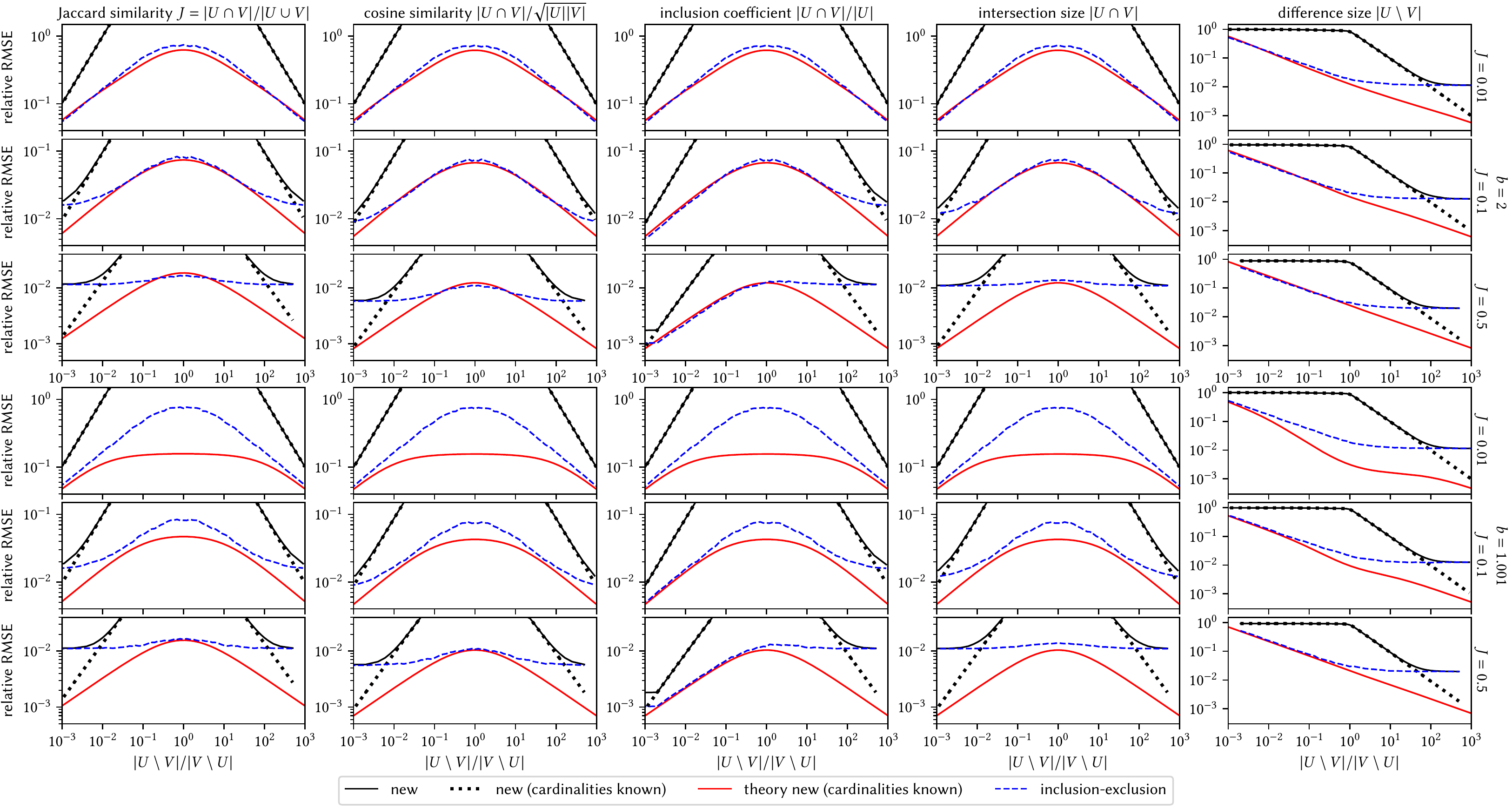}
\caption{\boldmath The relative \acs*{RMSE} of various estimated joint quantities when using \acs*{GHLL} with $\symBase\in\lbrace1.001, 2\rbrace$ and $\symNumReg=4096$ for sets with a fixed union cardinality of $|\symSetA\cup\symSetB|=10^3$.}
\label{fig:joint_ghll_1000}
\end{figure*}

\begin{figure*}[h]
\centering
\includegraphics[width=\textwidth]{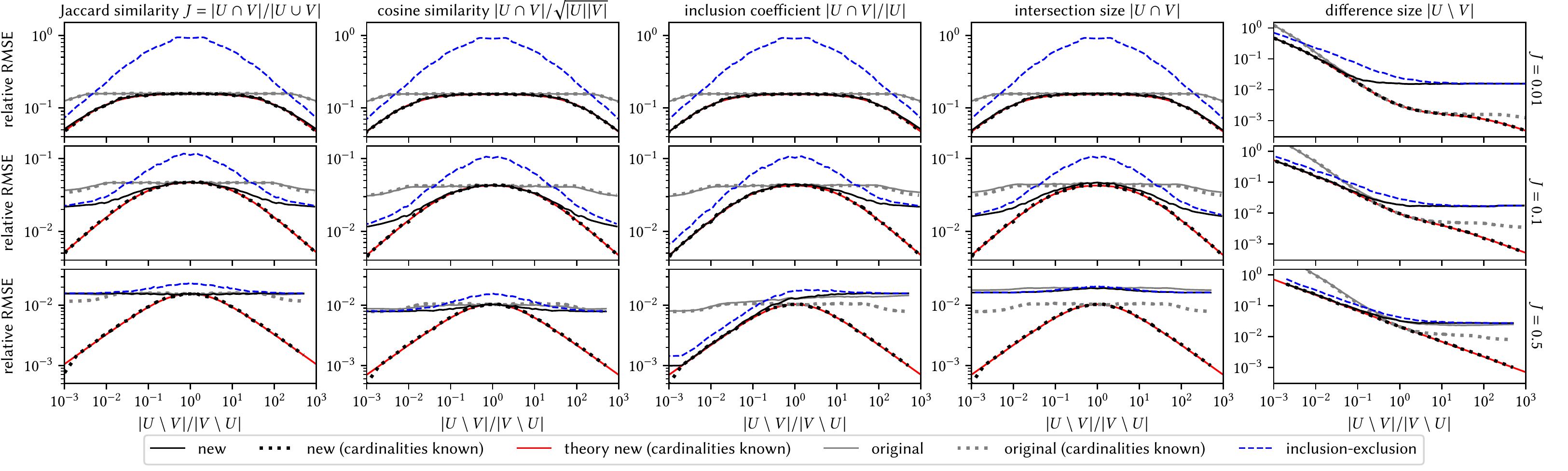}
\caption{\boldmath The relative \acs*{RMSE} of various estimated joint quantities when using \acs*{MH} with $\symNumReg=4096$ for sets with a fixed union cardinality of $|\symSetA\cup\symSetB|=10^3$.}
\label{fig:joint_minhash_1000}
\end{figure*}

\begin{figure*}[h]
\centering
\includegraphics[width=\textwidth]{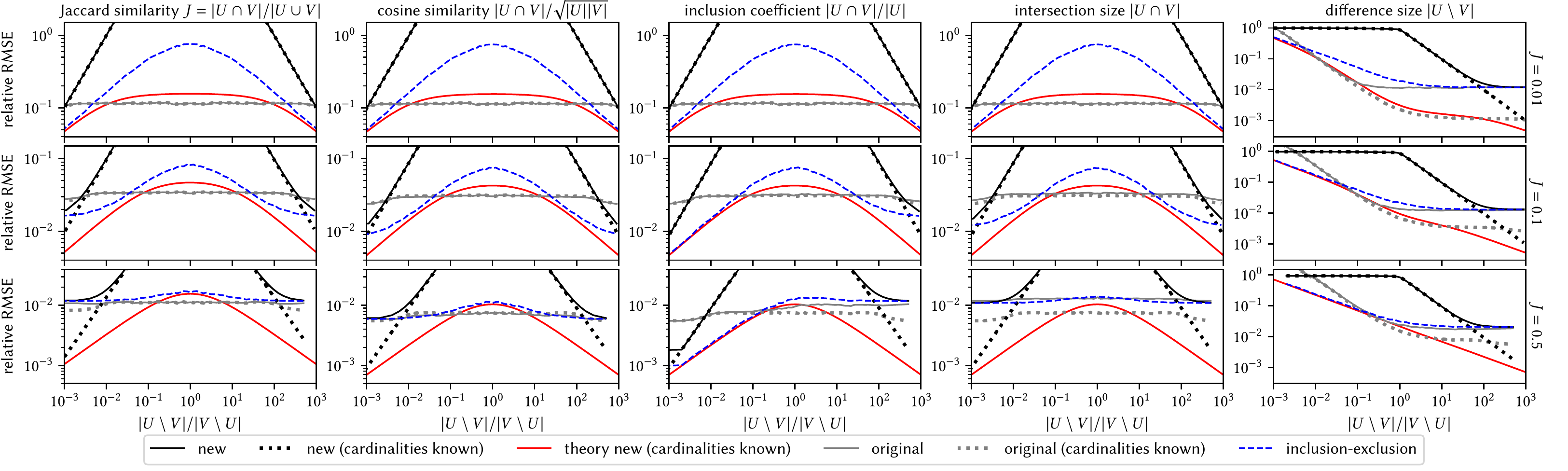}
\caption{\boldmath The relative \acs*{RMSE} of various estimated joint quantities when using HyperMinHash with $\symNumReg=4096$ and $\symHyperMinHashParameter=10$, which corresponds to $\symBase=2^{-2^{10}}\approx 1.000677$, for sets with a fixed union cardinality of $|\symSetA\cup\symSetB|=10^3$.}
\label{fig:joint_hyperminhash_1000}
\end{figure*}

\fi
\end{document}